\documentclass[a4paper,USenglish,cleveref,autoref,thm-restate]{lipics-v2021}

\pdfoutput=1 
\hideLIPIcs  

\title{The Algebra of Iterative Constructions} 

\author{Kevin {Batz}\footnote{This is a team effort. Authors are listed alphabetically.}}{Cornell University, Ithaca, USA}{ksb239@cornell.edu}{https://orcid.org/0000-0001-8705-2564}{ERC grant Autoprobe (no.\ 101002697)}

\author{Benjamin Lucien {Kaminski}}{Saarland University, Saarland Informatics Campus, Saarbrücken, Germany \and University College London, United Kingdom \and \url{https://quave.cs.uni-saarland.de/benjamin-kaminski/}}{kaminski@cs.uni-saarland.de}{https://orcid.org/0000-0001-5185-2324}{}

\author{Lucas {Kehrer}}{Saarland University, Saarland Informatics Campus, Saarbrücken, Germany \and \url{https://quave.cs.uni-saarland.de/members/lucas-kehrer/}}{kehrer@cs.uni-saarland.de}{https://orcid.org/0009-0006-9293-8742}{Deutsche Forschungsgemeinschaft (DFG, German Research Foundation) -- project number 563038229}

\author{Gerwin {Klein}}{Proofcraft, Sydney, Australia \and University of New South Wales, Sydney, Australia\and \url{https://doclsf.de/}}{gerwin.klein@proofcraft.systems}{https://orcid.org/0000-0001-8883-0559}{}

\author{Todd {Schmid}}{Bucknell University, Lewisburg, Pennsylvania, USA\and \url{https://www.eg.bucknell.edu/~tjws001/toddwayneschmid/index.php}}{t.schmid@bucknell.edu}{https://orcid.org/0000-0002-3265-7168}{}

\author{Henning {Urbat}}{Friedrich-Alexander-Universität Erlangen-Nürnberg, Germany\and \url{https://www8.cs.fau.de/people/henning-urbat/}}{henning.urbat@fau.de}{https://orcid.org/0000-0002-3265-7168}{Deutsche Forschungsgemeinschaft (DFG, German Research Foundation) -- project number 569130867}

\authorrunning{K. Batz, B. L. Kaminski, L. Kehrer, G. Klein, T. Schmid, H. Urbat} 

\Copyright{Kevin Batz, Benjamin L. Kaminski, Lucas Kehrer, Gerwin Klein, Todd Schmid, Henning Urbat} 



\EventEditors{John Q. Open and Joan R. Access}
\EventNoEds{2}
\EventLongTitle{42nd Conference on Very Important Topics (CVIT 2016)}
\EventShortTitle{CVIT 2016}
\EventAcronym{CVIT}
\EventYear{2016}
\EventDate{December 24--27, 2016}
\EventLocation{Little Whinging, United Kingdom}
\EventLogo{}
\SeriesVolume{42}
\ArticleNo{23}

\usepackage{booktabs}   
\usepackage{adjustbox}
\usepackage{amsmath}
\usepackage{pbox}
\usepackage{mathtools}
\usepackage{csquotes}
\usepackage{stmaryrd}
\usepackage{xspace}
\usepackage{tikz-cd}
\tikzstyle{seqelem}=[fill=none, draw=none, shape=circle, font=\Large]
\tikzstyle{blue edge}=[thick,fill=none, draw={rgb,255: red,55; green,114; blue,199}, ->]
\tikzstyle{blue double edge}=[thick,fill=none, draw={rgb,255: red,55; green,114; blue,199}, <->]
\tikzstyle{orange edge}=[thick,draw={rgb,255: red,239; green,134; blue,51}, ->]
\tikzstyle{aelem}=[fill=none, draw=none, shape=circle, font={\LARGE}, xshift=-0.125,yshift=-0.25]
\tikzstyle{majelem}=[fill=white, font={\Huge}]
\tikzstyle{minelem}=[fill=white, font={\Huge}]
\usepackage{wrapfig}
\usepackage{mathpartir}
\usepackage{mdframed}
\usepackage{tensor}
\usepackage{tikz}
\usetikzlibrary{patterns,decorations.pathreplacing,arrows,arrows.meta,decorations.pathmorphing,positioning,fit,trees,shapes,shadows,automata,calc}
\usepackage{pgfplots}
\usepackage[textwidth=13mm]{todonotes}      
\usepackage{fancybox}

\definecolor{DodgerBlue3}{rgb}{0.28, 0.51, 1.0}
\definecolor{DodgerBlue}{rgb}{0.12, 0.56, 1.0}
\definecolor{DodgerBlue4}{rgb}{0.06, 0.30, 0.54}  

\newmdenv[linecolor=brown, backgroundcolor=brown!6, skipabove=5pt, skipbelow=5pt, leftmargin=0pt, rightmargin=0pt]{axbox}
\newmdenv[linecolor=webgreen, backgroundcolor=webgreen!6, skipabove=5pt, skipbelow=5pt, leftmargin=0pt, rightmargin=0pt]{infbox} 
\newmdenv[linecolor=black, backgroundcolor=white, skipabove=5pt, skipbelow=5pt, leftmargin=0pt, rightmargin=0pt]{whitebox}

\theoremstyle{claimstyle}
\newtheorem{notation}[definition]{Notation}

\renewcommand{\arraystretch}{1.25}

\newcommand{\setfont}[1]{\ensuremath{\mathsf{#1}}}

\newcommand{\appref}[1]{\hyperref[#1]{Appendix~\ref{#1}}}

\newcommand{\Alg}{{\mathcal A}}
\newcommand{\AlgB}{{\mathcal B}}
\newcommand{\AlgTerm}{{\mathcal T}}
\newcommand{\AIC}{{\setfont{AIC}}}
\newcommand{\Mon}{{\setfont{MonoMaps}}}
\newcommand{\interpretsem}[1]{{^\interpret\llbracket #1 \rrbracket}}

\newcommand{\normalrules}{%
\normalsize%
\def \TirNameStyle ##1{\small \textsc{##1}}%
}
\newcommand{\smallrules}{%
\small%
\def \TirNameStyle ##1{\footnotesize \textsc{##1}}%
}
\newcommand{\footnotesizerules}{%
\footnotesize%
\def \TirNameStyle ##1{\scriptsize \textsc{##1}}%
}
\newcommand{\scriptsizerules}{%
\scriptsize%
\def \TirNameStyle ##1{\tiny \textsc{##1}}%
}

\newcommand{\wf}{\mathsf{wf}} 

\newcommand{\sequence}[4]{#1 ~ #2 ~ #3 ~ #4 \ldots{}}
\newcommand{\mmorespacesequence}[4]{#1 ~{}~{}~{} #2 ~{}~{}~{} #3 ~{}~{}~{} #4 \ldots{}}
\newcommand{\morespacesequence}[4]{#1 ~{}~{} #2 ~{}~{} #3 ~{}~{} #4 \ldots{}}
\newcommand{\constseq}[1]{\sequence{#1}{#1}{#1}{#1}}

\newcommand{\leftsuperscript}[2]{\tensor*[^{#1}]{#2}{}}

\newcommand{\Term}{\setfont{Terms}\xspace}

\newcommand{\Vars}{\setfont{Vars}}
\newcommand{\Funcs}{\setfont{Funcs}}

\newcommand{\fcta}{\ensuremath{F}}
\newcommand{\fctb}{\ensuremath{G}}
\newcommand{\fctc}{\ensuremath{H}}
\newcommand{\funca}{\fcta}
\newcommand{\funcb}{\fctb}
\newcommand{\funcc}{\fctc}
\newcommand{\funcaAlg}{\fcta^{\semigray{\Alg}}}

\newcommand{\funckind}{\funcb}

\DeclareFontFamily{U}{MnSymbolC}{}

\DeclareFontShape{U}{MnSymbolC}{m}{n}{
  <-6>   MnSymbolC5
  <6-7>  MnSymbolC6
  <7-8>  MnSymbolC7
  <8-9>  MnSymbolC8
  <9-10> MnSymbolC9
  <10-12> MnSymbolC10
  <12->  MnSymbolC12}{}

\DeclareFontShape{U}{MnSymbolC}{b}{n}{
  <-6>   MnSymbolC-Bold5
  <6-7>  MnSymbolC-Bold6
  <7-8>  MnSymbolC-Bold7
  <8-9>  MnSymbolC-Bold8
  <9-10> MnSymbolC-Bold9
  <10-12> MnSymbolC-Bold10
  <12->  MnSymbolC-Bold12}{}

\DeclareSymbolFont{MnSyC}{U}{MnSymbolC}{m}{n}
\SetSymbolFont{MnSyC}{bold}{U}{MnSymbolC}{b}{n}

\DeclareMathSymbol{\meddiamond}{\mathbin}{MnSyC}{110}
\DeclareMathSymbol{\Diamond}{\mathbin}{MnSyC}{110}
\newcommand{\Supsymbol}{{\raisebox{-0.00em}{$\Diamond$}}}

\newcommand{\Sup}{\Supsymbol\,}

\newcommand{\Infsymbol}{{{\raisebox{-0.09em}{$\square$}}}}

\newcommand{\Inf}{\Infsymbol\,}

\DeclareSymbolFont{txsymbolsC}{U}{txsyc}{m}{n}
\SetSymbolFont{txsymbolsC}{bold}{U}{txsyc}{bx}{n}
\DeclareFontSubstitution{U}{txsyc}{m}{n}
\DeclareMathSymbol{\medcirc}{\mathbin}{txsymbolsC}{7}
\newcommand{\Nextsymbol}{{\medcirc}}
\newcommand{\Next}{\Nextsymbol\,}
\newcommand{\Firstsymbol}{{\circleddash}}
\newcommand{\First}{\Firstsymbol\,}

\newcommand{\lmin}{\curlywedge}
\newcommand{\llmin}{\morespace{\lmin}}
\newcommand{\lmax}{\curlyvee}
\newcommand{\llmax}{\morespace{\lmax}}

\newcommand{\Fa}{\funca\,}

\newcommand{\Fkinditersymbol}[1]{\funckind^{#1}_{\forb}}
\newcommand{\Fkinditer}[1]{\Fkinditersymbol{#1}\,}
\newcommand{\Fas}{\fcta^*}

\newcommand{\F}{\fcta\,}
\newcommand{\Fs}{\fcta^*}

\newcommand{\botAlg}{\bot\!{}^{\semigray{\Alg}}}
\newcommand{\topAlg}{\top^{\semigray{\Alg}}}
\newcommand{\botL}{\bot\!{}^{\semigray{L}}}
\newcommand{\topL}{\top^{\semigray{L}}}
\newcommand{\lmaxAlg}{\mathrel{{\lmax}^{\semigray{\Alg}}\!}}
\newcommand{\lminAlg}{\mathrel{{\lmin}\!{}^{\semigray{\Alg}}\!}}
\newcommand{\SupAlgsymbol}{{\Supsymbol}^{\semigray{\Alg}}}
\newcommand{\SupAlg}{\SupAlgsymbol\,}
\newcommand{\InfAlgsymbol}{\Infsymbol^{\semigray{\Alg}}}

\newcommand{\NextAlgsymbol}{\Nextsymbol^{\semigray{\Alg}}}
\newcommand{\NextAlg}{\NextAlgsymbol\,}
\newcommand{\FirstAlgsymbol}{\Firstsymbol^{\semigray{\Alg}}}

\newcommand{\FsAlg}{{\fcta^*}^{\semigray{\Alg}}}

\newcommand{\botAlgB}{\bot\!{}^{\semigray{\AlgB}}}
\newcommand{\topAlgB}{\top^{\semigray{\AlgB}}}
\newcommand{\lmaxAlgB}{\mathrel{{\lmax}^{\semigray{\AlgB}}\!}}

\newcommand{\SupAlgBsymbol}{{\Supsymbol}^{\semigray{\AlgB}}}
\newcommand{\SupAlgB}{\SupAlgBsymbol\,}

\newcommand{\ladomain}{\ensuremath{L}}
\newcommand{\laelema}{\ensuremath{\ell}}
\newcommand{\laelemb}{\ensuremath{m}}
\newcommand{\laelem}{\laelema}

\newcommand{\laseqs}{\ensuremath{L^\omega}}
\newcommand{\lasubset}{\ensuremath{S}}
\newcommand{\laord}{\sqsubseteq}
\newcommand{\larevord}{\sqsupseteq}
\newcommand{\llaord}{\morespace{\laord}}

\newcommand{\lameet}{\sqcap}

\newcommand{\seqa}{\ensuremath{\varphi}}
\newcommand{\seqb}{\ensuremath{\psi}}
\newcommand{\seqc}{\ensuremath{\vartheta}}

\newcommand{\vara}{\ensuremath{x}}

\newcommand{\fora}{\ensuremath{a}}
\newcommand{\forb}{\ensuremath{b}}
\newcommand{\forc}{\ensuremath{c}}
\newcommand{\ford}{\ensuremath{d}}
\newcommand{\fore}{\ensuremath{e}}

\newcommand{\terma}{\ensuremath{s}}
\newcommand{\termb}{\ensuremath{t}}
\newcommand{\termc}{\ensuremath{u}}
\newcommand{\termd}{\ensuremath{w}}
\newcommand{\terme}{\ensuremath{e}}

\newcommand{\metaassume}[1]{\textnormal{\TirNameStyle{\semilightgray{$\lbag$#1$\rbag$}}}}

\newcommand{\interpret}{\ensuremath{\mathfrak{I}}}

\newcommand{\turnstilesymbol}{\hqmorespace{\preceq}}

\newcommand{\tss}{\turnstilesymbol}

\newcommand{\rulenamespace}{\hspace*{.5em}}
\newcommand{\assumptionname}{asm.}
\newcommand{\asmcolor}[1]{\textcolor{lightgray!50!DodgerBlue}{#1}}
\newcommand{\textasm}{\textcolor{lightgray!50!DodgerBlue}{\TirNameStyle{\assumptionname}}}
\newcommand{\iassume}[1]{%
	{%
		\inferrule*{\textnormal{\small \raisebox{-.5ex}{\textcolor{lightgray!50!DodgerBlue}{\TirNameStyle{\assumptionname}}}}}{#1}
	}%
}

\newcommand{\iIH}[1]{%
	{%
		\inferrule*{\textnormal{\small \raisebox{-.5ex}{\textcolor{lightgray!50!DodgerBlue}{\TirNameStyle{i.h.}}}}}{#1}
	}%
}

\newcommand{\sassume}[2]{\nsaxiom{\textcolor{lightgray!50!DodgerBlue}{\textsc{\assumptionname}}}{#1}{#2}}
\newcommand{\seassume}[2]{\nequivaxiom{\textcolor{lightgray!50!DodgerBlue}{\textsc{\assumptionname}}}{#1}{#2}}
\newcommand{\irule}[2]{\inferrule*{#1}{#2}}
\newcommand{\nirule}[3]{\inferrule*[right={\rulenamespace{}\semigray{#1}}]{#2}{#3}}
\newcommand{\nlirule}[3]{\inferrule*[left={\semigray{#1}\rulenamespace}]{#2}{#3}}
\newcommand{\nirulet}[3]{\inferrule*[Right={\rulenamespace{}\semigray{#1}}]{#2}{#3}}
\newcommand{\nlirulet}[3]{\inferrule*[Left={\semigray{#1}\rulenamespace}]{#2}{#3}}

\newcommand{\nequalrule}[5]{\nirule{#1}{#2 \eeq #3}{#4 \eeq #5}}

\newcommand{\nsaxiom}[3]{\irule{\textnormal{\TirNameStyle{\raisebox{-.5ex}{\semigray{#1}}}}}{#2 \tss #3}}
\newcommand{\nequivaxiom}[3]{\irule{\textnormal{\TirNameStyle{\raisebox{-.5ex}{\semigray{#1}}}}}{#2 \hqeq #3}}
\newcommand{\nequalaxiom}[3]{\irule{\textnormal{\TirNameStyle{\raisebox{-.5ex}{\semigray{#1}}}}}{#2 \hqeq #3}}

\newcommand{\sugrule}[3]{%
	{
		\mprset{fraction={{\cdot}{\cdots}{\cdot}}}%
		\inferrule*[right={\rulenamespace{}\semigray{#1}}]{%
		\mprset{defaultfraction}
		#2 } { #3 }
	}
}
\newcommand{\sugrulet}[3]{%
	{
		\mprset{fraction={{\cdot}{\cdots}{\cdot}}}%
		\inferrule*[Right={\rulenamespace{}\semigray{#1}}]{%
		\mprset{defaultfraction}
		#2 } { #3 }
	}
}
\newcommand{\lsugrule}[3]{%
	{
		\mprset{fraction={{\cdot}{\cdots}{\cdot}}}%
		\inferrule*[left={\rulenamespace{}\semigray{#1}}]{%
		\mprset{defaultfraction}
		#2 } { #3 }
	}
}
\newcommand{\lsugrulet}[3]{%
	{
		\mprset{fraction={{\cdot}{\cdots}{\cdot}}}%
		\inferrule*[Left={\rulenamespace{}\semigray{#1}}]{%
		\mprset{defaultfraction}
		#2 } { #3 }
	}
}

\newcommand{\leequivname}{sugar}
\newcommand{\leequivrule}[2]{\sugrule{\leequivname}{#1}{#2}}

\newcommand{\leequivrulet}[2]{\sugrulet{\leequivname}{#1}{#2}}
\newcommand{\lleequivrulet}[2]{\lsugrulet{\leequivname}{#1}{#2}}

\newcommand{\equivlename}{desugar}
\newcommand{\equivlerule}[2]{\sugrule{\equivlename}{#1}{#2}}
\newcommand{\lequivlerule}[2]{\lsugrule{\equivlename}{#1}{#2}}
\newcommand{\equivlerulet}[2]{\sugrulet{\equivlename}{#1}{#2}}
\newcommand{\lequivlerulet}[2]{\lsugrulet{\equivlename}{#1}{#2}}

\newcommand{\eqreflexname}{eq-reflex}
\newcommand{\axeqreflex}[1]{\nequalaxiom{\eqreflexname}{#1}{#1}}

\newcommand{\eqsymmname}{symm}

\newcommand{\eqsymmrulet}[2]{\nirulet{\eqsymmname}{#1}{#2}}
\newcommand{\leqsymmrule}[2]{\nlirule{\eqsymmname}{#1}{#2}}
\newcommand{\leqsymmrulet}[2]{\nlirulet{\eqsymmname}{#1}{#2}}

\newcommand{\eqtransname}{eq-trans}
\newcommand{\eqtransrule}[3]{\nirule{\eqtransname}{#1 \\ #2}{#3}}
\newcommand{\eqtransrulet}[3]{\nirulet{\eqtransname}{#1 \\ #2}{#3}}
\newcommand{\leqtransrule}[3]{\nlirule{\eqtransname}{#1 \\ #2}{#3}}
\newcommand{\leqtransrulet}[3]{\nlirulet{\eqtransname}{#1 \\ #2}{#3}}

\newcommand{\eqcongname}{cong}
\newcommand{\eqcongrule}[3]{\nirule{\eqcongname}{#1 \\ #2}{#3}}
\newcommand{\eqcongrulet}[3]{\nirulet{\eqcongname}{#1 \\ #2}{#3}}
\newcommand{\leqcongrule}[3]{\nlirule{\eqcongname}{#1 \\ #2}{#3}}
\newcommand{\leqcongrulet}[3]{\nlirulet{\eqcongname}{#1 \\ #2}{#3}}

\newcommand{\reflexname}{reflex}
\newcommand{\axreflex}[1]{\nsaxiom{\reflexname}{#1}{#1}}

\newcommand{\botname}{bot}
\newcommand{\axbot}[1]{\nsaxiom{\botname}{\bot}{#1}}

\newcommand{\topname}{top}
\newcommand{\axtop}[1]{\nsaxiom{\topname}{#1}{\top}}

\newcommand{\nextbotname}{$\Nextsymbol$-of-$\bot$}
\newcommand{\axnextbot}{\nsaxiom{\nextbotname}{\Next \bot}{\bot}}

\newcommand{\nexttopname}{$\Nextsymbol$-of-$\top$}
\newcommand{\axnexttop}{\nsaxiom{\nexttopname}{\top}{\Next \top}}

\newcommand{\cutname}{trans}
\newcommand{\cutrule}[3]{\nirule{\cutname}{#1 \\ #2}{#3}}
\newcommand{\cutrulet}[3]{\nirulet{\cutname}{#1 \\ #2}{#3}}
\newcommand{\lcutrule}[3]{\nlirule{\cutname}{#1 \\ #2}{#3}}
\newcommand{\lcutrulet}[3]{\nlirulet{\cutname}{#1 \\ #2}{#3}}

\newcommand{\equivinequalLname}{weakenR}
\newcommand{\equivinequalLrule}[2]{\nirule{\equivinequalLname}{#1}{#2}}
\newcommand{\equivinequalLrulet}[2]{\nirulet{\equivinequalLname}{#1}{#2}}

\newcommand{\lequivinequalLrulet}[2]{\nlirulet{\equivinequalLname}{#1}{#2}}

\newcommand{\equivinequalRname}{weakenL}
\newcommand{\equivinequalRrule}[2]{\nirule{\equivinequalRname}{#1}{#2}}
\newcommand{\equivinequalRrulet}[2]{\nirulet{\equivinequalRname}{#1}{#2}}

\newcommand{\inequalequivname}{antisymm}
\newcommand{\inequalequivrule}[2]{\nirule{\inequalequivname}{#1}{#2}}
\newcommand{\inequalequivrulet}[2]{\nirulet{\inequalequivname}{#1}{#2}}

\newcommand{\cutsname}{\cutname${}^{*}$}
\newcommand{\cutsrule}[2]{\nirule{\cutsname}{#1}{#2}}
\newcommand{\cutsrulet}[2]{\nirulet{\cutsname}{#1}{#2}}

\newcommand{\antisymmname}{indiscern}
\newcommand{\antisymmrule}[3]{\nirule{\antisymmname}{#1 \\ #2}{#3}}
\newcommand{\antisymmrulet}[3]{\nirulet{\antisymmname}{#1 \\ #2}{#3}}
\newcommand{\lantisymmrule}[3]{\nlirule{\antisymmname}{#1 \\ #2}{#3}}
\newcommand{\lantisymmrulet}[3]{\nlirulet{\antisymmname}{#1 \\ #2}{#3}}

\newcommand{\joincommname}{$\lmax$-comm}
\newcommand{\axjoincomm}[2]{\nequivaxiom{\joincommname}{#1 \lmax #2}{#2 \lmax #1}}
\newcommand{\joincommrule}[2]{\nirule{\joincommname}{#1}{#2}}
\newcommand{\joincommrulet}[2]{\nirulet{\joincommname}{#1}{#2}}

\newcommand{\meetcommname}{$\lmin$-comm}
\newcommand{\axmeetcomm}[2]{\nequivaxiom{\meetcommname}{#1 \lmin #2}{#2 \lmin #1}}

\newcommand{\lmeetcommrulet}[2]{\nlirulet{\meetcommname}{#1}{#2}}

\newcommand{\joinassocname}{$\lmax$-assoc}
\newcommand{\axjoinassoc}[3]{\nequivaxiom{\joinassocname}{(#1 \lmax #2) \llmax #3}{#1 \llmax (#2 \lmax #3)}}

\newcommand{\meetassocname}{$\lmin$-assoc}
\newcommand{\axmeetassoc}[3]{\nequivaxiom{\meetassocname}{(#1 \lmin #2) \llmin #3}{#1 \llmin (#2 \lmin #3)}}

\newcommand{\joinabsorbname}{$\lmax$-absorb}
\newcommand{\axjoinabsorb}[2]{\nequivaxiom{\joinabsorbname}{#1 \llmax (#1 \lmin #2)}{#1}}

\newcommand{\meetabsorbname}{$\lmin$-absorb}
\newcommand{\axmeetabsorb}[2]{\nequivaxiom{\meetabsorbname}{#1 \llmin (#1 \lmax #2)}{#1}}

\newcommand{\nextsupcommname}{$\Nextsymbol\Supsymbol$-comm}
\newcommand{\axnextsupcomm}[1]{\nequivaxiom{\nextsupcommname}{\Next \Sup #1}{\Sup \Next #1}}

\newcommand{\nextsupcommrulet}[2]{\nirulet{\nextsupcommname}{#1}{#2}}

\newcommand{\nextinfcommname}{$\Nextsymbol\Infsymbol$-comm}
\newcommand{\axnextinfcomm}[1]{\nequivaxiom{\nextinfcommname}{\Next \Inf #1}{\Inf \Next #1}}

\newcommand{\nextinfcommrulet}[2]{\nirulet{\nextinfcommname}{#1}{#2}}

\newcommand{\joinintroLname}{$\lmax$-introL}
\newcommand{\joinintroLrule}[3]{\nirule{\joinintroLname}{#1 \\ #2}{#3}}
\newcommand{\joinintroLrulet}[3]{\nirulet{\joinintroLname}{#1 \\ #2}{#3}}
\newcommand{\ljoinintroLrule}[3]{\nlirule{\joinintroLname}{#1 \\ #2}{#3}}
\newcommand{\ljoinintroLrulet}[3]{\nlirulet{\joinintroLname}{#1 \\ #2}{#3}}

\newcommand{\joinintroRname}{$\lmax$-introR}
\newcommand{\joinintroRrule}[2]{\nirule{\joinintroRname}{#1}{#2}}
\newcommand{\joinintroRrulet}[2]{\nirulet{\joinintroRname}{#1}{#2}}
\newcommand{\ljoinintroRrule}[2]{\nlirule{\joinintroRname}{#1}{#2}}
\newcommand{\ljoinintroRrulet}[2]{\nlirulet{\joinintroRname}{#1}{#2}}

\newcommand{\meetintroLname}{$\lmin$-introL}
\newcommand{\meetintroLrule}[2]{\nirule{\meetintroLname}{#1}{#2}}
\newcommand{\meetintroLrulet}[2]{\nirulet{\meetintroLname}{#1}{#2}}

\newcommand{\lmeetintroLrulet}[2]{\nlirulet{\meetintroLname}{#1}{#2}}

\newcommand{\meetintroRname}{$\lmin$-introR}
\newcommand{\meetintroRrule}[3]{\nirule{\meetintroRname}{#1 \\ #2}{#3}}
\newcommand{\meetintroRrulet}[3]{\nirulet{\meetintroRname}{#1 \\ #2}{#3}}

\newcommand{\lmeetintroRrulet}[3]{\nlirulet{\meetintroRname}{#1 \\ #2}{#3}}

\newcommand{\infintroLname}{$\Infsymbol$-introL}
\newcommand{\infintroLrule}[2]{\nirule{\infintroLname}{#1}{#2}}

\newcommand{\joinelimname}{$\lmax$-elim}
\newcommand{\joinelimrule}[2]{\nirule{\joinelimname}{#1}{#2}}

\newcommand{\meetelimname}{$\lmin$-elim}
\newcommand{\meetelimrule}[2]{\nirule{\meetelimname}{#1}{#2}}
\newcommand{\meetelimrulet}[2]{\nirulet{\meetelimname}{#1}{#2}}

\newcommand{\infelimname}{$\Infsymbol$-elim}
\newcommand{\infelimrule}[2]{\nirule{\infelimname}{#1}{#2}}
\newcommand{\infelimrulet}[2]{\nirulet{\infelimname}{#1}{#2}}

\newcommand{\supintroRname}{$\Supsymbol$-introR}
\newcommand{\supintroRrule}[2]{\nirule{\supintroRname}{#1}{#2}}

\newcommand{\lsupintroRrulet}[2]{\nlirulet{\supintroRname}{#1}{#2}}

\newcommand{\supelimname}{$\Supsymbol$-elim}
\newcommand{\supelimrule}[2]{\nirule{\supelimname}{#1}{#2}}

\newcommand{\nextmononame}{$\Nextsymbol$-mono}
\newcommand{\nextmonorule}[2]{\nirule{\nextmononame}{#1}{#2}}
\newcommand{\nextmonorulet}[2]{\nirulet{\nextmononame}{#1}{#2}}

\newcommand{\supinfinfsupmononame}{$\Supsymbol\Infsymbol$\,/\,$\Infsymbol\Supsymbol$-mono}
\newcommand{\supinfinfsupmonorule}[2]{\nirule{\supinfinfsupmononame}{#1}{#2}}

\newcommand{\ascitername}{asc-iter}
\newcommand{\asciterrule}[2]{\nirule{\ascitername}{#1}{#2}}
\newcommand{\asciterrulet}[2]{\nirulet{\ascitername}{#1}{#2}}

\newcommand{\descitername}{desc-iter}
\newcommand{\desciterrule}[2]{\nirule{\descitername}{#1}{#2}}
\newcommand{\desciterrulet}[2]{\nirulet{\descitername}{#1}{#2}}

\newcommand{\orbitascname}{orbit-asc}
\newcommand{\orbitascrule}[2]{\nirule{\orbitascname}{#1}{#2}}
\newcommand{\orbitascrulet}[2]{\nirulet{\orbitascname}{#1}{#2}}

\newcommand{\orbitdescname}{orbit-desc}
\newcommand{\orbitdescrule}[2]{\nirule{\orbitdescname}{#1}{#2}}
\newcommand{\orbitdescrulet}[2]{\nirulet{\orbitdescname}{#1}{#2}}

\newcommand{\supascpointname}{$\Supsymbol$-asc-flat}
\newcommand{\supascpointrule}[2]{\nirule{\supascpointname}{#1}{#2}}

\newcommand{\infdescpointname}{$\Infsymbol$-desc-flat}
\newcommand{\infdescpointrule}[2]{\nirule{\infdescpointname}{#1}{#2}}

\newcommand{\supinflatename}{$\Supsymbol$-inflate}
\newcommand{\axsupinflate}[1]{\nsaxiom{\supinflatename}{#1}{\Sup #1}}

\newcommand{\infdeflatename}{$\Infsymbol$-deflate}
\newcommand{\axinfdeflate}[1]{\nsaxiom{\infdeflatename}{\Inf #1}{#1}}

\newcommand{\supmononame}{$\Supsymbol$-mono}
\newcommand{\supmonorule}[2]{\nirule{\supmononame}{#1}{#2}}
\newcommand{\supmonorulet}[2]{\nirulet{\supmononame}{#1}{#2}}

\newcommand{\lsupmonorulet}[2]{\nlirulet{\supmononame}{#1}{#2}}

\newcommand{\infmononame}{$\Infsymbol$-mono}
\newcommand{\infmonorule}[2]{\nirule{\infmononame}{#1}{#2}}
\newcommand{\infmonorulet}[2]{\nirulet{\infmononame}{#1}{#2}}

\newcommand{\linfmonorulet}[2]{\nlirulet{\infmononame}{#1}{#2}}

\newcommand{\nextindname}{$\Supsymbol$-ind}
\newcommand{\nextindrule}[2]{\nirule{\nextindname}{#1}{#2}}
\newcommand{\nextindrulet}[2]{\nirulet{\nextindname}{#1}{#2}}

\newcommand{\nextcoindname}{$\Infsymbol$-coind}
\newcommand{\nextcoindrule}[2]{\nirule{\nextcoindname}{#1}{#2}}
\newcommand{\nextcoindrulet}[2]{\nirulet{\nextcoindname}{#1}{#2}}

\newcommand{\lnextcoindrulet}[2]{\nlirulet{\nextcoindname}{#1}{#2}}

\newcommand{\suptightname}{$\Supsymbol$-introL}
\newcommand{\suptightrule}[3]{\nirule{\suptightname}{#1 \\ #2}{#3}}
\newcommand{\suptightrulet}[3]{\nirulet{\suptightname}{#1 \\ #2}{#3}}

\newcommand{\inftightname}{$\Infsymbol$-introR}
\newcommand{\inftightrule}[3]{\nirule{\inftightname}{#1 \\ #2}{#3}}
\newcommand{\inftightrulet}[3]{\nirulet{\inftightname}{#1 \\ #2}{#3}}
\newcommand{\linftightrule}[3]{\nlirule{\inftightname}{#1 \\ #2}{#3}}

\newcommand{\supexpandname}{$\Supsymbol$-exp}
\newcommand{\axsupexpand}[1]{\nequivaxiom{\supexpandname}{\Sup #1}{#1 \llmax \Next \Sup #1}}
\newcommand{\supexpandrule}[2]{\nirule{\supexpandname}{#1}{#2}}

\newcommand{\infexpandname}{$\Infsymbol$-exp}
\newcommand{\axinfexpand}[1]{\nequivaxiom{\infexpandname}{\Inf #1}{#1 \llmin \Next \Inf #1}}
\newcommand{\infexpandrule}[2]{\nirule{\infexpandname}{#1}{#2}}

\newcommand{\itername}{iter}
\newcommand{\axiter}[1]{\nequivaxiom{\itername}{\Next \Fas #1}{\Fa \Fas \Next #1}}
\newcommand{\iterrule}[2]{\nirule{\itername}{#1}{#2}}
\newcommand{\iterrulet}[2]{\nirulet{\itername}{#1}{#2}}

\newcommand{\Fmononame}{$\funca$-mono}
\newcommand{\Fmonorule}[2]{\nirule{\Fmononame}{#1}{#2}}
\newcommand{\Fmonorulet}[2]{\nirulet{\Fmononame}{#1}{#2}}

\newcommand{\lFmonorulet}[2]{\nlirulet{\Fmononame}{#1}{#2}}

\newcommand{\Fsmononame}{$\funca^*$-mono}

\newcommand{\Fsmonorulet}[2]{\nirulet{\Fsmononame}{#1}{#2}}

\newcommand{\lFsmonorulet}[2]{\nlirulet{\Fsmononame}{#1}{#2}}

\newcommand{\FFscommname}{$\funca\funca^*$-comm}
\newcommand{\axFFscomm}[1]{\nequivaxiom{\FFscommname}{\Fa \Fas #1}{\Fas \Fa #1}}

\newcommand{\FFscommrulet}[2]{\nirulet{\FFscommname}{#1}{#2}}

\newcommand{\FNextcommname}{$\funca\Next$-comm}
\newcommand{\axFNextcomm}[1]{\nequivaxiom{\FNextcommname}{\Fa \Next #1}{\Next \Fa #1}}

\newcommand{\FNextcommrulet}[2]{\nirulet{\FNextcommname}{#1}{#2}}

\newcommand{\lFNextcommrulet}[2]{\nlirulet{\FNextcommname}{#1}{#2}}

\newcommand{\Findname}{$\funca$-ind}
\newcommand{\Findrule}[2]{\nirule{\Findname}{#1}{#2}}
\newcommand{\Findrulet}[2]{\nirulet{\Findname}{#1}{#2}}

\newcommand{\Fcoindname}{$\funca$-coind}
\newcommand{\Fcoindrule}[2]{\nirule{\Fcoindname}{#1}{#2}}
\newcommand{\Fcoindrulet}[2]{\nirulet{\Fcoindname}{#1}{#2}}

\newcommand{\lFcoindrulet}[2]{\nlirulet{\Fcoindname}{#1}{#2}}

\newcommand{\FsintroLname}{$\funca^*$-introL}
\newcommand{\FsintroLrule}[3]{\nirule{\FsintroLname}{#1 \\ #2}{#3}}

\newcommand{\lFsintroLrulet}[3]{\nlirulet{\FsintroLname}{#1 \\ #2}{#3}}

\newcommand{\FsintroRname}{$\funca^*$-introR}
\newcommand{\FsintroRrule}[3]{\nirule{\FsintroRname}{#1 \\ #2}{#3}}
\newcommand{\FsintroRrulet}[3]{\nirulet{\FsintroRname}{#1 \\ #2}{#3}}

\newcommand{\semicontname}{semi-cont}
\newcommand{\axsemicont}[1]{\nsaxiom{\semicontname}{\Sup \Fa #1}{\Fa \Sup #1}}

\newcommand{\semicocontname}{semi-cocont}
\newcommand{\axsemicocont}[1]{\nsaxiom{\semicocontname}{\Fa \Inf #1}{\Inf \Fa #1}}

\newcommand{\alephzeroshortname}{c}
\newcommand{\alephcontname}{\alephzeroshortname-cont}
\newcommand{\axalephcont}[1]{\nsaxiom{\alephcontname}{\Fa \Sup #1}{\Sup \Fa #1}}

\newcommand{\alephcocontname}{\alephzeroshortname-cocont}
\newcommand{\axalephcocont}[1]{\nsaxiom{\alephcocontname}{\Inf \Fa #1}{\Fa \Inf #1}}

\newcommand{\omegacontname}{$\omega$-cont}
\newcommand{\omegacontrule}[2]{\nirule{\omegacontname}{#1}{#2}}

\newcommand{\omegacocontname}{$\omega$-cocont}
\newcommand{\omegacocontrule}[2]{\nirule{\omegacocontname}{#1}{#2}}
\newcommand{\omegacocontrulet}[2]{\nirulet{\omegacocontname}{#1}{#2}}

\newcommand{\supdescname}{$\Supsymbol$-desc}
\newcommand{\axsupdesc}[1]{\nsaxiom{\supdescname}{\Next \Sup #1}{\Sup #1}}

\newcommand{\infascname}{$\Infsymbol$-asc}
\newcommand{\axinfasc}[1]{\nsaxiom{\infascname}{\Inf #1}{\Next \Inf #1}}

\newcommand{\nextoverjoinname}{$\Nextsymbol$-over-$\lmax$}
\newcommand{\axnextoverjoin}[2]{\nequivaxiom{\nextoverjoinname}{\Next (#1 \lmax #2)}{\Next #1 \llmax \Next #2}}

\newcommand{\nextoverjoinrulet}[2]{\nirulet{\nextoverjoinname}{#1}{#2}}

\newcommand{\nextovermeetname}{$\Nextsymbol$-over-$\lmin$}
\newcommand{\axnextovermeet}[2]{\nequivaxiom{\nextovermeetname}{\Next (#1 \lmin #2)}{\Next #1 \llmin \Next #2}}

\newcommand{\lnextovermeetrulet}[2]{\nlirulet{\nextovermeetname}{#1}{#2}}

\newcommand{\supoverjoinname}{$\Supsymbol$-over-$\lmax$}
\newcommand{\axsupoverjoin}[2]{\nequivaxiom{\supoverjoinname}{\Sup (#1 \lmax #2)}{\Sup #1 \llmax \Sup #2}}

\newcommand{\infovermeetname}{$\Infsymbol$-over-$\lmin$}
\newcommand{\axinfovermeet}[2]{\nequivaxiom{\infovermeetname}{\Inf (#1 \lmin #2)}{\Inf #1 \llmin \Inf #2}}

\newcommand{\supovermeetname}{$\Supsymbol$-over-$\lmin$}
\newcommand{\axsupovermeet}[2]{\nsaxiom{\supovermeetname}{\Sup (#1 \lmin #2)}{\Sup #1 \llmin \Sup #2}}

\newcommand{\infoverjoinname}{$\Infsymbol$-over-$\lmax$}
\newcommand{\axinfoverjoin}[2]{\nsaxiom{\infoverjoinname}{\Inf #1 \llmax \Inf #2}{\Inf (#1 \lmax #2)}}

\newcommand{\infidemname}{$\Infsymbol$-idem}
\newcommand{\axinfidem}[1]{\nequivaxiom{\infidemname}{\Inf \Inf #1}{\Inf #1}}

\newcommand{\infidemrulet}[2]{\nirulet{\infidemname}{#1}{#2}}

\newcommand{\supidemname}{$\Supsymbol$-idem}
\newcommand{\axsupidem}[1]{\nequivaxiom{\supidemname}{\Sup \Sup #1}{\Sup #1}}
\newcommand{\supidemrule}[2]{\nirule{\supidemname}{#1}{#2}}

\newcommand{\joinidemname}{$\lmax$-idem}
\newcommand{\axjoinidem}[1]{\nequivaxiom{\joinidemname}{#1 \lmax #1}{#1}}

\newcommand{\meetidemname}{$\lmin$-idem}
\newcommand{\axmeetidem}[1]{\nequivaxiom{\meetidemname}{#1 \lmin #1}{#1}}

\newcommand{\supquasiprefpname}{$\Supsymbol$-quasi-pre-fp}
\newcommand{\supquasiprefprule}[2]{\nirule{\supquasiprefpname}{#1}{#2}}
\newcommand{\supquasiprefprulet}[2]{\nirulet{\supquasiprefpname}{#1}{#2}}

\newcommand{\supquasipostfpname}{$\Supsymbol$-quasi-post-fp}
\newcommand{\supquasipostfprule}[2]{\nirule{\supquasipostfpname}{#1}{#2}}

\newcommand{\tkpfpname}{tkp-fp}
\newcommand{\tkpfprule}[2]{\nirule{\tkpfpname}{#1}{#2}}

\newcommand{\tkppostfpname}{tkp-post-fp}

\newcommand{\tkpleastname}{tkp-least}
\newcommand{\tkpleastrule}[2]{\nirule{\tkpleastname}{#1}{#2}}

\newcommand{\ltkpleastrulet}[2]{\nlirulet{\tkpleastname}{#1}{#2}}

\newcommand{\tkpabovename}{tkp-above}
\newcommand{\tkpaboverule}[2]{\nirule{\tkpabovename}{#1}{#2}}

\newcommand{\olprefpname}{ol-pre-fp}
\newcommand{\olprefprule}[2]{\nirule{\olprefpname}{#1}{#2}}

\newcommand{\olpostfpname}{ol-post-fp}
\newcommand{\olpostfprule}[2]{\nirule{\olpostfpname}{#1}{#2}}

\newcommand{\olfpname}{ol-fp}
\newcommand{\olfprule}[2]{\nirule{\olfpname}{#1}{#2}}

\newcommand{\infsupquasiprefpname}{$\Infsymbol\Supsymbol$-quasi-pre-fp}
\newcommand{\infsupquasiprefprule}[2]{\nirule{\infsupquasiprefpname}{#1}{#2}}
\newcommand{\infsupquasiprefprulet}[2]{\nirulet{\infsupquasiprefpname}{#1}{#2}}

\newcommand{\infsupquasipostfpname}{$\Infsymbol\Supsymbol$-quasi-post-fp}
\newcommand{\infsupquasipostfprule}[2]{\nirule{\infsupquasipostfpname}{#1}{#2}}

\newcommand{\linfsupquasipostfprulet}[2]{\nlirulet{\infsupquasipostfpname}{#1}{#2}}

\newcommand{\shiftpointname}{$\Nextsymbol$-flat}
\newcommand{\shiftpointrule}[2]{\nirule{\shiftpointname}{#1}{#2}}

\newcommand{\kindname}{$k$-ind}
\newcommand{\kindrule}[3]{\nirule{\kindname}{#1 \\ #2}{#3}}

\newcommand{\axkinddescname}[1]{$\Fkinditersymbol{#1}$-deflate}
\newcommand{\axkinddesc}[3]{\nsaxiom{\axkinddescname{#3}}{#1}{#2}}

\newcommand{\axkinddescbname}[1]{$\Fkinditersymbol{#1}$-below}
\newcommand{\axkinddescb}[3]{\nsaxiom{\axkinddescbname{#3}}{#1}{#2}}

\newcommand{\kindparkrulename}{$k$-ind-park}
\newcommand{\kindparkrule}[2]{\nirule{\kindparkrulename}{#1}{#2}}
\newcommand{\kindparkrulet}[2]{\nirulet{\kindparkrulename}{#1}{#2}}

\newcommand{\kindpreserveascname}[1]{$\Fkinditer{#1}$-asc-pres}
\newcommand{\kindpreserveascrule}[3]{\nirule{\kindpreserveascname{#3}}{#1}{#2}}
\newcommand{\kindpreserveascrulet}[3]{\nirulet{\kindpreserveascname{#3}}{#1}{#2}}

\newcommand{\sfsymbol}[1]{\textsf{\upshape {#1}}}

\newcommand{\conditionalPair}[2]{{\let\oldarraystretch\arraystretch}\renewcommand{\arraystretch}{1}~\holter{~\raisebox{.5ex}{${#1}$}~}{~\raisebox{.125ex}{${#2}$}~}~\renewcommand{\arraystretch}{\oldarraystretch}}

\newcommand{\AVAILLOC}[1]{\PosNats}

\newcommand{\Terms}{\ensuremath{\mathsf{Terms}}\xspace}   

\newcommand{\Nats}{\ensuremath{\mathbb{N}}\xspace}
\newcommand{\PosNats}{\ensuremath{\mathbb{N}_{>0}}\xspace}

\newcommand{\Reals}{\mathbb{R}}

\newcommand{\minfty}{{-}\infty}
\newcommand{\pinfty}{{+}\infty}

\newcommand{\subst}[2]{\left[ {#1} \middle\backslash {#2}\right]}

\newcommand{\nsem}[3]{\ensuremath{\tensor*[^{#2}_{}]{\left\llbracket {#1} \right\rrbracket}{_{{#3}}}}}
\newcommand{\nsemi}[2]{\nsem{#1}{\interpret}{#2}}

\newcommand{\nsemin}[1]{\nsem{#1}{\interpret}{n}}

\newcommand{\sem}[2]{\ensuremath{\llbracket {#1} \rrbracket}_{#2}}
\newcommand{\semantics}[1]{\ensuremath{\llbracket {#1} \rrbracket}}
\newcommand{\semn}[1]{\sem{#1}{n}}

\newcommand{\colmark}[1]{\colorbox{DodgerBlue3!30!white}{$#1$}}

\newcommand{\qqiff}{\qquad\textnormal{iff}\qquad}

\newcommand{\lqqiff}{\textnormal{iff}\qquad}

\newcommand{\qand}{\quad\textnormal{and}\quad}
\newcommand{\qqand}{\qquad\textnormal{and}\qquad}

\newcommand{\qimplies}{\quad\textnormal{implies}\quad}
\newcommand{\qqimplies}{\qquad\textnormal{implies}\qquad}

\newcommand{\lqqimplies}{\textnormal{implies}\qquad}

\newcommand{\morespace}[1]{~{}#1{}~}
\newcommand{\qmorespace}[1]{\quad{}#1{}\quad}
\newcommand{\hqmorespace}[1]{\hspace*{.35em}{}#1{}\hspace*{.35em}}
\newcommand{\qqmorespace}[1]{\qquad{}#1{}\qquad}

\newcommand{\ppreceq}{~{}\preceq{}~}

\newcommand{\eeq}{\morespace{=}}
\newcommand{\hqeq}{\hqmorespace{=}}

\newcommand{\hqequiv}{\hqmorespace{\equiv}}

\newcommand{\ssqcap}{~{}\sqcap{}~}
\newcommand{\ssqcup}{~{}\sqcup{}~}

\newcommand{\qmid}{\quad{\!}\semigray{|}{\!}\quad}

\newcommand{\qeq}{\quad{}={}\quad}

\newcommand{\ssqsubseteq}{~{}\sqsubseteq{}~}
\newcommand{\ssqsupseteq}{\morespace{\sqsupseteq}}

\newcommand{\setcomp}[2]{\left\{\, {#1} ~\middle|~ {#2} \,\right\}}

\definecolor{webgreen}{rgb}{0,.5,0}
\newcommand{\gray}[1]{\textcolor{gray}{#1}}
\newcommand{\semigray}[1]{\textcolor{gray!75!lightgray}{#1}}
\newcommand{\semilightgray}[1]{\textcolor{gray!50!lightgray}{#1}}
\newcommand{\lightgray}[1]{\textcolor{lightgray}{#1}}
\newcommand{\verylightgray}[1]{\textcolor{lightgray!75!white}{#1}}

\newcommand{\green}[1]{\textcolor{webgreen}{#1}}

\newcommand{\blue}[1]{\textcolor{DodgerBlue3!75!DodgerBlue4}{#1}}
\newcommand{\orange}[1]{\textcolor{orange}{#1}}
\newcommand{\purple}[1]{\textcolor{purple!70!red}{#1}}

\newcounter{computationarrowsone}
\newcounter{computationarrowstwo}

\newcounter{sarrow}

\newcommand{\lfp}{\ensuremath{\textnormal{\sfsymbol{lfp}}~}}
\newcommand{\gfp}{\ensuremath{\textnormal{\sfsymbol{gfp}}~}}

\newcommand{\lfpabove}[1]{\ensuremath{\textnormal{\sfsymbol{lfp}}_{{\larevord} #1}~}}

\newcommand{\shift}{\Next}

\newcommand{\mn}{\medskip\noindent}

\renewcommand{\implies}{\Longrightarrow}

\newcommand{\Ax}{\mathsf{Ax}}


\makeatletter
\newcommand{\pushright}[1]{\ifmeasuring@#1\else\omit\hfill$\displaystyle#1$\fi\ignorespaces}
\newcommand{\pushleft}[1]{\ifmeasuring@#1\else\omit$\displaystyle#1$\hfill\fi\ignorespaces}
\makeatother

\tikzstyle{shiftarr}=[
        rounded corners,%
        to path={--([#1]\tikztostart.center)
                     -- ([#1]\tikztotarget.center) \tikztonodes
                     -- (\tikztotarget)},
]

\tikzset{
    commutative diagrams/.cd,
    arrow style=tikz,
    diagrams={>=stealth},
    row sep=large,
    column sep = huge
}

\newcommand{\Deltas}{\Delta^{*}}
\newcommand{\Sigmas}{\Sigma^{*}}

\newcommand{\Var}{\mathop{\mathsf{Var}}}
\newcommand{\V}{\mathcal{V}}
\newcommand{\Nat}{\mathbb{N}}

\newcommand{\seq}{\subseteq}
\newcommand{\ol}{\overline}

\newcommand{\interpretsemol}[1]{{^{\ol{\interpret}}\llbracket #1 \rrbracket}}

\newcommand{\quasieq}{\ensuremath{\textup{\textsc{q}}}\xspace}

\category{extended version} 

\relatedversion{This article has been accepted at \emph{Logic in Computer Science} (LICS) 2026.} 
\relatedversiondetails[linktext={\semigray{[to appear]}},cite=aic-conference-version]{Conference version}{https://lics.siglog.org/lics26/} 

\supplement{Isabelle formalization of AIC and formal proofs of fixed point theorems}
\supplementdetails[linktext={link to isa-afp.org},cite=aic-afp, subcategory={Isabelle/HOL Formalization}]{Software}{https://isa-afp.org/entries/Iteration_Algebra.html} 

\acknowledgements{This work was initiated during a workshop organized by Prakash Panangaden and Alexandra Silva in 2023 at the Bellairs Research Institute of McGill University in Holetown, Barbados.
We thank Bellairs for providing a wonderful research environment.
Benjamin Kaminski would like to thank Jan Reineke for his valuable comments.}

\nolinenumbers 

\begin{document}

\maketitle

\begin{abstract}%
Fixed points are a recurring theme in computer science and are often constructed as limits of suitably seeded fixed point iterations.
We present the \emph{algebra of iterative constructions} (AIC) \mbox{-- a purely} algebraic approach to reasoning about fixed point iterations of continuous endomaps on complete lattices.
AIC allows derivations of constructive fixed point theorems via equational logic and avoids explicit computations with indices.
For example,%
\begin{align*}
	\Fa \Sup \Fs \bot \qeq \Sup \Fs \bot
\end{align*}%
states in AIC that $\sup_n F^n (\bot)$ -- a construction known from the Kleene fixed point theorem -- is a fixed point of~$F$.
We demonstrate the applicability of AIC by providing algebraic proofs of several well- and less-well-known fixed point theorems: 
Among others, we prove the \emph{Tarski-Kantorovich principle} -- a generalization of the \emph{Kleene fixed point theorem} -- as well as a fixed point-theoretic generalization of $k$-induction -- a technique used in software verification. 

We moreover present a novel fixed point theorem. 
It improves a recent generalization of the Tarski-Kantorovich principle due to Olszewski for obtaining pre- and postfixed points from lattice-theoretic limit inferiors and limit superiors through iterating an endomap on an \emph{arbitrary} seed element:
We identify sufficient continuity conditions on the endomaps so that these limits become \emph{proper} fixed points.
We have mechanized our algebra in Isabelle/HOL. 
Isabelle's sledgehammer tool is able to find proofs of the above fixed point theorems fully automatically.

Finally, we investigate the completeness of our axiomatization of AIC.  
We prove that our finite set of finitary axioms is (a) sound but incomplete for standard models of AIC (sequences of elements from a complete lattice) and that (b) a different finite set of \emph{infinitary} axioms is complete.
We also prove that infinitary axioms are \emph{unavoidable}: there exists no complete axiomatization of standard models given by \emph{finitely many finitary} axioms.
\end{abstract}

\section{Introduction and Overview}

%
Fixed points are ubiquitous in computer science and related fields: 
shortest paths~\cite{DBLP:journals/ipl/Misra01}, game equilibria~\cite{doi:10.1073/pnas.36.1.48}, social choice theory~\cite{DBLP:journals/corr/abs-1907-10381}, dynamical systems~\cite{Hirschmonomaps}, or semantics~\cite{abramsky1994domain,gunter1992semantics} and verification~\cite{DBLP:conf/focs/Clarke77} are only a few of many examples.

In verification, for example, a common task is to lower-bound the greatest fixed point~(\textsf{gfp}) of a continuous endomap $\funca$, because this amounts to finding a loop invariant that implies the most general loop invariant.
Let $\fora$ be a \underline{\emph{candidate}} for such a lower bound on $\gfp \funca$.
Then if $\fora \preceq \funca(\fora)$ holds, we have \mbox{-- by} \mbox{coinduction --} verified that the candidate~$\fora$ \emph{is indeed} a lower bound on $\gfp \funca$.
Dually, if $\funca(\fora) \preceq \fora$, then $\fora$ is a verified \emph{upper bound} on $\lfp \funca$, so:
\begin{align*}
	\fora \ppreceq \funca(\fora) \qimplies \fora \ppreceq \gfp \funca 
	\quad\qqand\quad 
	\funca(\fora) \ppreceq \fora \qimplies \lfp \funca \ppreceq \fora
\end{align*}
One justification of the above goes through a certain \emph{fixed point iteration} principle:
Consider the case for $\gfp \funca$.
If $\fora \preceq \funca(\fora)$, then this \enquote{kicks off} an ascending fixed point iteration $\fora \preceq \funca(\fora) \preceq \funca^2(\fora) \preceq {\ldots}$ which (under suitable continuity conditions on $\funca$) in the limit (supremum) converges to a fixed point \enquote{just above}~$a$.
More precisely: $\sup_{k \in \omega} \funca^k(\fora)$ is the least fixed point of~$\funca$ that is greater than or equal to $\fora$.
Since $\fora \preceq \sup_{k \in \omega} \funca^k(\fora)$ and since $\sup_{k \in \omega} \funca^k(\fora)$ is itself a fixed point and is hence necessarily less or equal to the \emph{greatest} fixed point $\gfp \funca$, we get that $\fora \preceq \gfp \funca$.
Dually, if $\funca(\fora) \preceq \fora$, then $\inf_{k \in \omega} \funca^k(\fora)$ is the \emph{greatest} fixed point of~$\funca$ \enquote{just below}~$\fora$ and $\fora$ upper-bounds $\lfp \funca$.
This fixed point iteration principle is called the \emph{Tarski-Kantorovich principle}~\cite{jachymski2000tarski}.
\begin{center}
	\bfseries%
	But what if \underline{neither} $\boldsymbol{\fora \preceq \funca(\fora)}$ \underline{nor} $\boldsymbol{\funca(\fora) \preceq \fora}$ hold?\\[.25em]
	What does iterating $\boldsymbol{\funca}$ on $\boldsymbol{\fora}$ yield?
\end{center}%
Is $\lim_{k \to \omega} \funca^k(\fora)$ even well-defined?
Indeed, both $\liminf_{k \to \omega} \funca^k(\fora)$ and $\limsup_{k \to \omega} \funca^k(\fora)$ are well-defined and (under certain slightly stronger continuity conditions on $\funca$) they are both \emph{fixed points} of~$\funca$.
In other words: From \emph{any} candidate $\fora$, we can construct by suitable iteration \emph{some} fixed point of $\funca$ (possibly even two distinct ones).
This is an improvement on a recent (but unfortunately rather unknown) result due to Olszewski~\cite{Ol:omegaSequences} and to the best of our knowledge a novel fixed point theorem.

Also novel -- and our main contribution -- is the way in which we prove such fixed point iteration theorems.
Let us first take a completely standard definition of a limit superior:
\begin{align*}
	\limsup_{k \to \omega}~ \funca^k(\fora) 
	\qeq 
	\inf_{0 \leq \orange{n}} \: \sup_{\orange{n} \leq \blue{k}}\: \funca^{\blue{k}}(\fora)
\end{align*}%
On the right, there are really three nested \emph{sequences} involved:
The first one is \(\bigl(F^{\blue{k}}(a) \bigr)_{\blue{k} \in \omega}\), obtained by iterating \(F\) on~$\fora$.
The second sequence is \(\bigl(\sup_{\orange{n} \leq \blue{k}} F^{\blue{k}(a)}\bigr)_{\orange{n} \in \omega}\).
By making the index $\blue{k}$ of the first sequence depend on the index~$\orange{n}$ of the second one, we can imagine how \enquote{$\sup_{\orange{n} \leq \blue{k}}$} acts as an operation, transforming the first into the second sequence.
The third sequence is, perhaps surprisingly, \(\bigl(\inf_{\green{m} \le \orange{n}} \sup_{\orange{n} \le \blue{k}} F^{\blue{k}}(a)\bigr)_{\green{m} \in \omega}\), constructed by \enquote{$\inf_{\green{m} \leq \orange{n}}$} from the second one, making $\orange{n}$ depend on~$\green{m}$.
Our desired limit superior is then at index \(\green{m} = 0\).

Reasoning about limit superiors and inferiors naturally involves such interdependent indices.
Nevertheless:
\emph{Our key observation is that the exact indices do not really matter and are in fact distracting.}
We can rather think of \(\sup\), \(\inf\), and \(F^{(-)}\) as operations \mbox{-- almost} modalities -- on sequences:
Concretely, we think of $\inf$ as $\Infsymbol$, of $\sup$ as $\Supsymbol$, and we turn $F^{(-)}$ into a sort of iteration operation~$\funca^{*}$\!.
Altogether, \mbox{we obtain}%
\begin{align*}
	\limsup_{k \to \omega}~ \funca^k(\fora) 
	\qeq
	\inf_{0 \leq n} \: \sup_{n \leq k}\: \funca^k(\fora)
	\,\qmorespace{\textnormal{\enquote{$=$}}}\, \Inf \Sup \Fas \fora~.
\end{align*}
In this paper, we will make the above ideas rigorous and introduce the \emph{algebra of iterative constructions}~(AIC), in which we can reason algebraically about fixed points of endomaps on complete lattices.
For example, we will be able to prove that for continuous $\funca$,%
\begin{align*}
	\Fa \Sup \Fas \bot \hqeq \Sup \Fas \bot~,
\end{align*}%
(cf.\ $\sup_{k} \funca^k(\bot)$), thus proving part of the Kleene fixed point theorem.
We demonstrate the usefulness of AIC by means of several case studies.
Purely within AICs (or, to put it a bit flippantly, purely by pushing boxes and diamonds around), we prove:%
\begin{enumerate}
	\item 
		The Tarski-Kantorovich principle, a generalized version of the well-known Kleene fixed point theorem: If $\funca$ is $\omega$-continuous, then $\Sup \Fas \fora$ is the least fixed point of $\funca$ just above~$\fora$.

	\item 
		A generalization of Park induction ($\fora \preceq F(\fora)$ implies $\lfp F \preceq \fora$), which is a heavily used principle in verification of various types of systems.

	\item 
		A novel fixed point theorem which improves a result due to Olszewski~\cite{Ol:omegaSequences}: 
		If $\funca$ is countably continuous (preserving suprema of \emph{arbitrary countable sets}) and $\omega$-cocontinous, then $\Inf \Sup \Fas \fora$ is a fixed point of~$\funca$.
	\item Latticed $k$-induction~\cite{kind_cav}, a generalization of the $k$-induction principle~\cite{DBLP:conf/fmcad/SheeranSS00}, which is extremely effective in hardware and software verification.
\end{enumerate}%
We have mechanized AIC and the algebraic proofs of the theorems above in Isabelle/HOL~\cite{aic-afp}. 
The mechanized algebraic proofs follow the manual exposition very closely, with intermediate steps solved automatically. 
High-level theorems such as Tarski-Kantorovich or Park induction can even be proved fully automatically.

\subparagraph*{Overview and Organization of this Paper.}

We present the algebraic structures we treat in this paper \mbox{-- in} particular, our standard models, which we call \emph{sequence algebras} -- in \Cref{sec:synsem}.
There, we also show how equational proofs for such structures work.
In \Cref{sec:calc}, we present a finitary base axiomatization of AIC called $\AIC_0$, which is sound for sequence algebras.
We also provide an axiomatization with additional (redundant) axioms called $\AIC_1$, which is more usable for automatic proof search.
We demonstrate applicability of AIC through algebraic proofs of fixed point theorems 
and comment on pen-and-paper-style algebraic proofs 
in \Cref{sec:case_studies}.
We comment on the mechanization of AIC in Isabelle/HOL in \Cref{sec:isabelle}.
We investigate aspects of (in)completeness of AIC axiomatizations in \Cref{sec:completeness}.
We discuss limitations and give an outlook on future directions in \Cref{sec:concl}.

\subparagraph*{Related Work.} 

\emph{Axiomatic approaches} to fixed point theory have a long tradition in theoretical computer science, having appeared in different settings and levels of generality. 
The most prominent examples are Kozen's complete axiomatization of the modal $\mu$-calculus~\cite{kozen82,walukiewicz95}, and the abstract perspective on fixed point equations of Bloom and Esik's iteration theories~\cite{bloom-esik93,amv07}, which extend the categorical account of universal algebra based on Lawvere theories. 
While these works are aimed at axiomatically capturing \emph{properties} of least or greatest fixed points, the focus of AIC lies on reasoning about the iterative \emph{construction} of (not necessarily least or greatest) fixed points and related objects, such as pre- \mbox{and postfixed points}.

Approximating (instead of merely stating the existence of) fixed points was somewhat recently investigated in a series of papers~\cite{DBLP:journals/iandc/BaldanKP24,DBLP:journals/lmcs/BaldanEKP23,DBLP:conf/csl/000124}.
In particular \cite{DBLP:journals/lmcs/BaldanEKP23} tackles the problem of approximating in the directions that are \emph{not} covered by usual (co)induction: \emph{under}approximating \emph{least} and \emph{over}approximating \emph{greatest} fixed points.

Finally, our $\Supsymbol$ operations seem very related to the notion of \emph{least decreasing majorants} from Banach function spaces~\cite{sinnamon2007monotonicity,SANTACRUZHIDALGO2024110490,kiwerski2024arithmeticinterpolationfactorizationamalgams}.
We discuss further related work in \Cref{sec:concl}.%

\section{The Algebra of Iterative Constructions}
\label{sec:synsem}

The \emph{algebra of iterative constructions} (AIC) is an \emph{algebraic} theory, in the sense that it comprises an \emph{algebraic signature} that dictates the available operations and their arities, and a set of equational inference rules that tell us how these operations interact with one another.
We parameterize the theory in a set $\Funcs = \{\funca,\funcb,\funcc,\ldots\}$ of \emph{function symbols}, which abstractly represent monotonic maps for which we will be constructing and reasoning about fixed points.

Let us first deal with the signature.
Given \(\Funcs\), the \emph{signature of AIC} consists of two constants~$\bot$ and $\top$, two binary operations \(\lmax\) and~\(\lmin\), four unary operations $\Supsymbol$, $\Infsymbol$, $\Firstsymbol$, and~$\Nextsymbol$, as well as for each functions symbol \(F \in \Funcs\) two unary operations $F$ and~$F^*$\!. 
Models of AIC will be \emph{(algebraic) structures} with this signature.%
\begin{definition}[AIC-Structures]%
\label{def:algebra}%
	An \emph{AIC-structure}
	is a tuple%
	\begin{align*}
		\Alg 
		\!\eeq\!
		\bigl(
			A,\, 
			\botAlg,\, 
			\topAlg, 
			{\lmaxAlg},\, 
			{\lminAlg},\, 
			\SupAlgsymbol,\, 
			\InfAlgsymbol,\, 
			\FirstAlgsymbol,\, 
			\NextAlgsymbol,\,
			\setcomp{\smash{\funcaAlg}}{\funca \in \Funcs}\!,\, 
			\setcomp{\smash{\FsAlg}}{\funca \in \Funcs}
		\bigr),
	\end{align*}%
	where \(A\) is a set with $\botAlg, \topAlg \in A$ and moreover \({\lmaxAlg}, {\lminAlg} \colon A\times A \to A\) and $\SupAlgsymbol, \InfAlgsymbol, \FirstAlgsymbol, \NextAlgsymbol, \funcaAlg,$ $\FsAlg \colon A \to A$.
	The set \(A\) is called the \emph{carrier} of \(\Alg\) and we also say that $A$ \emph{carries} $\Alg$.
	
	Given two AIC-structures $\Alg$ and $\AlgB$, an \emph{algebra homomorphism} \(h \colon \Alg \to \AlgB\) is a function \(h \colon A \to B\) that preserves the algebraic structure.
	i.e.\ \(h\bigl(\botAlg\bigr) = \botAlgB\) and \(h\bigl(\topAlg\bigr) = \topAlgB\), and for any \(\ell, m \in A\), 
	we have \(h\bigl(\ell \lmaxAlg m \bigr) = h(\ell) \lmaxAlgB h(m)\), and
	\(h\bigl(\SupAlg \ell \bigr) = \SupAlgB h(\ell)\), and so on.%
\end{definition}%
We drop the superscript \(\semigray{{}^\Alg}\) when the $\Alg$ is clear from context
An overview of the meta{\-}variables we use throughout this paper is given in \Cref{tab:metavars}.
%
%
%
\renewcommand{\arraystretch}{1.25}%
\begin{table}[t]%
	\caption{Overview of metavariables.}%
	\label{tab:metavars}%
	\centering%
	\begin{adjustbox}{max width=.999\linewidth}%
		\begin{tabular}{l@{\qquad}l}
			\toprule
			%
			$\fora,\, \forb,\, \forc,\, \ldots{}$		& variables (appearing in terms)				\\
			$\Alg,\, \AlgB,\, \AlgTerm,\, \ldots{}$	& algebraic structures / AIC-structures 			\\
			$\funca,\, \funcb,\, \funcc,\, \ldots{}$	& function symbols representing monotonic maps		\\
			$\interpret$					& interpretation of variables (appearing in terms)		\\
			%
			%
			%
			%
			$\interpret_0$					& interpretation of the function symbols \\
			$\laelem,\, \laelemb,\, \cdots$		& elements from the ambient lattice\\
			%
			%
			$\seqa,\, \seqb,\, \seqc,\, \ldots{}$	& sequences of lattice elements\\
			$\terma,\, \termb,\, \termc,\, \terme,\, \ldots{}$	& iterative construction terms (denoting sequences)	\\
			\bottomrule
		\end{tabular}%
	\end{adjustbox}%
\end{table}%

So far,
nothing in the definition of AIC-structures
dictates what the different operations mean or how they interact.
Since our central aim is to capture iterative constructions of lattice fixed points, our primary object of study are going to be AIC-structures that are carried by sequences of elements from a lattice. 
We will call these sequence algebras.

\subsection{Sequence Algebras}

Let $\omega = \{0,\, 1,\, 2,\, 3,\, \ldots\}$ be the natural numbers.
Given a set \(\ladomain\), a function $\varphi\colon \omega \to \ladomain$ is called a \emph{sequence (in~$\ladomain$)}. 
We let $\laseqs$ be the \emph{set of all sequences in $\ladomain$}.
We usually write \(\varphi = \sequence{\varphi_0}{\varphi_1}{\varphi_2}{\varphi_3}\), where $\varphi_i$ denotes the $i^{\textnormal{th}}$ element of $\varphi$.
We say that a sequence is \emph{flat} if it is a \emph{constant} sequence $\constseq{\ell}$ and denote such a sequence by $\overline{\ell}$.

For the remainder of the paper, \((L,\, {\sqsubseteq})\) is a fixed \emph{countably complete} lattice (i.e.\ a lattice where every countable subset $\lasubset \subseteq \ladomain$, including the empty set, has a supremum and infimum in~$\ladomain$) with bottom and top elements \(\botL\) and \(\topL\), and join and meet denoted by \(\sqcup\) and \(\sqcap\).
We denote by \(\Mon(L)\) the set of monotonic maps \(L \to L\) (monotonic with respect to $\sqsubseteq$). 
An example of a countably-complete lattice is the extended real numbers $\overline{\mathbb{R}} = \bigl(\Reals \cup \{\minfty,\, \pinfty\},\, {\leq}\bigr)$, for which $\Mon(\smash{\overline{\mathbb{R}}})$ consists of the nondecreasing functions.%

We will now define the kind of AIC-structures suitable for reasoning about fixed point iterations of monotonic endomaps in $L$.%
\begin{definition}[Sequence Algebras]
\label{def:seq-algebras}
	Every countably complete lattice $(L,\sqsubseteq)$ together with an \emph{interpretation of the function symbols} $\interpret_0 \colon \Funcs \to \Mon(L)$ induces an AIC-structure carried by~$\ladomain^\omega$ with the following operations:  for any \(\varphi,\vartheta \in L^\omega\) and \(n \in \omega\),
	\allowdisplaybreaks%
	\begin{align}
		\label{def:top bot operations}
		\bot_n &\eeq \botL
		&\top_n &\eeq \topL\\
		\label{def:join meet operations}
		(\varphi \lmax \vartheta)_n &\eeq \varphi_n \sqcup \vartheta_n
		&(\varphi \lmin \vartheta)_n &\eeq \varphi_n \sqcap \vartheta_n
		\\
		\label{def:head shift operations}
		(\First\varphi)_n &\eeq \varphi_{0}
		&(\Next\varphi)_n &\eeq \varphi_{n + 1}
		\\
		\label{def:maj min operations}
		\vphantom{\sup_{k \ge n}^{\infty}}
		(\Sup\varphi)_n &\eeq \sup_{k \ge n}\, \varphi_k
		&(\Inf\varphi)_n &\eeq \inf_{k \ge n} \varphi_k
		\\
		\label{def:function operations}
		(\F\varphi)_n &\eeq \interpret_0(F)(\varphi_n)
		&(\Fs\varphi)_n &\eeq \interpret_0(F)^n(\varphi_n)
	\end{align}
An AIC-structure of this kind is called a \emph{sequence algebra}.%
\end{definition}%
\noindent%
To capture fixed point iterations, members of a sequence algebra are sequences. 
As we will demonstrate in \Cref{sec:case_studies}, the constants and algebraic operations of sequence algebras are ubiquitous in lattice fixed point (iteration) theorems.
Lines~(\ref{def:top bot operations}--\ref{def:join meet operations}) lift the constants and operations from the underlying lattice pointwise to sequences: $\bot$ and~$\top$ are the \emph{flat} sequences~$\constseq{\botL} \in \laseqs$ and $\constseq{\topL} \in \laseqs$.
The operation~$\lmax$ takes two sequences $\varphi$ and $\vartheta$ and yields their element-wise join 
$\morespacesequence{\varphi_0 \sqcup \vartheta_0}{\varphi_1 \sqcup \vartheta_1}{\varphi_2 \sqcup \vartheta_2}$; dually, $\lmin$ is the element-wise meet.

The two operations on Line \eqref{def:head shift operations} are called \emph{head} and \emph{shift}.
The \emph{head} operation $\Firstsymbol$ takes a sequence $\varphi = \sequence{\varphi_0}{\varphi_1}{\varphi_2}{\varphi_3}$ and yields a flat repetition of $\varphi$'s $0^{\textnormal{\scriptsize th}}$ element, i.e.\ $\constseq{\varphi_0}$.
The \emph{shift} operation $\Nextsymbol$ shifts sequences by one position to the left (while deleting the $0^{\textnormal{\scriptsize th}}$ element). 
Together, $\Firstsymbol$ and $\Nextsymbol$ allow extracting single sequence elements in the following sense: 
identifying a lattice element \(\ell \in L\) with its corresponding flat sequence \(\constseq{\ell} \in L^\omega\), 
the \(k^\textnormal{th}\) element of a sequence \(\varphi\) can be expressed as \( \First \Next^k \varphi =  \First \Next \,{}\stackrel{\vphantom{t}\mathclap{\smash{\textnormal{\tiny $k$ times}}}}{\cdots}{}\,{}\, \Next\, \varphi\).

The operations on Line~\eqref{def:maj min operations} are called \emph{majorum} and \emph{minorum}, respectively.
For the majorum~$\Sup \varphi$, the element at index \(n\) is the supremum of all (countably many) elements of $\varphi$ with index~$\geq n$.
It is clear that $\Sup \varphi$ is a \emph{majorant} of $\varphi$, i.e., $\varphi_n \sqsubseteq (\Sup \varphi)_n$ for all \(n\).
Moreover, the sequence $\Sup \varphi$ always \emph{descends}:
For instance, the supremum at index~$0$ can \enquote{see} all elements of $\varphi$, whereas at index~$1$ it can \enquote{see} all but the $0^{\textnormal{\scriptsize th}}$~element.
Hence, the latter can only ever yield a potentially smaller supremum.
In fact, $\Sup \varphi$ is the \emph{least decreasing majorant of~$\varphi$}.
Similar concepts have been studied for Banach function spaces also under the name \emph{least decreasing majorant} or \emph{essential supremum}~\cite{sinnamon2007monotonicity,SANTACRUZHIDALGO2024110490,kiwerski2024arithmeticinterpolationfactorizationamalgams}.
Dually, for the minorum~$\Inf \varphi$, the $n^{\textnormal{\scriptsize th}}$ element is the infimum of all elements with index $\geq n$.
$\Inf \varphi$~is a \emph{minorant} of $\varphi$, i.e.\ $(\Inf \varphi)_n \sqsubseteq \varphi_n$ and it is the \emph{greatest increasing minorant of~$\varphi$}.
\begin{wrapfigure}[8]{r}{.59\textwidth}
	\centering%
	\begin{adjustbox}{max width=.999\linewidth}%
	\begin{tikzpicture}[yscale=.5,xscale=1]
			\draw [gray,opacity=.65,ultra thick,dashed] plot [smooth, tension=0.6] coordinates{(-10, 6) (-8, 4) (-6, 9) (-4, 0) (-2, 11) (0, 1) (2, 9) (4, 3) (6, 8) (8, 5) (10, 7)};
			\draw [blue!, opacity=.3,line width=3pt] plot [smooth, tension=0.6] coordinates{(-10, 11) (-8, 11) (-6, 11) (-4, 11) (-2, 11) (0, 9) (2, 9) (4, 8) (6, 8) (8, 7) (10, 7)};
			\draw [orange, opacity=.3,line width=3pt] plot [smooth, tension=0.6] coordinates{(-10, 0) (-8, 0) (-6, 0) (-4, 0) (-2, 1) (0, 1) (2, 3) (4, 3) (6, 5) (8, 5) (10, 5.1)};
			\node [] (0) at (-11.7, 6) {\Huge\gray{${\varphi\colon}$}};
			\node [] (0) at (-12, 11) {\Huge\blue{${\Sup\varphi\colon}$}};
			\node [] (0) at (-12, 0) {\Huge\orange{${\Inf\varphi\colon}$}};
			\node [] (0) at (-10, -2.5) {\huge\lightgray{$0$}};
			\node [] (0) at (-8, -2.5) {\huge\lightgray{$1$}};
			\node [] (0) at (-6, -2.5) {\huge\lightgray{$2$}};
			\node [] (0) at (-4, -2.5) {\huge\lightgray{$3$}};
			\node [] (0) at (-2, -2.5) {\huge\lightgray{$4$}};
			\node [] (0) at (0, -2.5) {\huge\lightgray{$5$}};
			\node [] (0) at (2, -2.5) {\huge\lightgray{$6$}};
			\node [] (0) at (4, -2.5) {\huge\lightgray{$7$}};
			\node [] (0) at (6, -2.5) {\huge\lightgray{$8$}};
			\node [] (0) at (8, -2.5) {\huge\lightgray{$9$}};
			\node [] (0) at (10, -2.5) {\huge\lightgray{$\vphantom{9}\ldots$}};
			\node [style=majelem] (11) at (-10, 11) {\blue{$\Supsymbol$}};
			\node [style=majelem] (12) at (-8, 11) {\blue{$\Supsymbol$}};
			\node [style=majelem] (13) at (-6, 11) {\blue{$\Supsymbol$}};
			\node [style=majelem] (14) at (-4, 11) {\blue{$\Supsymbol$}};
			\node [style=majelem] (15) at (-2, 11) {\blue{$\Supsymbol$}};
			\node [style=majelem] (16) at (0, 9) {\blue{$\Supsymbol$}};
			\node [style=majelem] (17) at (2, 9) {\blue{$\Supsymbol$}};
			\node [style=majelem] (18) at (4, 8) {\blue{$\Supsymbol$}};
			\node [style=majelem] (19) at (6, 8) {\blue{$\Supsymbol$}};
			\node [style=majelem] (20) at (8, 7) {\blue{$\Supsymbol$}};
			\node [style=majelem] (21) at (10, 7) {\blue{$\Supsymbol$}};
			\node [style=minelem] (22) at (-10, 0) {\orange{$\Infsymbol$}};
			\node [style=minelem] (23) at (-8, 0) {\orange{$\Infsymbol$}};
			\node [style=minelem] (24) at (-6, 0) {\orange{$\Infsymbol$}};
			\node [style=minelem] (25) at (-4, 0) {\orange{$\Infsymbol$}};
			\node [style=minelem] (26) at (-2, 1) {\orange{$\Infsymbol$}};
			\node [style=minelem] (27) at (0, 1) {\orange{$\Infsymbol$}};
			\node [style=minelem] (28) at (2, 3) {\orange{$\Infsymbol$}};
			\node [style=minelem] (29) at (4, 3) {\orange{$\Infsymbol$}};
			\node [style=minelem] (30) at (6, 5) {\orange{$\Infsymbol$}};
			\node [style=minelem] (31) at (8, 5) {\orange{$\Infsymbol$}};
			\node [style=aelem] (0) at (-10, 6) {\gray{$\bullet$}};
			\node [style=aelem] (1) at (-8, 4) {\gray{$\bullet$}};
			\node [style=aelem] (2) at (-6, 9) {\gray{$\bullet$}};
			\node [style=aelem] (3) at (-4, 0) {\gray{$\bullet$}};
			\node [style=aelem] (4) at (-2, 11.1) {\gray{$\bullet$}};
			\node [style=aelem] (5) at (0, 1) {\gray{$\bullet$}};
			\node [style=aelem] (6) at (2, 9.1) {\gray{$\bullet$}};
			\node [style=aelem] (7) at (4, 3) {\gray{$\bullet$}};
			\node [style=aelem] (8) at (6, 8.1) {\gray{$\bullet$}};
			\node [style=aelem] (9) at (8, 5) {\gray{$\bullet$}};
			\node [style=aelem] (10) at (10, 7.1) {\gray{$\bullet$}};
	\end{tikzpicture}%
	\end{adjustbox}%
\end{wrapfigure}%
A depiction of $\Sup \varphi$ and $\Inf \varphi$ is provided on the right.
Gray~\gray{$\bullet$}'s depict elements of ${\varphi}$, blue~\blue{${\Supsymbol}$}'s elements of \blue{${\Sup \varphi}$}, and orange~\orange{${\Infsymbol}$}'s elements of~\orange{${\Inf \varphi}$}.
	Continuous lines connecting elements are for better visualization. 
	All sequences are discrete.
	
Here, we can see visually that $\Sup \varphi$ is a \emph{descending majorant} of~$\varphi$ whereas $\Inf \varphi$ is an \emph{ascending minorant}. 
We can also see how $\Sup \varphi$ and $\Inf \varphi$ together form the \emph{tightest monotonic envelope} around $\varphi$.
Finally, we reiterate that $\Sup \varphi$ and $\Inf \varphi$ are \emph{not} simply the global supremum and infimum of all the elements of $\varphi$.
Global infimum and supremum of $\varphi$ are expressible as $\First \Inf \varphi$ and $\First \Sup \varphi$.

The operations in Line~\eqref{def:function operations} are the \emph{function application} and \emph{orbit} operations.
For simplicity, let us omit \(\interpret_0\) from the notation and treat \(F\) (as opposed to $\interpret_0(\funca)$) as a monotonic endomap on \(L\).
Then $\Fa \varphi$ is the element-wise application of $\funca$ to $\varphi$, i.e.\ it yields the sequence\footnote{More accurately, the sequence $\morespacesequence{\interpret_0(\funca)(\varphi_0)}{\interpret_0(\funca)(\varphi_1)}{\interpret_0(\funca)(\varphi_2)}{\interpret_0(\funca)(\varphi_3)}$~.}%
\begin{align*}
	\Fa \varphi \qeq \morespacesequence{\funca(\varphi_0)}{\funca(\varphi_1)}{\funca(\varphi_2)}{\funca(\varphi_3)}~.
\end{align*}%
The unary operation \(\funca^*\) yields the \emph{iteration} of $\funca$ on a lattice element.
However, the input to $\funca^*$ is already a sequence $\varphi$ and so we need also $\funca^*$ to operate on sequences.
For this, we apply~\(\funca\) zero times to the $0^{\textnormal{\scriptsize th}}$ element of $\varphi$, once to the $1^{\textnormal{\scriptsize st}}$ element, twice to the $2^{\textnormal{\scriptsize nd}}$ element, and so forth, thus yielding the sequence\footnote{%
Again, more accurately, the sequence $\morespacesequence{\varphi_0}{\interpret_0(\funca)(\varphi_1)}{\interpret_0(\funca)^2(\varphi_2)}{\interpret_0(\funca)^3(\varphi_3)}$~.}%
\begin{align*}
	\Fas \varphi \qeq \morespacesequence{\varphi_0}{\funca(\varphi_1)}{\funca^2(\varphi_2)}{\funca^3(\varphi_3)}~.
\end{align*}%
When we apply $\funca^*$ to a flat sequence $\overline{\ell} = \constseq{\ell}$
then $\Fas \overline{\ell}$ is indeed the $\emph{$\funca$-orbit}$ of $\ell \in \ladomain$, i.e.\ the sequence $\Fas \overline{\ell} \eeq \morespacesequence{\ell}{\funca(\ell)}{\funca^2(\ell)}{\funca^3(\ell)}$.

\subsection{Iterative Construction Terms}
\label{sec:syntax-semantics-terms}

Sequence algebras are where iterative constructions of lattice fixed points live.
By combining the operations of a sequence algebra, one can obtain a variety of different recipes for constructing fixed points.
In fact, we will see in \Cref{sec:case_studies} that many fixed point constructions can be carried out using $\Supsymbol$, $\Infsymbol$, and $\funca^*$ alone.
For example, the least fixed point of $\funca$ can (under certain continuity assumptions on $\funca$) be expressed as $\Sup \Fas \bot$, which should look familiar when recalling the construction $\sup_n \funca^n(\bot)$ from the Kleene fixed point theorem.

Our ultimate goal is to be able to algebraically manipulate these different recipes for constructing lattice fixed points and thereby carry out algebraic proofs of fixed point theorems.
This requires a syntax for writing these recipes down, as well as a set of inference rules with which we can prove identities and other algebraic laws.

\begin{definition}[Syntax of Iterative Construction Terms]\label{def:il-syntax}
	For a fixed set \(\Vars = \{\fora,\, \forb,\,{\dots}\}\) of \emph{variables}, the set $\Term$ of \emph{iterative construction terms} (or simply, \emph{terms}), ranging over metavariables $\terma,\, \termb,\, \ldots{}$, adheres to the grammar below, where $\fora \in \Vars$ and $\funca \in \Funcs$:\footnote{%
		Parenthesizing is handled in the usual manner: unary symbols ($\Nextsymbol$, $\Supsymbol$, \dots) bind stronger than binary ones \mbox{($\lmax$, $\lmin$)}, and unary symbols associate to the right, e.g.\ $\Sup \Fa \terma = \Sup (\Fa \terma)$.%
	}%
	\begin{align*}
		\hspace*{-1em}\terma \semigray{\quad\longrightarrow\quad}  \bot   
		&\qmid \top 
		\qmid \fora 
		\qmid \terma \lmax \terma 
		\qmid \terma \lmin \terma  
		\qmid \First \terma  
		\qmid \Next \terma  
		\qmid \Sup \terma  
		\qmid \Inf \terma  
		\qmid \Fa \terma  
		\qmid \Fas \terma 
	\end{align*}
\end{definition}%
\noindent%
Every term $\terma$ represents a recipe for constructing new sequences from known ones. 
The known sequences can well be left arbitrary/uninterpreted and terms hence feature variables $\fora,\, \forb,\, \forc,\, \ldots$ which have to be interpreted to obtain a concrete sequence.%
\begin{definition}[Interpretations]
	\label{def:interpretation}
	Given an AIC-structure \(\Alg\) carried by \(A\), an \emph{interpretation} (in~\(\Alg\)) is a function \(\interpret \colon \Vars \to \Alg\)
	 that maps variables to elements of \(A\).
	We say that \(\interpret\) is a \emph{standard interpretation} if \(\Alg\) is a sequence algebra.%
\end{definition}%
\noindent%
Given a sequence algebra \(\Alg\) carried by \(L^\omega\) and a standard interpretation \(\interpret \colon \Vars \to L^\omega\), each term \(\terma\) denotes the sequence obtained by recursively evaluating its operations after replacing each occurrence of a variable \(\fora\) in \(\terma\) with \(\interpret(\fora)\).
Formally, the set \(\Term\) of iterative construction terms itself carries an AIC-structure called the \emph{term algebra} \(\AlgTerm\), whose operations are given by term formation (e.g.\ $\Nextsymbol^{\semigray{\AlgTerm}} \fora = \Next \fora$), and every interpretation \(\interpret \colon \Vars \to \Alg\)  extends to a unique algebra homomorphism \(\interpretsem{-} \colon \AlgTerm \to \Alg\). For instance, $\interpretsem{\Next \fora}  = \interpretsem{\Nextsymbol^{\semigray{\AlgTerm}} \fora} = \NextAlg \interpretsem{\fora} = \NextAlg \interpret(\fora)$.
This unique homomorphism yields the semantics of terms under a given interpretation.%
\begin{definition}[Semantics of Terms]
	\label{def:term-semantics}
	Let $\Alg$ be an AIC-structure, let $\interpret$ be an interpretation in $\Alg$, and let \(\interpretsem{-}\colon\AlgTerm \to \Alg\) be the unique algebra homomorphism satisfying \(\interpretsem{\fora} = \interpret(\fora)\) for every \(\fora \in \Vars\).
	Then for each term \(\terma \in \Terms\), we call \(\interpretsem{\terma}\) the \emph{\(\interpret\)-semantics} of \(\terma\).%
\end{definition}%
\noindent%
%
%
%
\renewcommand{\arraystretch}{1.5}%
\begin{table*}[t]%
	\caption{%
		\label{tab:semantics}%
		Semantics of iterative construction terms under a standard interpretation $\interpret$. 
	}%
	%
	\centering
	\begin{adjustbox}{max width=.999\textwidth}%
		\small%
		\begin{tabular}{r@{\qquad}l@{\quad}l@{\quad}l@{\quad}l}
			\toprule
			\textbf{Term} $\boldsymbol{\termc}$	& \textbf{Semantics} ${\nsemin{\boldsymbol{\termc}}}$	& \textbf{Name}	& \textbf{Remarks} & \textbf{Intuition}\\[.25em]
			\midrule
			$\bot$				& $\bot$									& bottom				& constant sequence of $\botL$'s										& $\sequence{\bot}{\bot}{\bot}{\bot}$\\
			$\top$				& $\top$									& top 				& constant sequence of $\topL$'s										& $\sequence{\top}{\top}{\top}{\top}$\\
			$\fora$		& $\interpret(\fora)_n$								& variable				& the \emph{sequence} $\interpret(\fora)$									& $\sequence{\fora_0}{\fora_1}{\fora_2}{\fora_3}$\\
			$\terma \lmax \termb$		& $\nsemin{\terma} \ssqcup \nsemin{\termb}$			& join				& element-wise join of $\terma$ and $\termb$									& $\mmorespacesequence{\terma_0 {\lmax} \termb_0}{\terma_1 {\lmax} \termb_1}{\terma_2 {\lmax} \termb_2}{\terma_3 {\lmax} \termb_3}$\\
			$\terma \lmin \termb$		& $\nsemin{\terma} \ssqcap \nsemin{\termb}$			& meet				& element-wise meet of $\terma$ and $\termb$								& $\mmorespacesequence{\terma_0 {\lmin} \termb_0}{\terma_1 {\lmin} \termb_1}{\terma_2 {\lmin} \termb_2}{\terma_3 {\lmin} \termb_3}$\\
			$\Sup \terma$			& $\displaystyle \sup_{n \leq k}~ \nsemi{\terma}{k}$	& majorum			&\begin{tabular}{@{}l}$n^{\textnormal{\scriptsize th}}$ element is the supremum over\\[-.75em] 
																								all elements with index $\geq n$\end{tabular}							& $\displaystyle\mmorespacesequence{\sup_{0 \leq k}\, \terma_k}{\sup_{1 \leq k}\, \terma_k}{\sup_{2 \leq k}\, \terma_k}{\sup_{3 \leq k}\, \terma_k}$\\
			$\Inf \terma$			& $\displaystyle \inf_{n \leq k}\, \nsemi{\terma}{k}$	& minorum			& \begin{tabular}{@{}l}$n^{\textnormal{\scriptsize th}}$ element is the infimum over\\[-.75em] 
																								all elements with index $\geq n$\end{tabular} 							& $\displaystyle\mmorespacesequence{\inf_{0 \leq k}\, \terma_k}{\inf_{1 \leq k}\, \terma_k}{\inf_{2 \leq k}\, \terma_k}{\inf_{3 \leq k}\, \terma_k}$\\
			$\First \terma$			& $\nsemi{\terma}{0}$							& head				& repeat $0^{\textnormal{\scriptsize th}}$ element							& $\sequence{\terma_0}{\terma_0}{\terma_0}{\terma_0}$\\
			$\Next \terma$			& $\nsemi{\terma}{n+1}$						& shift				& \begin{tabular}{@{}l}drop $0^{\textnormal{\scriptsize th}}$ element and left-shift all\\[-.75em] 
																								other elements by 1 index \end{tabular} 								& $\sequence{\terma_1}{\terma_2}{\terma_3}{\terma_4}$\\
			$\Fa \terma$			& $\interpret(\fcta)\left(\nsemin{\terma} \right)$		& \begin{tabular}{@{}l}function\\[-.75em]application\end{tabular}	& apply $\interpret(\fcta)$ element-wise to $\terma$							& $\morespacesequence{F(\terma_0)}{F(\terma_1)}{F(\terma_2)}{F(\terma_3)}$ \\
			$\Fas \terma$			& $\interpret(\fcta)^n\left(\nsemin{\terma} \right)$		& orbit				& \begin{tabular}{@{}l}$n^{\textnormal{\scriptsize th}}$ element is the $n$-fold iteration\\[-.75em] 
																								of $\interpret(F)$ on the $n^{\textnormal{\scriptsize th}}$ element of $\terma$ \end{tabular} & $\morespacesequence{\terma_0}{F(\terma_1)}{F^2(\terma_2)}{F^3(\terma_3)}$
			\\[.25em]
			\bottomrule
		\end{tabular}%
	\end{adjustbox}%
\end{table*}
While \Cref{def:interpretation,def:term-semantics} allow interpretations (and hence semantics) in arbitrary AIC-structures, we are primarily trying to capture the sequence algebras.
Hence, we call interpretations in sequence algebras \emph{standard interpretations}.
We record the semantics of iterative construction terms under a standard interpretation in \Cref{tab:semantics}.
As we will see, standard interpretations capture many important aspects of iterative fixed point constructions.%

\begin{remark}[On Interpretations of Variables and Functions]
	We typically bundle the interpretation of the function symbols \(\interpret_0\) in with the standard interpretation, and simply write \(\interpret(F)\) instead of \(\interpret_0(F)\).
	While an abuse of notation, we could have equivalently defined a standard interpretation to be a suitable function \(\interpret \colon \Vars \cup \Funcs \to L^\omega \cup \Mon(L)\).
	\lipicsEnd
\end{remark}

\subsection{Quasiequations}

For two terms $\terma$ and $\termb$, we call a formal expression $\terma = \termb$ an \emph{identity}.
For instance, $\Fa \fora = \fora$ is an identity (expressing that $\fora$ is a fixed point of $\funca$),\footnote{More precisely, (standard) interpreted in a sequence algebra, this states that $\fora_n$ is a fixed point of $\funca$ for all \(n \in \omega\). This is because \(\Fa \fora\) is \(\funca\) applied pointwise to \(\fora\).} even though the terms $\Fa \fora$ and $\fora$ are clearly not (syntactically) equal.
Fixed point theorems usually require certain assumptions, so that $\fora$ is indeed a fixed point. 
We therefore want to be able to derive identities that hold, provided that other identities (the assumptions) hold.
This is captured by \emph{quasiequations} (a.k.a.~\emph{quasi-identities}).%
\begin{definition}[Quasiequations \& Axiom Systems]
\label{def:quasi-equations}
	A \emph{quasiequation} $\quasieq$ is a \mbox{formula of the form}
	\begin{align*}
		\textstyle\left(\bigwedge_{i \in I} \leftsuperscript{i}{\terma}{} \eeq \leftsuperscript{i}{\termb}{} \right) 
		\qmorespace{\implies} 
		\terma \eeq \termb
		\tag{$\quasieq$}
	\end{align*} 
	for some index set \(I\) and iterative construction terms \(\terma,\termb, \leftsuperscript{i}{\terma}{}, \leftsuperscript{i}{\termb}{} \in \Terms\), for each \(i \in I\).
	The identities $\leftsuperscript{i}{\terma}{} = \leftsuperscript{i}{\termb}{}$ are called the \emph{premises} and $\terma = \termb$ is called the \emph{conclusion} of \quasieq.
	The quasiequation~\quasieq is called \emph{finitary} if it has only finitely many premises (i.e.\ if \,$I$ is finite).

	An AIC-structure \(\Alg\) carried by $A$ \emph{satisfies} \quasieq, denoted $\Alg \models \quasieq$, if for any interpretation $\interpret \colon \Vars \to A$ the following holds:
	\begin{align*}
			\textnormal{if}		
			\qquad
			\textnormal{for all }i \in I\colon ~~ \interpretsem{\leftsuperscript{i}{\terma}{}} \eeq \interpretsem{\leftsuperscript{i}{\termb}{}}
			\qquad
			\textnormal{then}	
			\qquad
			\interpretsem{\terma} \eeq \interpretsem{\termb}.
	\end{align*}
	A quasiequation is \emph{valid (for sequence algebras)} if it is satisfied by every \underline{se}q\underline{uence al}g\underline{ebra}.

	An \emph{axiom system} \(\Ax\) is a set of quasiequations and its elements are called the \emph{axioms} (of \(\Ax\)). We say that $\Ax$ is \emph{finite} if it is a finite set, and \emph{finitary} if all axioms are finitary.
	An AIC-structure \(\Alg\) \emph{satisfies} \(\Ax\), written \(\Alg \models \Ax\), if \(\Alg \models \quasieq\) for every \(\quasieq \in \Ax\). 
	In this case, $\Alg$ is a \emph{model} of $\Ax$.
	An axiom system $\Ax$ is \emph{sound (for sequence algebras)} if every axiom is valid for sequence algebras.
\end{definition}%
We prefer to write quasiequations like $\quasieq$ from \Cref{def:quasi-equations} as inference rules of the form%
\begin{align*}
	\nirule{q}{~\leftsuperscript{1}{\terma}{} \eeq \leftsuperscript{1}{\termb}{} \qquad \leftsuperscript{2}{\terma}{} \eeq \leftsuperscript{2}{\termb}{} \qquad {\cdots}~}{\terma \eeq \termb}
	\quad\qquad{\textnormal{\raisebox{0.75em}{or}}}\qquad\quad
	\irule{\textnormal{\TirNameStyle{\raisebox{-.5ex}{\semigray{\quasieq}}}}}{\terma \eeq \termb}
\end{align*}%
where \(I = \{1, 2, \dots\}\) in case of the left inference rule and \(I = \emptyset\) in case of the right one.
For $I = \emptyset$, the notion of quasiequation and identity collapse into a single notion and so we can also speak of satisfaction and validity of identities.

Quasiequations are, in a formal sense~\cite{kozenKA}, \emph{algebraic} laws, which means that satisfaction is preserved by substitution:
Let $\terma, \terme \in \Terms$, $\fora \in \Vars$, and let \(\terma \subst{\fora}{\fore}\) be the term obtained by replacing each occurrence of \(\fora\) in \(\terma\) with \(\terme\).
Then algebraicity means that if \(\Alg \models (\bigwedge_{i \in I} \leftsuperscript{i}{\terma}{} = \leftsuperscript{i}{\termb}{}) \implies \terma = \termb\), then  we also know that \(\Alg \models (\bigwedge_{i \in I} \leftsuperscript{i}{\terma}{}\subst{\fora}{\terme} = \leftsuperscript{i}{\termb}{}\subst{\fora}{\terme}) \implies \terma\subst{\fora}{\terme} = \termb\subst{\fora}{\terme}\).
The latter quasiequation is a \emph{substitution instance} of the former. 
Substitution instances of quasiequations can be obtained from any substitution of variables with terms: 
given a mapping \(\sigma \colon \Vars \to \Terms\) and a term \(\terma\), we write \(\terma[\sigma]\) for the term obtained by simultaneously substituting each variable \(\fora\) in term $\terma$ with term \(\sigma(\fora)\).%
\footnote{
	In other words, \(\terma[\sigma] = {^\sigma}\semantics{s}\); cf.\ \Cref{sec:syntax-semantics-terms}.
}%

\begin{figure}
	\centering%
	\begin{adjustbox}{max width=.9999\textwidth}%
		\begin{minipage}{1.235\linewidth}%
			\begin{align*}
				\nequalaxiom{eq-reflex}{\fora}{\fora}
				\qquad
				\nequalrule{symm}{\fora}{\forb}{\forb}{\fora}
				\qquad
				\nirule{eq-trans}{\fora \eeq \forb \\ \forb \eeq \forc}{\fora \eeq \forc}
				\qquad
				\nirulet{cong}{
					\leftsuperscript{1}{\fora} \eeq \leftsuperscript{1}{\forb} 
					\\ 
					{\cdots} 
					\\ 
					\leftsuperscript{n}{\fora} \eeq \leftsuperscript{n}{\forb}
				}{
					{\boxast}\left(\leftsuperscript{1}{\fora},\, \ldots,\! \leftsuperscript{n}{\fora}\right) 
					\eeq
					{\boxast}\left(\leftsuperscript{1}{\forb},\, \ldots,\! \leftsuperscript{n}{\forb}\right) 
				}
			\end{align*}%
		\end{minipage}
	\end{adjustbox}%
	\caption{
		\label{fig:equational logic}
		The inference rules of equational logic.
		Above, ${\fora, \forb, \leftsuperscript{i}{\fora}, \leftsuperscript{i}{\forb}, \forc \in \Vars}$, ${\boxast}$ ranges over the operations of arity ${n>0}$ appearing in the signature of the algebra under consideration.
		In our case ${\boxast \in \{{\lmax},\, {\lmin},\, \Firstsymbol,\, \Nextsymbol,\, \Supsymbol,\, \ldots\}}$, so ${n \in \{1,\, 2\}}$.
	}
\end{figure}%
Quasiequational logic lets us derive new quasiequations syntactically from known ones using \emph{derivations} -- a form of inference tree in the following sense:
a \emph{(finite) inference tree} is a (finite) tree whose nodes are labelled with identities (i.e.\ expressions of the form $\terma = \termb$ for two terms $\terma$ and $\termb$).
In an inference tree, branching points are interpreted as inference rules:%

\vspace*{-.5\baselineskip}%
\footnotesize%
\belowdisplayskip=0pt%
\begin{align*}
	\begin{gathered}
		\irule{\leftsuperscript{1}{\terma}{} \eeq \leftsuperscript{1}{\termb}{} \\{\!\!}\cdots{\!\!}\\ \leftsuperscript{n}{\terma}{} \eeq \leftsuperscript{n}{\termb}{}}{\terma = \termb}
	\end{gathered}
	\quad
	\qquad\textnormal{\enquote{$=$}}
	\qquad
	\begin{gathered}
		\begin{tikzpicture}[yscale=0.8]
			\node (root) at (0,0) {\(\terma = \termb\)};
			\node (l) at (-1.5,1) {\(\leftsuperscript{1}{\terma}{} \eeq \leftsuperscript{1}{\termb}{}\)};
			\node[lightgray!50!gray] (dots) at (0,1) {\(\vphantom{b^1}{}\cdots{}\)};
			\node (r) at (1.5,1) {\(\leftsuperscript{n}{\terma}{} \eeq \leftsuperscript{n}{\termb}{}\)};
			\draw (root) -- (l);
			\draw (root) -- (r);
			\draw[lightgray!50!gray] (root) -- (-.15, 0.6);
			\draw[lightgray!50!gray] (root) -- (.15, 0.6);
			\draw[lightgray!50!gray] (root) -- (-.5, 0.6);
			\draw[lightgray!50!gray] (root) -- (.5, 0.6);
		\end{tikzpicture}
	\end{gathered}
\end{align*}%
\normalsize%
\begin{definition}[Derivations \& Provability]
	\label{def:finite provability}
	Let $\Ax$ be an axiom system, and let \(\quasieq\) be a finitary quasiequation. A \emph{derivation of $\quasieq$ from $\Ax$} is a finite inference tree such that the root is labelled with the conclusion of \(\quasieq\), and moreover every node is either (i) a leaf labelled with some premise of $\quasieq$, or (ii) a point of branching of the form 
	\begin{equation*}
		\label{eq:substitution}
		\nirule{}{
				\leftsuperscript{1}{\terma}{}[\sigma] \eeq \leftsuperscript{1}{\termb}{}[\sigma] 
				\\
				{\cdots}
				\\
				\leftsuperscript{n}{\terma}{}[\sigma] \eeq \leftsuperscript{n}{\termb}{}[\sigma]
			}{	\terma[\sigma] \eeq \termb[\sigma]	}
	\end{equation*}
for some quasiequation
\((\bigwedge_{i=1}^n \leftsuperscript{i}{\terma}{} = \leftsuperscript{i}{\termb}{}) \implies \terma = \termb \) in $\Ax$ or in equational logic (see \Cref{fig:equational logic}), and some substitution \(\sigma \colon \Vars \to \Terms\).
	We say that 
	\quasieq 
	is \emph{derivable (\emph{or} provable)} from \(\Ax\) and write 
	$\Ax \vdash \quasieq$ 
	if \quasieq has a derivation from \(\Ax\).%
\end{definition}%
\begin{remark}
	\begin{enumerate}
		\item 
			We stress that \Cref{def:finite provability} (ii) may apply to leaves (where $n=0$); for example any instance of using $\TirName{eq-reflex}$. 

		\item 
			An axiom system $\Ax$ is sound for sequence algebras iff all quasiequations derivable from $\Ax$ are valid for sequence algebras. 
			This is because our proof system is standard quasiequational logic, which cannot derive non-valid quasiequations from valid ones.

		\item 
			If $\Ax\vdash \quasieq$, we may add \quasieq to $\Ax$ without increasing the expressiveness of $\Ax$, i.e.\ if both $\Ax \vdash \quasieq$ and $\Ax\cup \{\quasieq\} \vdash \quasieq'$, 
			then already $\Ax\vdash \quasieq'$. 
			Indeed, whenever we would use \quasieq within a derivation of $\quasieq'$ from $\Ax\cup \{\quasieq\}$, we might as well \enquote{inline} a derivation of a suitable substitution instance of \quasieq from $\Ax$ instead. 
			It \emph{can} make a difference for automatic proof search, however, 
			whether $\quasieq'$ is present in the axiom system or not, cf.\ \Cref{sec:isabelle}.
		\lipicsEnd
	\end{enumerate}
\end{remark}

\subsection{Quasi-Inequalities}

Lattice-theoretic fixed point theorems often refer to the underlying partial order $\sqsubseteq$ of the lattice.
For instance, \emph{Park induction} states that $\funca(\ell) \sqsubseteq \ell$ implies $\lfp \funca \sqsubseteq \ell$.
Since $\sqsubseteq$ is so prevalent in lattice-theoretic fixed point theorems, we want to be able to make judgements directly about $\sqsubseteq$, thus rendering theorem statements and proofs more readable and potentially more automatable.
Reasoning about~$\sqsubseteq$ is not more general than reasoning about $=$, because in any lattice 
we have 
that \(\ell \sqsubseteq m\) if and only if \(\ell \sqcup m = m \) (or equivalently $\ell = m \sqcap \ell$). 

As every sequence algebra is a lattice with meet \(\lmin\) and join~\(\lmax\)\!, also here we can recover the underlying partial order via the meet and join operations: 
\(\varphi_n \sqsubseteq \vartheta_n\) for all \(n \in \omega\) if and only if \(\varphi \lmax \vartheta = \vartheta\).
This allows us to treat point-wise inequalities between sequences from \(L\), as equations and thus appeal to the correctness of equational logic reasoning.

While the reasoning under the hood is equational, our aim is still to reason about inequalities. 
We thus introduce the shorthand notation \(\terma \preceq \termb\) for the identity \(\terma \lmax \termb = \termb\).%
\begin{definition}[Quasi-Inequalities]
	\label{def:comparant}
	An \emph{inequality} is an identity of the form $\terma \lmax \termb = \termb$, denoted \(\terma \preceq \termb\), for \(\terma, \termb \in \Terms\).
	A quasiequation \(\bigl(\bigwedge_{i \in I} \leftsuperscript{i}{\terma} \preceq \leftsuperscript{i}{\termb}\bigr) \implies \terma \preceq \termb\) is called a \emph{quasi-inequality}.%
\end{definition}%

\begin{wrapfigure}[9]{l}{.666\linewidth}%
\vspace*{-.75\intextsep}%
\begin{minipage}{.999\linewidth}%
\begin{whitebox}%
\begin{adjustbox}{max width=.999\linewidth}%
\renewcommand{\arraystretch}{1.125}%
\begin{tabular}{r@{\hspace{1.5em}}l}
	$\fora$ is an ascending chain						& $\fora \tss \Next \fora$						\\
	$\fora$ is a descending chain						& $\Next \fora \tss \fora$						\\[.25em]
	$\fora$ is flat									& $\Sup \fora \tss \Inf \fora$					\\
	$\fora$ converges								& $\Inf \Sup \fora \tss \Sup \Inf \fora$					\\[.25em]
	all elements of $\fora$ are prefixed points of $\funca$	& $\Fa \fora \tss \fora$							\\
	all elements of $\fora$ are postfixed points of $\funca$	& $\fora \tss \Fa \fora$							\\[.25em]
\end{tabular}%
\renewcommand{\arraystretch}{1.5}%
\end{adjustbox}%
\end{whitebox}%
\end{minipage}%
\end{wrapfigure}%
\noindent%
An inequality \(\terma \preceq \termb\) is valid for sequence algebras iff for any standard interpretation \(\interpret\) and any \(n \in \omega\), we have \(\interpretsem{\terma}_n \sqsubseteq \interpretsem{\termb}_n\).
It is direct from the definition of lattice that \(\terma = \termb\) is valid if and only if both \(\terma \preceq \termb\) and \(\termb \preceq \terma\) are valid.
A number of interesting properties of sequences can be expressed as inequalities, see the box above left.%
\begin{wrapfigure}[5]{r}{.6\textwidth}%
\vspace*{-1.25\intextsep}%
\begin{minipage}{.599\textwidth}%
\begin{adjustbox}{max width=.999\linewidth}%
\begin{tikzpicture}
		\node [style=seqelem] (0) at (-7, 4.5) {$a_0$};
		\node [style=seqelem] (1) at (-5, 4.5) {$a_1$};
		\node [style=seqelem] (2) at (-3, 4.5) {$a_2$};
		\node [style=seqelem] (3) at (-1, 4.5) {$a_3$};
		\node [style=seqelem] (13) at (1, 4.5) {$\ldots$};
		%
		\node [style=seqelem] (5) at (-7, 6) {$a_1$};
		\node [style=seqelem] (6) at (-5, 6) {$a_2$};
		\node [style=seqelem] (7) at (-3, 6) {$a_3$};
		\node [style=seqelem] (8) at (-1, 6) {$\ldots$};
		
		\node [style=seqelem] (11) at (-7.8, 4.5) {$\boldsymbol{\fora\colon}$};
		\node [style=seqelem] (12) at (-8, 6) {$\boldsymbol{\Next \fora\colon}$};
		\draw [style=blue edge] (0) to (5);
		\draw [style=blue edge] (1) to (6);
		\draw [style=blue edge] (2) to (7);
		\draw [style=blue edge, dashed] (3) to (8);
		\draw [style=blue double edge] (5) to (1);
		\draw [style=blue double edge] (6) to (2);
		\draw [style=blue double edge] (7) to (3);
		\draw [style=orange edge] (0) to (1);
		\draw [style=orange edge] (1) to (2);
		\draw [style=orange edge] (2) to (3);
		\draw [style=orange edge, dashed] (3) to (13);
\end{tikzpicture}%
\end{adjustbox}%
\end{minipage}%
\end{wrapfigure}%

For ascending chains, the situation is depicted on the right:
An arrow $a_n \rightarrow a_k$ indicates that $\sem{a}{n} \laord \sem{a}{k}$.
Vertical blue arrows are induced by $\fora \preceq \Next \fora$, diagonal blue double arrows by equality.
Horizontal orange arrows are induced by transitivity of the blue arrows and show that $\fora$ is ascending.
%
%
%
%

\section{Axioms for Iterative Constructions}
\label{sec:calc}

We now present two axiom systems, $\AIC_0$ and $\AIC_1$, for the algebra of iterative constructions.
$\AIC_0$ is a \emph{basic} set of axioms that is sound for sequence algebras.
It is intended to be small, so that it is easy to prove whether or not some class of structures forms models of AIC.

%
%
%
\renewcommand{\arraystretch}{3}%
\begin{figure*}%
	%
	\centering%
	\begin{adjustbox}{max width=.9999\textwidth, max totalheight=.9\textheight}%
		\normalsize%
		\centering%
		\begin{minipage}{1.0\linewidth}
		\centering%
			\begin{tabular}{rlcrl}
				\multicolumn{2}{c}{\hspace{0em}\textbf{\large\textsf{Bounded Lattice}}} 
				&&
				\multicolumn{2}{c}{\hspace{0em}\textbf{\large\textsf{Majora \& Minora}}} 
				\\[-.75em]
				$\axbot{\fora}$
				&
				$\axtop{\fora}$	
				&\quad&
				$\axsupinflate{\fora}$
				&
				$\axinfdeflate{\fora}$
				\\
				$\axjoincomm{\fora}{\forb}$
				&
				$\axmeetcomm{\fora}{\forb}$
				&&
				$\axsupidem{\fora}$
				&
				$\axinfidem{\fora}$
				\\
				$\axjoinabsorb{\fora}{\forb}$
				&
				$\axmeetabsorb{\fora}{\forb}$
				&&
				$\supmonorulet{\fora \tss \forb}{\Sup \fora \tss \Sup \forb}$\hspace*{2.5em}
				&
				\hspace*{2.5em}$\infmonorulet{\fora \tss \forb}{\Inf \fora \tss \Inf \forb}$
				\\
				\multicolumn{2}{c}{\hspace*{0em}$\axjoinassoc{\fora}{\forb}{\forc}$}
				&&
				$\axnextsupcomm{\fora}$
				&
				$\axnextinfcomm{\fora}$
				\\
				\multicolumn{2}{c}{\hspace*{0em}$\axmeetassoc{\fora}{\forb}{\forc}$}
				&&
				${\mprset{fraction={===}}%
				\nextindrulet{\Next \fora \tss \fora}{\Sup \fora \tss \fora}}$\hspace*{2em}
				&
				${\mprset{fraction={===}}%
				\hspace*{2em}\nextcoindrulet{\fora \tss \Next \fora}{\fora \tss \Inf \fora}}$
				\\[.5em]
				\multicolumn{2}{c}{\hspace{0em}\textbf{\large\textsf{Shifts}}}
				&&
				\multicolumn{2}{c}{\hspace{0em}\textbf{\large\textsf{Function Applications \& Orbits}}}  
				\\[-.75em]
				\multicolumn{2}{c}{\hspace{0em}$\nextmonorulet{\fora \tss \forb}{\Next \fora \tss \Next \forb}$}%
				&&
				$\Fmonorulet{\fora \tss \forb}{\Fa \fora \tss \Fa \forb}$%
				\hspace*{2.5em}&\hspace*{2.5em}
				$\Fsmonorulet{\fora \tss \forb}{\Fas \fora \tss \Fas \forb}$
				\\
				$\axnextbot{}$
				&
				$\axnexttop{}$
				&&
				$\axFNextcomm{\fora}$
				&
				$\axFFscomm{\fora}$
				\\
				\multicolumn{2}{c}{\hspace{0em}$\axnextoverjoin{\fora}{\forb}$}
				&&
				\multicolumn{2}{c}{\hspace{0em}$\axiter{\fora}$}
				\\
				\multicolumn{2}{c}{\hspace{0em}$\axnextovermeet{\fora}{\forb}$}
				&&
				$\Findrulet{\Fa \fora \tss \fora}{\Fas \fora \tss \fora}$\hspace*{2em}
				&
				\hspace*{2em}$\Fcoindrulet{\fora \tss \Fa \fora}{\fora \tss \Fas \fora}$
			\end{tabular}
				\lightgray{\hrule}
		\end{minipage}
	\end{adjustbox}%
	\caption{%
		\label{fig:axioms-aic0}%
		The basic axiom system ${\AIC_0}$.
		Above, ${\fora, \forb, \forc}$ are fixed variables, and ${\fcta}$ ranges over all function symbols.}%
		\vspace*{-.1cm}%
\end{figure*}%
\renewcommand{\arraystretch}{1.5}%
For more usable, i.e.\ more readable \emph{and more automatable}, algebraic proofs of lattice-theoretic fixed point theorems, we introduce the extended axiom system $\AIC_1$.
All additional axioms in $\AIC_1 \setminus \AIC_0$ are derivable from $\AIC_0$ and hence every model of $\AIC_0$ is also a model of $\AIC_1$.
The extended system $\AIC_1$ is indeed more usable:
Isabell's sledgehammer tool finds proofs of all fixed point theorems in \Cref{sec:case_studies} \emph{automatically} in $\AIC_1$ but not in $\AIC_0$.

We stress that $\AIC_0$ is \emph{\underline{not} complete} for sequence algebras (thus neither is $\AIC_1$), nor can any finite system of finitary axioms achieve completeness, cf.\ \Cref{sec:incompleteness}.
There is hence no point in attempting to provide a \emph{minimal} finitary axiomatization.
A complete but \emph{infinitary} axiomatization of sequence algebras is given in \Cref{sec:infinitary-ax}.%

\subsection{The Basic Axiom System $\boldsymbol{\AIC_0}$}
\label{sec:aic0}

The basic axiom system $\AIC_0$ is shown in \Cref{fig:axioms-aic0}. 
Since we want to capture sequence algebras, we provide intuitions for the setting where variables represent \emph{sequences}.
However, it should be noted that our axioms do \emph{not} explicitly axiomatize that we are operating on sequences.
There could well be other models.

\paragraph*{Bounded Lattice and Shifts}

These axioms state that models of $\AIC_0$ form a \emph{bounded lattice} with join $\lmax$, meet $\lmin$, least and greatest element $\bot$ and $\top$, and that $\Nextsymbol$ is a monotonic operation which distributes over $\lmax$ and~$\lmin$.
The axioms $\textnormal{\TirName{\nextbotname}}$, $\textnormal{\TirName{\nexttopname}}$,$ \TirName{\botname}$, and $\TirName{\topname}$ together ensure that $\bot$ and $\top$ are flat.

We postulate nothing on $\Firstsymbol$ because our axiomatization is bound to be incomplete anyway and we will not be needing $\Firstsymbol$ for the fixed point theorems of \Cref{sec:case_studies}.
We \emph{do} state axioms on $\Firstsymbol$ for the \emph{complete infinitary} axiomatization of sequence algebras in \Cref{sec:infinitary-ax}.%

\paragraph*{Majora and Minora}

\mbox{\TirName{\supinflatename}} and \TirName{\infdeflatename} state that $\Supsymbol$~can only enlarge a sequence whereas $\Infsymbol$ can only make it smaller.
By \TirName{\supidemname} and \TirName{\infidemname}, $\Supsymbol$ and $\Infsymbol$ are idempotent.
Intuitively, $\Sup \fora$ and $\Inf \fora$ together form the \emph{tightest} monotonic envelope around~$\fora$.
Putting a tightest envelope around the already tightest envelope cannot yield an even tighter one.
By \TirName{\supmononame} and \TirName{\infmononame}, $\Supsymbol$ and $\Infsymbol$ are monotonic.
These six axioms together describe that $\Supsymbol$ and $\Infsymbol$ form so-called \emph{closure} or \emph{hull} operators.
By \TirName{\nextsupcommname} and \TirName{\nextinfcommname}, $\Nextsymbol$ commutes with $\Supsymbol$ and $\Infsymbol$.

Finally, we postulate the induction rule \TirName{\nextindname} for $\Supsymbol$:
If $\fora$ is descending, i.e.\ $\Next \fora \preceq \fora$, then $\Sup \fora \preceq \fora$. 
In fact, then \mbox{even $\Sup \fora = \fora$} (since $\fora \preceq \Sup \fora$ already holds by \TirName{\supinflatename}), so taking majora of descending chains has no effect.
The converse is also true (indicated by the double inference bar): 
if $\Sup \fora \preceq \fora$, i.e.\ taking majorum of~$\fora$ has no effect, then $\fora$ is descending.
The coinduction principle \TirName{\nextcoindname} for ascending chains is entirely dual.

\paragraph*{Function Applications and Orbits}

In sequence algebras, functions are applied \emph{element-wise} to sequences, see \Cref{tab:semantics}.
Naturally, $\Nextsymbol$ and $\funca$ must then commute, as postulated by \TirName{\FNextcommname}.

All functions in AIC are \emph{monotonic}, postulated by an axiom \TirName{\Fmononame} for every function \mbox{symbol $\funca \in \Funcs$}.
It is not hard to prove that $\funca$ is monotonic 
(in the standard sense on single lattice-elements) 
\emph{if and only if} \TirName{\Fmononame} is valid for sequence algebras, see Appendix~\ref{sec:proof:FMono}.
Hence, \mbox{\TirName{\Fmononame}} captures \emph{precisely} standard monotonicity.
Inherited from monotonicity of $\funca$, applying $\funca$ \emph{multiple times} is also a monotonic operation.
Consequently, \TirName{\Fsmononame} postulates that $\funca^{*}$ is monotonic.
Contrary to our finitary axiomatization, monotonicity of~$\funca^{*}$ is provable from monotonicity of $\funca$ in the complete infinitary axiomatization.

$\TirName{\FFscommname}$ postulates that it does not matter whether we first apply $\funca$ once and then arbitrarily many times or first arbitrarily many times and then once more.
In other words, $\funca$ and $\funca^{*}$ commute.

\begin{wrapfigure}[3]{r}{11em}
\vspace*{-1.625\intextsep}%
\begin{minipage}{1\linewidth}
\begin{axbox}%
	\vspace*{-0.25\abovedisplayskip}%
	\smallrules%
	\begin{center}%
		$\axiter{\fora}$%
	\end{center}%
	\normalrules%
\end{axbox}%
\end{minipage}%
\end{wrapfigure}
Next, we investigate shifts of orbits.
Intuitively, we want $\Fas \fora$ to be the $\funca$-orbit of a lattice element.
If we left-shift such an orbit, we might expect that this results in $\funca$ being applied once more, i.e.\ $\Next \Fas \fora = \Fa \Fas \fora$.
However, in AIC $\fora$ is not a single element but a \emph{sequence} and hence $\funca^{*}$ has a more intricate semantics, see \Cref{tab:semantics}.
To be valid under this semantics, the identity \TirName{\itername} for shifting orbits is also more involved, see above right.

\begin{wrapfigure}[6]{l}{.45\linewidth}
\vspace*{-.75\intextsep}%
\begin{minipage}{1\linewidth}
\footnotesize%
\belowdisplayskip=0pt%
\begin{align*}
	\hspace*{-2em}\begin{array}{r@{\quad}r@{~{}~}r@{~{}~}r@{~{}~}r}
		\boldsymbol{\Next \Fas \fora\colon}		& \gray{\funca(\orange{\fora_1})}	& \gray{\funca^2 (\blue{\fora_2})}	& \gray{\funca^3 (\purple{\fora_3})}	& \gray{\funca^4 (\textcolor{violet}{\fora_4})}	\\[-.25em]
		\boldsymbol{\Fas \fora\colon}			& \green{\fora_0}\phantom{)}		& \gray{\funca(\orange{\fora_1})}	& \gray{\funca^2 (\blue{\fora_2})}	& \gray{\funca^3 (\purple{\fora_3})}			\\[-.25em]
		\boldsymbol{\fora\colon}				& \green{\fora_0}\phantom{)}		& \orange{\fora_1}\phantom{)}		& \blue{\fora_2}\phantom{)}		& \purple{\fora_3}\phantom{)}				\\[-.25em]
		\boldsymbol{\Next \fora\colon}			& \orange{\fora_1}\phantom{)}		& \blue{\fora_2}\phantom{)}		& \purple{\fora_3}\phantom{)}		& \textcolor{violet}{\fora_4}\phantom{)}		\\[-.25em]
		\boldsymbol{\Fas \Next \fora\colon}		& \orange{\fora_1}\phantom{)}		& \gray{\funca (\blue{\fora_2})}		& \gray{\funca^2 (\purple{\fora_3})}	& \gray{\funca^3 (\textcolor{violet}{\fora_4})}	\\[-.25em]
		\boldsymbol{\Fa \Fas \Next \fora\colon}	& \gray{\funca(\orange{\fora_1})}	& \gray{\funca^2 (\blue{\fora_2})}	& \gray{\funca^3 (\purple{\fora_3})}	& \gray{\funca^4 (\textcolor{violet}{\fora_4})}
	\end{array}
\end{align*}%
\normalsize%
\end{minipage}%
\end{wrapfigure}
\noindent%
The listing on the left may provide some graphical intuition on $\TirName{\itername}$ and it's \enquote{naive} version.

Finally, we postulate (co)induction principles \TirName{\Findname} and \TirName{\Fcoindname}:
If applying $\funca$ once takes us in a certain direction (with respect to~$\preceq$), then $\funca$'s orbit takes us in the same direction.

\paragraph*{Soundness of $\boldsymbol{\AIC_0}$}

We have the following theorem whose proof has been formalized in Isabelle/HOL~\cite{aic-afp}:%
\begin{restatable}[Soundness of $\AIC_0$]{theorem}{thmAICZeroSoundness}
\label{thm:soundness}
	$\AIC_0$ is sound for sequence algebras.
\end{restatable}%
\begin{proof}
See Appendix \ref{sec:app-soundness}.
\end{proof}

\subsection{The Usable Axiom System $\boldsymbol{\AIC_1}$}
\label{sec:aic1}

%
%
\renewcommand{\arraystretch}{3}%
\begin{figure*}%
	%
	\centering%
	\begin{adjustbox}{max width=.9999\textwidth, max totalheight=.9\textheight}%
		\normalsize%
		\centering%
		\begin{minipage}{1\linewidth} 
		\centering%
			\begin{tabular}{c}
					\textbf{\large\textsf{Additional Partial Order \& Lattice Axioms}}
				%
				%
				\\[-.75em]
				\begin{tabular}{rcl}%
					$\axreflex{\fora}$
					&
					$\cutrule{
						\fora \tss \forb
					}{
						\forb \tss \forc
					}{
						\fora \tss \forc
					}$
					&
					$\inequalequivrulet{
						\fora \tss \forb
						\\
						\forb \tss \fora
					}{
						\fora \hqeq \forb
					}$
				\end{tabular}%
				\\
				\begin{tabular}{rcl}%
					$\antisymmrule{
						\terma\subst{\fora}{\termc} \tss \termb\subst{\fora}{\termc}
					}{
						\termc \hqeq \termd
					}{
						\terma\subst{\fora}{\termd} \tss \termb\subst{\fora}{\termd}
					}$
					&
					$\equivinequalLrule{
						\fora \hqeq \forb
					}{
						\fora \tss \forb
					}$
					&
					$\equivinequalRrulet{
						\fora \hqeq \forb
					}{
						\forb \tss \fora
					}$
				\end{tabular}%
				\\
				\begin{tabular}{rccl}%
					$\axjoinidem{\fora}$
					&
					$\axmeetidem{\fora}$
					&
					$\joinintroLrule{
						\fora \tss \forc
					}{
						\forb \tss \forc
					}{
						\fora \lmax \forb \tss \forc
					}$
					&
					$\meetintroRrulet{
						\fora \tss \forb
					}{
						\fora \tss \forc
					}{
						\fora \tss \forb \lmin \forc
					}$
				\end{tabular}%
				\\
				\begin{tabular}{rccl}%
					$\meetintroLrule{
						\forb \tss \forc
					}{
						\fora \lmin \forb \tss \forc
					}$
					&
					$\joinintroRrule{
						\fora \tss \forb
					}{
						\fora \tss \forb \lmax \forc
					}$
					&
					$\joinelimrule{\fora \lmax \forb \tss \forc}{\forb \tss \forc}$
					&
					$\meetelimrulet{\fora \tss \forb  \lmin \forc}{\fora \tss \forb}$
				\end{tabular}%
			\end{tabular}
			
			\begin{tabular}{c}
					\textbf{\large\textsf{Additional Majora \& Minora Axioms}}
				%
				%
				\\[-.75em]
				\begin{tabular}{rccl}%
					$\supintroRrule{\fora \tss \forb}{\fora \tss \Sup \forb}$
					&
					$\infintroLrule{\fora \tss \forb}{\Inf \fora  \tss \forb} $
					&
					$\supelimrule{\Sup \fora \tss \forb}{\fora \tss \forb}$
					&
					$\infelimrulet{\fora \tss \Inf \forb}{\fora \tss \forb}$
				\end{tabular}%
				\\
				\begin{tabular}{rl}%
					$\suptightrule{\fora \tss \forb}{\Next \forb \tss \forb }{\Sup \fora \tss \forb}$
					&
					$\inftightrulet{\fora \tss \Next \fora }{\fora \tss \forb}{ \fora \tss \Inf \forb}$
				\end{tabular}%
				\\
				\begin{tabular}{rccl}%
					$\axsupexpand{\fora}$
					&
					$\axinfexpand{\fora}$
					&
					$\axsupdesc{\fora}$
					&
					$\axinfasc{\fora}$
				\end{tabular}%
			\end{tabular}
			
			\begin{tabular}{c}
					\textbf{\large\textsf{Additional Function Application \& Orbit Axioms}}
				%
				%
				\\[-.75em]
				\begin{tabular}{rccl}%
					$\axsemicont{\fora}$
					&
					$\axsemicocont{\fora}$
					&
					$\asciterrule{
						\fora \tss \Next \fora
					}{
						\Fa \Fas \fora 
						\tss  
						\Next \Fas \fora
					}$
					&
					$\desciterrulet{
						\Next \fora \tss \fora
					}{
						\Next \Fas \fora
						\tss  
						\Fa \Fas \fora 
					}$
				\end{tabular}%
				\\
				\begin{tabular}{rl}%
					$\orbitascrule{
						\fora \tss \Fa \fora
						\\
						\fora \tss \Next \fora
					}{
						\Fas \fora 
						\tss  
					\Next \Fas \fora
					}$
					&
					$\orbitdescrulet{
						\Next \fora \tss \fora
						\\
						\Fa \fora \tss \fora
					}{
						\Next \Fas \fora 
						\tss  
						\Fas \fora
					}$
				\end{tabular}%
				\\
				\begin{tabular}{rl}%
					$\FsintroLrule{\fora \tss \forb}{\Fa \forb \tss \forb }{\Fas \fora \tss \forb}$
					&
					$\FsintroRrulet{\fora \tss \Fa \fora }{\fora \tss \forb}{ \fora \tss \Fas \forb}$
				\end{tabular}%
			\end{tabular}
				\lightgray{\hrule}
		\end{minipage}
	\end{adjustbox}%
	\caption{%
		\label{fig:axioms-aic1}%
		Additional axioms for the automation-friendly (and also more human-friendly) axiom system $\AIC_1$.
		The axiom system $\AIC_1$ contains $\AIC_0$ (see \Cref{fig:axioms-aic0})  plus the additional axioms listed above.
		All additional axioms listed here are \emph{finitely} derivable from $\AIC_0$.
	}%
\end{figure*}%
\renewcommand{\arraystretch}{1.5}
The extended axiom system $\AIC_1$ is given by $\AIC_0$ plus all axioms shown in \Cref{fig:axioms-aic1}.
Let us go over some of \mbox{these additional axioms}.

\paragraph*{Additional Partial Order and Lattice Axioms}

$\AIC_0$ axiomatizes algebraically the underlying bounded lattice.
$\AIC_1$, through \TirName{\reflexname}, \TirName{\cutname}, and \TirName{\inequalequivname} axiomatizes directly the induced partial order~$\preceq$ on the lattice.
We will sometimes write%
\vspace*{-1.75\baselineskip}%
\begin{center}%
\begin{adjustbox}{max width=.999\linewidth}%
\begin{minipage}{1.14\linewidth}%
\scriptsizerules%
\begin{align*}
	\cutsrule{\fora_1 \tss \fora_2 \\ \fora_2 \tss \fora_3 \\ \ldots \\ \fora_{k-1} \tss \fora_k}{\fora_1 \tss \fora_k} 
	\qquad\textnormal{{\normalsize \raisebox{.55em}{\large instead of}}}\qquad
	\cutrulet{
		\fora_1 \tss \fora_2 
	}{
		\cutrulet{
			\fora_2 \tss \fora_3
		}{
			\nirulet{\cutname}{\vdots}{\fora_3 \tss \fora_k}
		}{
		\fora_2 \tss \fora_k
		}
	}{
	\fora_1 \tss \fora_k
	}
\end{align*}%
\normalrules%
\normalsize%
\end{minipage}%
\end{adjustbox}%
\end{center}%
as a shorthand for several intermediate applications of \TirName{\cutname}.

\begin{wrapfigure}[4]{l}{.49\columnwidth}
\vspace*{-.625\intextsep}%
\begin{minipage}{.999\linewidth}
\begin{infbox}%
	\smallrules%
	\begin{center}%
		$\antisymmrule{
			\terma\subst{\fora}{\termc} \tss \termb\subst{\fora}{\termc}
		}{
			\termc \hqeq \termd
		}{	
			\terma\subst{\fora}{\termd} \tss \termb\subst{\fora}{\termd}
		}$
	\end{center}%
\end{infbox}%
\end{minipage}%
\end{wrapfigure}%
\noindent%
The \TirName{indiscern} axiom is morally the congruence rule of equa{\-}tion{\-}al logic, but for inequalities and for arbitrary terms (not just flat ones of the form $\boxast\, (\fora,\, \forb,\ldots)$).
Intuitively, if $\termc = \termd$, then we can freely replace any occurrence of term $\termc$ with term $\termd$ in any inequality.
Formulating this as a rule is slightly more involved:
Let $\terma \preceq \termb$ be the context in which the subterm~$\termc$ potentially occurs.
All to-be-replaced occurrences of~$\termc$ are represented by a placeholder variable $\fora$ occurring instead of $\termc$ in $\terma$ and $\termb$ and hence $\terma\subst{\fora}{\termc} \preceq \termb\subst{\fora}{\termc}$ is the inequality in which occurrences of $\termc$ are to be replaced with~$\termd$.
If now $\termc = \termd$, we can actually make the replacement and conclude that if $\terma\subst{\fora}{\termc} \preceq \termb\subst{\fora}{\termc}$ holds, then \mbox{also $\terma\subst{\fora}{\termd} \preceq \termb\subst{\fora}{\termd}$ holds}.

\begin{wrapfigure}[3]{l}{.33\columnwidth}
\vspace*{-.875\intextsep}%
\begin{minipage}{1\linewidth}
%
\begin{center}%
	$\lantisymmrulet{
		\fora \lmax \forb \tss \forb \lmin \fora
	}{
		\fora \hqeq \Fa \fora
	}{
	\fora \lmax \forb \tss \forb \lmin \colmark{\Fa \fora}
	}$
\end{center}%
\end{minipage}%
\end{wrapfigure}%
\noindent
As an example, the derivation on the left is valid.
Here, term $\terma$ is $\fora \lmax \forb$, term $\termb$ is $\forb \lmin \forc$, and the placeholder is~$\forc$.
The substitution instance $\forb \lmin \Fa \fora$ is obtained by replacing every $\forc$ in $\forb \lmin \forc$ by $\Fa \fora$ and $\forb \lmin \fora$ is obtained by replacing $\forc$ with $\fora$.
Notice that when reading the inference rule from top to bottom, \emph{\underline{not} every occurrence} of $\fora$ was replaced by $\Fa \fora$.
This is because $\terma$ was $\fora \lmax \forb$ 
but placeholder $\forc$ does not occur in $\fora \lmax \forb$.

Just for readability (but with \emph{no} technical necessity), we write $\forb \lmin \colmark{\Fa \fora}$ (i.e.\ with a blue marking) below the inference bar to indicate where $\Fa \fora$ replaced $\fora$.
Intuitively, \TirName{\antisymmname} gives us the flexibility to freely select only \emph{certain} occurrences of $\termc$ for replacement by $\termd$ and we indicate the locations were a replacement happened in \colmark{\textnormal{blue}}.

\begin{wrapfigure}[4]{r}{.43\columnwidth}
	\begin{minipage}{.999\linewidth}
		\smallrules%
		\begin{infbox}%
			\begin{center}
				$\joincommrule{\forc\subst{\vara}{\forb \lmax \fora} \tss \ford\subst{\vara}{\forb \lmax \fora}}{\forc\subst{\vara}{\fora \lmax \forb} \tss \ford\subst{\vara}{\fora \lmax \forb}}$%
			\end{center}%
		\end{infbox}%
	\end{minipage}%
\end{wrapfigure}%
We will often use \textsc{\antisymmname} implicitly (but rigorously) in the 
context of identity axioms:
Whenever we introduce an identity \textsc{abc}, we will implicitly introduce an inference rule which is derivable using \textsc{abc} as the right-hand premise of \textsc{\antisymmname}, \emph{and} we will call this new inference rule also \textsc{abc}.
For example, the identity \textsc{$\lmax$-comm} induces the rule \textsc{$\lmax$-comm} shown in the box above.
Its derivation is shown below left.
\textcolor{lightgray!50!DodgerBlue}{\TirName{\assumptionname}} indicates an assumption (i.e.\ a premise above the bar) of an inference rule / quasiequation that we want to derive.
Below right, we see an example of applying the implicit \mbox{\textsc{$\lmax$-comm}} rule on
the left-hand-side (were we again indicate in \colmark{\textnormal{blue}} the part to which the explicit \textsc{$\lmax$-comm} identity was applied).%

\vspace{-1\baselineskip}%
\smallrules%
\begin{align*}
	\lantisymmrulet{
				\sassume{\forc\subst{\vara}{\forb \lmax \fora}}{\ford\subst{\vara}{\forb \lmax \fora}}
			}{
				\axjoincomm{\fora}{\forb}
			}{
				\forc\subst{\vara}{\fora \lmax \forb} \tss \ford\subst{\vara}{\fora \lmax \forb}
			}	
	\qquad\qquad
	\joincommrule{\forb \lmax \fora \tss \forb \lmin \fora}{\colmark{\fora \lmax \forb} \tss \forb \lmin \fora}
\end{align*}%
\normalrules%

\paragraph*{Additional Majora and Minora Axioms}

By \TirName{\supdescname} and \TirName{\infascname}, majora form descending chains and minora form ascending chains.
In fact, we had mentioned earlier that the majorum $\Sup \fora$ is the \emph{least decreasing majorant} of $\fora$.
So if some sequence $\forb$ is descending~($\Next \forb \preceq \forb$) and $\forb$ is a majorant of $\fora$ ($\fora \preceq \forb$), then $\forb$ is also a majorant of $\Sup \fora$. 
Dually the minorum $\Inf \fora$ is $\fora$'s \emph{greatest increasing minorant}.
The axioms \TirName{\suptightname} and \TirName{\inftightname} capture precisely this.
Finally, much like in LTL, both $\Supsymbol$~and~$\Infsymbol$ each satisfy an \emph{expansion law}.%
\begin{remark}[Majora-Minora Alternation Depth]%
Just like in LTL \cite[Figure 5.7, absorption laws]{DBLP:books/daglib/0020348}, multiple consecutive $\Supsymbol$'s or $\Infsymbol$'s can be contracted to a single one (via idempotence), and the depth of alternating $\Supsymbol$'s and~$\Infsymbol$'s collapses at level 2, leaving only the two innermost operators relevant.
For instance, $\Sup\Inf\Inf\Sup\Sup \fora = \Inf \Sup \fora$. See \hyperref[sec:app-collapse]{Appendix \ref{sec:app-collapse}} for a more formal statement with proof.
\lipicsEnd
\end{remark}%

\paragraph*{Additional Function Application and Orbit Axioms}

\begin{wrapfigure}[5]{r}{.525\textwidth}%
\vspace*{-2.5\intextsep}%
\begin{minipage}{.999\linewidth}%
	\smallrules%
	\vspace*{-1\abovedisplayskip}%
	\begin{gather*}
		\hspace{-1.5em}\suptightrulet{
			\Fmonorule{
				\axsupinflate{\fora}
			}{
				\Fa \fora \tss \Fa \Sup \fora
			}
		}{
			\FNextcommrulet{
				\Fmonorulet{
					\axsupdesc{\fora}
				}{
					\Fa \Next \Sup \fora \tss \Fa \Sup \fora
				}
			}{
				\colmark{\Next \Fa \Sup \fora} 
				\tss 
				\Fa \Sup \fora
			}
		}{
			\Sup \Fa \fora \tss \Fa \Sup \fora
		}
	\end{gather*}%
	\normalrules%
\end{minipage}%
\end{wrapfigure}%
Any monotonic function is \emph{semi-continu{\-}ous}.
In AIC, this reads $\Sup \Fa \fora \preceq \Fa \Sup \fora$, stated by \TirName{\semicontname}.
Dually, \TirName{\semicocontname} states 
$\Fa \Inf \fora \preceq \Inf \Fa \fora$.
To get a first feel for how to conduct actual AIC-style algebraic proofs, the proof tree above right shows the derivation of \TirName{\semicontname} from other axioms.

\begin{wrapfigure}[7]{l}{.36\columnwidth}
\vspace*{-.5\intextsep}%
\begin{minipage}{.999\linewidth}
\begin{infbox}%
	\smallrules%
	\begin{center}
		$\asciterrule{
			\fora \tss \Next \fora
		}{
			\Fa \Fas \fora 
			\tss  
			\Next \Fas \fora
		}$
		\\[.75em]
		$\desciterrule{
			\Next \fora \tss \fora
		}{
			\Next \Fas \fora
			\tss  
			\Fa \Fas \fora 
		}$
	\end{center}%
\end{infbox}%
\end{minipage}%
\end{wrapfigure}%
\noindent%
As for additional orbit axioms, we can recover the \enquote{naive} version of \TirName{\itername}, namely $\Next \Fas \fora = \Fa \Fas \fora$, for flat sequences:
For ascending chains~$\fora$, the two rule \TirName{\ascitername} yields $\Fa \Fas \fora \preceq \Next \Fas \fora$, and for descending chains the rule \TirName{\descitername} yields $\Next \Fas \fora \preceq \Fa \Fas \fora$.
If $\fora$ is now flat ($\fora \preceq \Next \fora \preceq \fora$), we obtain $\Next \Fas \fora = \Fa \Fas \fora$.

\noindent%
The axioms \TirName{\orbitascname} and \TirName{\orbitdescname} assert that certain orbits form chains.
This is of interest since chains always converge in our setting and we often want orbits to constitute fixed point iterations that converge to a fixed point.
If $\fora \preceq \Fa \fora$ and $\fora$ is ascending, for instance, then the orbit $\Fas \fora$ is provably also ascending.
Dually, if $\Fa \fora \preceq \fora$ and $\fora$ is descending, $\Fas \fora$ is also descending.

Finally, we can derive from \textsc{\Findname} and \textsc{\Fcoindname} rules that are similar to \textsc{\suptightname} and \textsc{\inftightname}, namely \TirName{\FsintroRname} and \TirName{\FsintroLname}.%

\paragraph*{Soundness of $\boldsymbol{\AIC_1}$}

We have the following theorem whose proof has been formalized in Isabelle/HOL:%
\begin{restatable}[Soundness of $\AIC_1$]{theorem}{thmAICOneSoundness}
\label{thm:soundness-aic1}
	$\AIC_1$ is sound for sequence algebras.
\end{restatable}%
\begin{proof}
	$\AIC_1$ contains all axioms of $\AIC_0$ which are all valid for sequence algebras (see \Cref{thm:soundness}).
	All additional axioms in $\AIC_1 \setminus \AIC_0$ are finitely derivable from $\AIC_0$ and hence also valid for any model of $\AIC_0$.
	For those derivations, see Appendix \ref{sec:app-soundness-AIC1}.
\end{proof}

\paragraph*{Continuous Functions}
\label{sec:continuity-axioms}

\begin{wrapfigure}[6]{r}{.23\columnwidth}
\vspace*{0\intextsep}\vspace*{-1\baselineskip}%
\begin{minipage}{.999\linewidth}
\begin{axbox}%
	\vspace*{-.25\abovedisplayskip}%
	\smallrules%
	\begin{center}
		$\axalephcont{\fora}$
		\\[.5em]
		$\axalephcocont{\fora}$
	\end{center}%
\end{axbox}%
\end{minipage}%
\end{wrapfigure}%
From monotonicity, all functions are necessarily semi-con{\-}tinuous: $\Sup \Fa \fora \preceq \Fa \Sup \fora$.
If the other direction $\Fa \Sup \fora \preceq \Sup \Fa \fora$ (and thus the identity $\Fa \Sup \fora = \Sup \Fa \fora$) holds, the function is \enquote{fully} continuous.
We can state that certain functions $\funca$ are continuous by additionally postulating  \emph{non-derivable} continuity axioms for those functions.

More specifically, if we can pull out majora from~$\funca$ without discharging any further premises, this essentially means that one can pull the supremum of any \emph{unordered \underline{non-em}p\underline{t}y countable set $S$} out from~$\funca$, i.e.\ $\funca(\sup S) = \sup \funca(S)$. 
We would then postulate for such functions the axiom \TirName{\alephcontname} shown above right.
Dually, we would postulate \TirName{\alephcocontname} for functions that preserve countable infima.

\begin{wrapfigure}[7]{l}{.39\columnwidth}
\vspace*{-.5\intextsep}%
\begin{minipage}{.999\linewidth}
\begin{axbox}%
	\smallrules%
	\begin{center}
		$\omegacontrule{\fora \tss \Next \fora}{\F \Sup \fora \tss \Sup \F \fora}$
		\\[.75em]
		$\omegacocontrule{\Next \fora \tss \fora}{\Inf \F \fora \tss \F \Inf \fora}$
	\end{center}%
\end{axbox}%
\end{minipage}%
\end{wrapfigure}%
\noindent%
A more commonly used notion of continuity, called \mbox{\emph{$\omega$-continuity}}~\cite{abramsky1994domain}, is when one need not be able to pull out suprema of any unordered set but only of ascending chains.
In AIC, this corresponds to being able to pull out majora of ascending chains and we would capture $\omega$-continuity of~$\funca$ by postulating \TirName{\omegacontname}.
As with monotonicity, one can prove that $\funca$ is \alephzeroshortname- or $\omega$-(co)continuous 
if and only if the appropriate quasi-inequalities are valid for sequence algebras, see \Cref{sec:proof:omegaCont,sec:proof:alephCont}.
Notions and axioms for minora and cocontinuity are entirely dual.

\section{Algebraic Proofs of Fixed Point Theorems}
\label{sec:case_studies}

To demonstrate the usefulness of AIC, we will now present purely algebraic proofs of well-known, less-well-known, and novel fixed point theorems.

\subsection{Kleene, Tarski-Kantorovich, and Park}
\label{sec:case_studies:kleene}

\subsubsection{The Tarski-Kantorovich Principle}

Before we prove in AIC the well-known \emph{Kleene fixed point theorem}\footnote{\label{foot:restricted}Though for our restricted setting where the lattice is countably complete.}~\cite{DBLP:journals/ipl/LassezNS82}, we consider a slightly more general result -- the \emph{Tarski-Kan{\-}toro{\-}vich principle} (TKP)~\cite{jachymski2000tarski}.
It says: for $\omega$-continuous~$\funca$, if $\laelem \laord \funca(\laelem)$ then iterating $\funca$ on $\laelem$ will converge to the least fixed point of $F$ just above
$\laelem$, which we denote by $\lfpabove{\laelem} \funca$:%
\begin{theorem}[Tarski-Kantorovich Principle]
	\label{thm:tkp}
	Let $\funca\colon \ladomain \to \ladomain$ be an $\omega$-continuous endomap on a complete lattice~$\ladomain$, and $\laelem\in\ladomain$. 
	Then
	\begin{align*}
		\laelem \llaord \funca(\laelem) \qqimplies \lfpabove{\laelem} \funca \eeq \sup_{n \in \omega}\, \funca^n(\laelem). \tag*{\textsc{(TKP)}}
	\end{align*}
	%
	
\end{theorem}%
\begin{remark}[On Duals]
TKP and all other fixed point theorems in this section have duals, which emerge by reversing the order (e.g.\ replace least with greatest fixed points, $\omega$-continuity with $\omega$-cocontinuity, etc.), which can be proved in AIC in the same way.%
\lipicsEnd%
\end{remark}%
\paragraph*{AIC Formalization of the Tarski-Kantorovich Principle (\Cref{thm:tkp})}

Before we \emph{algebraically prove} the theorem, we will first describe how to \emph{formalize} it in AIC.
Let us assume some~$\fora$ such that $\fora \preceq \Fa \fora$ and that \textsc{\omegacontname} holds for $\funca$.
While $\fora$ is by default any sequence in AIC, $\fora$ morally takes the role of the single lattice element $\laelem$.
It would thus make sense to require that $\fora$ is flat, so as to emulate a single lattice element.
However, it turns out that this is not necessary, as long as $\fora$ is ascending, i.e.\ $\fora \preceq \Next \fora$.
Note that this makes the theorem we will prove more general. 
To obtain the classical single-element version of \Cref{thm:tkp}, we would instantiate $\fora$ with the flat sequence $\overline{\laelem} = \constseq{\laelem}$ and that sequence is trivially ascending, hence $\overline{\laelem} \preceq \Next \overline{\laelem}$.

\begin{wrapfigure}[9]{r}{.56\columnwidth}
\vspace*{-1.5\intextsep}%
\begin{minipage}{.999\linewidth}
\begin{infbox}%
	\smallrules%
	\begin{center}
		$\tkpfprule
		{
		\metaassume{\Fa~\textsc{\omegacontname}}
		\\
		\fora \tss \Next \fora
		\\
		\fora \tss \Fa \fora
	}
	{\Fa \Sup \Fas \fora 
		\hqeq
	\Sup \Fas \fora 
	}$
	\\[1em]
	$\tkpaboverule{
			\fora
			\tss
			\Fa \fora
		}{
			\fora
			\tss
			\Sup \Fas \fora
		}$
		\\[1em]
		$\tkpleastrule{
			\fora \tss \forb
			\\
			\Fa \forb \tss \forb
			\\
			\Next \forb \tss \forb
		}{
			\Sup \Fas \fora 
			\tss
			\forb
		}$
	\end{center}%
	\normalrules%
\end{infbox}%
\end{minipage}%
\end{wrapfigure}%
\noindent%
Assuming $\fora \preceq \Fa \fora$ and $\fora \preceq \Next \fora$, we formalize that $\Sup \Fas \fora$ (cf.\ $\sup_{n \in \omega} \funca^n(\laelem)$) is the least fixed point of $\Fa$ above~$\fora$. 
This comprises proving that (i) $\Sup \Fas \fora$ is a fixed point, i.e.\ that \TirName{\tkpfpname} shown above right
is derivable.
By $\metaassume{\ldots}$, we \emph{annotate} which continuity axioms we need to postulate on $\funca$ (cf.\ \Cref{sec:calc}, \hyperref[sec:continuity-axioms]{\small\textsf{\textbf{Continuous Functions}}}) to prove the rule derivable.
This is just for convenience and bears no \mbox{further formal meaning}.

We further formalize that $\Sup \Fs \fora$~is the \emph{\underline{least}} fixed point \emph{above}~$\fora$, i.e.\ (ii) $\Sup \Fs \fora$ is itself above $\fora$ and (iii) any fixed point (and in particular any \emph{pre}fixed point)~$\forb$ above $\fora$ is also above $\Sup \Fs \fora$.
We thus derive the two quasi-inequalities \TirName{\tkpabovename} and \TirName{\tkpleastname} above.%

\paragraph*{Algebraic Proof of the Tarski-Kantorovich Principle (\Cref{thm:tkp})}

\begin{wrapfigure}[7]{l}{.46\columnwidth}
\vspace*{-1\intextsep}%
\begin{minipage}{1\linewidth}
\begin{infbox}%
	\smallrules%
	\begin{center}
		$\supquasiprefprule{\fora \tss \Next \fora}{\Sup \Fa \Fas \fora \tss \Sup \Fas \fora}$
		\\[1.5em]
		$\supquasipostfprule{\fora \tss \Fa \fora}{\Sup \Fas \fora \tss \Sup \Fa \Fas \fora}$
	\end{center}%
	\normalrules%
\end{infbox}%
\end{minipage}%
\end{wrapfigure}%
For proving \TirName{\tkpfpname},
it is useful to first prove that $\Fas \fora$ is a \emph{majoral quasi fixed point}, i.e.\ \emph{under the scope of a majorum}, it makes no 
difference whether or not $\funca$ is applied once more to~$\Fas \fora$; formally, 
$\Sup \Fa \Fas \fora \hqeq \Sup \Fas \fora$, which intuitively means that $\Sup \Fas \fora$ is a fixed point of $\funca$ \emph{up to continuity}.
We divide this identity into the two quasi-inequalities shown in the box above right, and will later reuse these lemmas to prove the Olszewski fixed point theorem.
For a proof of \textsc{\small\supquasiprefpname}, consider below the proof tree on the left; for \textsc{\small\supquasipostfpname} the one on the right.

\vspace*{-1\abovedisplayskip}%
\smallrules%
	\begin{align*}
	\supexpandrule{
		\joincommrulet{
			\joinintroRrulet{
				\nextsupcommrulet{
					\supmonorulet{
						\asciterrulet{
							\sassume{\fora}{\Next \fora}
						}{
							\Fa \Fas \fora 
							\tss
							\Next \Fas \fora
						}
					}{
						\Sup \Fa \Fas \fora 
						\tss
						\Sup \Next \Fas \fora
					}
				}{
					\Sup \Fa \Fas \fora 
					\tss
					\colmark{\Next \Sup \Fas \fora}
				}
			}{
				\Sup \Fa \Fas \fora 
				\tss
				\Next \Sup \Fas \fora \llmax \Fas \fora
			}
		}{
			\Sup \Fa \Fas \fora 
			\tss
			\colmark{\Fas \fora \llmax \Next \Sup \Fas \fora}
		}
	}{
		\Sup \Fa \Fas \fora 
		\tss
		\colmark{\Sup \Fas \fora}
	}
	\qquad\qquad\quad
		\supmonorule{
			\FFscommrulet{
				\Fsmonorulet{
					\sassume{\fora}{\Fa \fora}
				}{
					\Fas \fora 
					\tss 
					\Fas \Fa \fora
				}
			}{
				\Fas \fora 
				\tss 
				\colmark{\Fa \Fas \fora}
			}
		}{
			\Sup \Fas \fora 
			\tss 
			\Sup \Fa \Fas \fora
		}
\end{align*}%
\normalrules%
\normalsize%
Before we show that $\Sup \Fas \fora$ is a \emph{proper} fixed point, we would like to point out the seemingly mundane essence of the two proofs shown above:
\emph{Both are about compensating for \underline{one} excess~$\funca$.}
For \textsc{\small\supquasiprefpname} (left tree), the excess $\funca$ is on the left-hand side of the inequality. 
To compensate, we can create an additional $\funca$ on the right by expanding~$\Supsymbol$, crucially: forgetting the irrelevant $\Fas \fora$, creating a $\Nextsymbol$ on the non-forgotten part, and this $\Nextsymbol$ eventually (further up) acts on $\funca^{*}$ to create an additional~$\funca$ (via \textsc{\small\itername} hidden in \textsc{\small\ascitername}).

For \textsc{\small\supquasipostfpname} (right tree), the excess $\funca$ is on the right, and while we \emph{could} expand the left~$\Supsymbol$, we cannot forget joinees ($\curlyvee$) on the left.
Therefore, instead of creating an~$\funca$ on the left, we need that the excess $\funca$ on the right makes that side bigger which is essentially ensured by $\fora \preceq \Fa \fora$.
As a precursor to the Olszewski theorem about $\Inf \Sup \Fas \fora$ we will see later:
$\fora \preceq \Fa \fora$ would  not have been necessary if we had some additional $\Infsymbol$ guarding $\Fas \fora$.
Then we can expand $\Infsymbol$ on the left and forget the unshifted part of that $\Infsymbol$-expansion.

We now show that 
$\Sup \Fas \fora$ is a \emph{proper} fixed point:
For $\Sup \Fas \fora$ being a \emph{post}fixed point (\textsc{\small\tkppostfpname}), we require only semi-continuity (which all monotonic functions satisfy):

\smallrules%
\vspace*{-1\baselineskip}%
\begin{align*}
	\cutrule{
		\supquasipostfprule{
			\sassume{\fora}{\Fa \fora}
		}{
			\Sup \Fas \fora 
			\tss
			\Sup \Fa \Fas \fora 
		}
	}{
		\axsemicont{\Fas \fora}
	}{
		\Sup \Fas \fora 
		\tss
		\Fa \Sup \Fas \fora
	}
\end{align*}%
\normalrules%
\normalsize%
For $\Sup \Fas \fora$ being a \emph{pre}fixed point, we require \enquote{full} continuity. 
For $\omega$-continuous $\funca$, we have:

\smallrules%
\vspace*{-1\baselineskip}%
\begin{align*}
	\cutrule{
		\omegacontrule{
			\orbitascrulet{
				\sassume{\fora}{\Fa \fora}
			}{
				\Fas \fora 
				\tss 
				\Next \Fas \fora
			}
		}{
			\Fa \Sup \Fas \fora 
			\tss 
			\Sup \Fa \Fas \fora
		}
	}{
		\supquasiprefprulet{
			\sassume{\fora}{\Next \fora}
		}{
			\Sup \Fa \Fas \fora 
			\tss
			\Sup \Fas \fora 
		}
	}{
		\Fa \Sup \Fas \fora 
		\tss
		\Sup \Fas \fora
	}
\end{align*}%
\normalrules%
\normalsize%
For \alephzeroshortname-continuous $\funca$, the left branch of the above tree simplifies to:

\smallrules%
\vspace*{-1\baselineskip}%
\begin{align*}
	\axalephcont{\Fas \fora}
\end{align*}%
\normalrules%
\normalsize%
\begin{remark}[Prefixed Points of \alephzeroshortname-continuous Maps]
As a by-product, we have shown for ascending $\fora$ that $\Sup \Fas \fora$ is a prefixed point of any \alephzeroshortname-continuous map $\funca$.
Indeed, \mbox{$\fora \preceq \Fa \fora$} is \emph{not} required for being a prefixed point.
In the classical single-lattice-element setting this reads: 
If $\funca$ is \alephzeroshortname-continuous, then $\sup_{n \in \omega} \funca^n(\laelem)$ is a prefixed point of $\funca$ for any~\mbox{$\laelem \in \ladomain$}.
While a minor result, this is to the best of our knowledge a novel result.%
\lipicsEnd%
\end{remark}%
We can now combine ${\Sup \Fas \fora} \preceq {\Fa \Sup \Fas \fora}$ and ${\Fa \Sup \Fas \fora} \preceq {\Sup \Fas \fora}$ using \textsc{\small\inequalequivname} to obtain the desired quasiequation \TirName{\tkpfpname}.
It is then left to prove \textsc{\small\tkpabovename} and \textsc{\small\tkpleastname}.
For \textsc{\small\tkpabovename}, see the proof tree below left.
For \textsc{\small\tkpleastname}, see the tree below right, thus completing the algebraic proof of \Cref{thm:tkp}.

\vspace*{-.75\baselineskip}%
\smallrules%
{\begin{align*}
	\supintroRrule{
		\Findrulet{
			\sassume{\fora}{\Fa \fora}
		}{
			\fora
			\tss
			\Fas \fora
		}
	}{
		\fora
		\tss
		\Sup \Fas \fora
	}
	\qquad\qquad\qquad\quad
	\suptightrulet{
		\lFsintroLrulet{
			\sassume{\fora}{\forb}
		}{
			\sassume{\Fa \forb}{\forb}
		}{
			\Fas \fora 
			\tss
			\forb
		}
	}{
		\sassume{\Next \forb}{\forb}
	}{
		\Sup \Fas \fora 
		\tss
		\forb
	}
	\tag*{\normalsize\qed}
\end{align*}%
}%
\normalrules%
\noindent%

\subsubsection{The Kleene Fixed Point Theorem}

We obtain the Kleene fixed point theorem as an immediate corollary of the Tarski-Kantorovich principle, namely by instantiating $\fora = \bot$ to obtain the very least fixed point (i.e.\ the least one above~$\bot$), then given by $\Sup \Fas \bot$ (cf. $\sup_n \funca^n(\bot)$).

We can also easily read this corollary off the AIC proof of TKP.
This is because wherever $\fora \preceq \Next \fora$ or $\fora \preceq \Fa \fora$ occur as assumptions in our proofs and we have $\fora = \bot$, we can vacuously discharge these assumptions by applying the axiom \textsc{\small\botname}.

\subsubsection{Park Induction}

Similarly to the case for the Kleene fixed point theorem, the Park induction principle \cite{park} for upper-bounding least fixed points of continuous functions $\funca$, i.e.\ for all $\laelema\in\ladomain$,
\[
	\funca(\laelema) \llaord \laelema 
	\qqimplies \lfp \funca \llaord \laelema
\]
is an immediate instance of \textsc{\small\tkpfpname}, \textsc{\small\tkpabovename}, \textsc{\small\tkpleastname} for $\fora =\bot$ and $\forb = \constseq{\laelema}$. In fact, we can readily read off our algebraic proofs a \emph{generalized} Park induction principle:
\begin{theorem}[Generalized Park Induction]
	\label{thm:generalized_park}
For $\laelema,\laelemb \in \ladomain$ and $\omega$-continuous $\funca\colon \ladomain\to\ladomain$,%
\begin{align*}
	\laelemb \llaord \funca(\laelemb)
	\qand \laelemb \llaord \laelema 
	\qand \funca(\laelema) \llaord \laelema  
	\qqimplies
	\lfpabove{\laelemb} \funca \llaord \laelema~.
\end{align*}
\end{theorem}
\begin{proof}
	 Instantiate \textsc{\small\tkpfpname} and \textsc{\small\tkpabovename} with $\fora = \constseq{\laelemb}$ and $\forb = \constseq{\laelema}$ (both flat). 
	 It follows that $\sem{\Sup \Fs \fora}{0} = \lfpabove{\laelemb} \funca$, and thus $\lfpabove{\laelemb} \funca \laord \laelema$ by $\textsc{\small\tkpleastname}$. 
\end{proof}
%
%

\subsubsection*{*Digression: Pen-and-Paper Algebraic Proofs of Fixed Point Theorems}

\begin{wrapfigure}[10]{r}{.7\linewidth}%
	\vspace*{-1\intextsep}%
	\verylightgray{\shadowbox{\includegraphics[width=.963\linewidth]{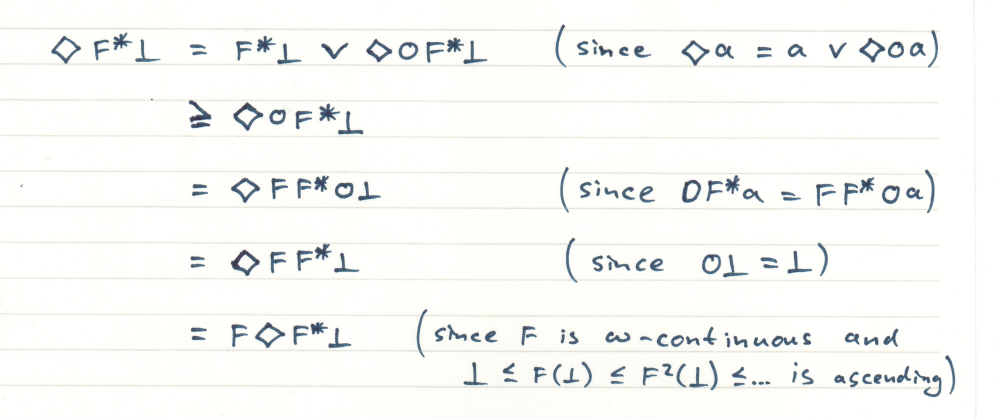}}}%
\end{wrapfigure}%
Although the quasi-equa{\-}tional proof trees above are rigorous and readily automatable, they are not the way a human would typically approach a fixed point theorem.
For AIC to feel natural and genuinely useful, it should also facilitate the kind of compact and intuition-guided pen-and-paper-style arguments by which  humans discover and prove such results.
Inspired by \cite[Figure 1]{DBLP:journals/pacmpl/RosselSKSG26}, the example above illustrates this second mode of reasoning.
It showcases a pen-and-paper proof of the inequation $\Fa \Sup \Fas \bot \preceq \Sup \Fas \bot$, i.e.\ the \enquote{more difficult direction} in showing that~$\sup_n \funca^n(\bot)$ is a fixed point of~$\funca$ as part of the Kleene fixed point theorem.
This is the kind of reasoning we ultimately want AIC to enable:
not only proofs that a theorem prover can check, but proofs that a human can actually read, understand, and internalize as a standard fixed point argument.
\subsection{Olszewski}
\label{sec:case_studies:ol}

By the Tarski-Kantorovich principle (\autoref{thm:tkp}), seeding the iteration of an $\omega$-(co)con{\-}tinuous map $\funca$ with a pre- or postfixed point of $\funca$ converges to a \emph{proper} fixed point. 
But how to find pre- or postfixed points to begin with? 
A recent result due to Olszewski shows that they emerge by iterating on an \emph{arbitrary} seed element.%
\begin{restatable}[\textnormal{Olszewski~\cite{Ol:omegaSequences}}]{theorem}{thmOlszewski}
\label{thm:olszewski}
	Let $\funca\colon \ladomain \to \ladomain$ be $\omega$-cocontinuous and $\laelem\in \ladomain$.
	Then
		$\inf_{n\in \omega}~\sup_{k\geq n}~ \funca^k(\laelem)$ is a p\underline{ost}fixed point of $\funca$.
\end{restatable}%
For Olszewski's proof, see \appref{app:proof:olprose}. %
Under stronger continuity assumptions on $\funca$, we can improve on Olszewski's result and obtain \emph{proper} fixed points:
\begin{theorem}[Proper Olszewski Fixed Point]\label{cor}
	Let $F\colon \ladomain \to \ladomain$ be $\omega$-cocontinuous and \alephzeroshortname-continuous and $\laelem\in \ladomain$.
	Then
    $\inf_{n\in \omega}\sup_{k\geq n} F^k(\laelem)$ is a fixed point of $F$.
\end{theorem}

\begin{remark}
There is also a transfinite version of \autoref{thm:olszewski} that only assumes monotonicity of~$F$, see~\cite[Thm.~1]{olszewski21_order}. Similarly to the transfinite version of the Knaster-Tarski Theorem by Cousot \& Cousot~\cite{cousot1979constructive}, this result is currently out of scope of AIC.
\lipicsEnd%
\end{remark}

\begin{figure*}[t]
	\begin{minipage}{.999\linewidth}%
	\scriptsizerules%
	\begin{gather*}
		\cutsrulet{
			\linfsupquasipostfprulet{
				\sassume{\Next \fora}{\fora}
			}{
				\Inf \Sup \Fas \fora
				\tss 
				\Inf \Sup \Fa \Fas \fora
			}
			\\
			\infmonorule{
				\axsemicont{\Fas \fora}
			}{
				\Inf \Sup \Fa \Fas \fora
				\tss 
				\Inf \Fa \Sup \Fas \fora
			}
			\\
			\omegacocontrulet{
				\axsupdesc{\Fas \fora}
			}{
				\Inf \Fa \Sup \Fas \fora
				\tss 
				\Fa \Inf \Sup \Fas \fora
			}
		}{
			\Inf \Sup \Fas \fora
			\tss 
			\Fa \Inf \Sup \Fas \fora
		}
		\\[1em]
		\cutsrulet{
			\axsemicocont{\Sup \Fas \fora}
			\\
			\infmonorule{
				\axalephcont{\Fas \fora}
			}{
				\Inf \Fa \Sup \Fas \fora
				\tss 
				\Inf \Sup \Fa \Fas \fora
			}
			\\
			\infsupquasiprefprulet{
				\sassume{\fora}{\Next \fora}
			}{
				\Inf \Sup \Fa \Fas \fora
				\tss 
				\Inf \Sup \Fas \fora
			}
		}{
			\Fa \Inf \Sup \Fas \fora
			\tss 
			\Inf \Sup \Fas \fora
		}
	\end{gather*}%
	\normalrules%
	\end{minipage}%
	\caption{Proof of \TirName{\olpostfpname} (top tree) and \TirName{\olprefpname} (bottom tree).}
	\label{fig:ol-postfp-proof}
\end{figure*}

\paragraph*{AIC Formalization of Olszewski Fixed Points (\Cref{thm:olszewski,cor})} 

\begin{wrapfigure}[10]{r}{.52\columnwidth}
\vspace*{-1\intextsep}%
\begin{minipage}{.999\linewidth}
\begin{infbox}%
	\smallrules%
	\begin{center}
		$\olpostfprule{
			\metaassume{$\Fa$ \textsc{\omegacocontname}}
			\\
			\Next \fora \tss \fora
		}{
			\Inf \Sup \Fas \fora \tss \Fa \Inf \Sup \Fas \fora
		}$
		\\[1em]
		$\olprefprule{
			\metaassume{$\Fa$ \textsc{\alephcontname}}
			\\
			\fora \tss \Next \fora
		}{
			\Fa \Inf \Sup \Fas \fora \tss \Inf \Sup \Fas \fora
		}$
		\\[1em]
		$\olfprule{
			\metaassume{$\Fa$ \textsc{\omegacocontname}, $\Fa$ \textsc{\alephcontname}}
			\\
			\Next \fora \tss \fora
			\\
			\fora \tss \Next \fora
		}{
			\Inf \Sup \Fas \fora \hqeq \Fa \Inf \Sup \Fas \fora
		}$
	\end{center}%
	\normalrules%
\end{infbox}%
\end{minipage}%
\end{wrapfigure}%
%
%
As in \Cref{sec:case_studies:kleene}, the fixed seed element~$\laelema$ is associated with a variable $\fora$. 
\Cref{thm:olszewski} then reads in AIC as \TirName{\olpostfpname}, 
where -- analogously to the AIC formalization of \Cref{thm:tkp} -- we prove a more general statement and hence have an additional premise $\Next \fora \preceq \fora$. \Cref{thm:olszewski} is obtained from instantiating $\fora$ with the flat sequence $\constseq{\laelema}$, for which this premise holds vacuously.
Similarly, \Cref{cor} is formalized as \TirName{\olfpname}.

\paragraph*{Algebraic Proof of Olszewski Fixed Points (\Cref{thm:olszewski,cor})} 

\begin{wrapfigure}[6]{r}{.52\columnwidth}
\vspace*{-1\intextsep}%
\begin{minipage}{.999\linewidth}
\begin{infbox}%
	\smallrules%
	\begin{center}
		$\infsupquasipostfprule{
			\Next \fora \tss \fora
		}{
			\Inf \Sup \Fas \fora \tss \Inf \Sup \Fa \Fas \fora
		}$
		\\[.75em]
		$\infsupquasiprefprule{
			\fora \tss \Next \fora
		}{
			\Inf \Sup \Fa \Fas \fora \tss \Inf \Sup \Fas \fora
		}$
	\end{center}%
	\normalrules%
\end{infbox}%
\end{minipage}%
\end{wrapfigure}%

Similarly to the proof of \textsc{\small\tkpfpname}, it is useful to first prove that $\Fas \fora$ is a \emph{mino-majoral quasi fixed point of $\Fa$}, i.e.\ \emph{under the scope of a nested minorum-majorum},  it makes no difference whether or not $\funca$ is applied once more to $\Fas \fora$; formally, $\Inf\Sup \Fa \Fas \fora \hqeq \Inf\Sup \Fas \fora$.
Thus, we derive the two rules \TirName{\infsupquasipostfpname} and \TirName{\infsupquasiprefpname} above.

For \textsc{\small\infsupquasipostfpname}, consider below the left proof tree.
Here, we can expand a $\Infsymbol$ on the left, which we did not have available for proving \textsc{\small\supquasipostfpname}.
This is the algebraic essence of why Olszewski will not require $\fora \preceq \Fa \fora$ or alike.

For \textsc{\small\infsupquasiprefpname}, we have the right proof tree below.
Notice that we re-used the $\TirName{\supquasiprefpname}$ rule from the proof of TKP.
Using \textsc{\small\infsupquasipostfpname}, we prove \textsc{\small\olpostfpname} in \Cref{fig:ol-postfp-proof} (top).
Finally, with \textsc{\small\olpostfpname} at hand, the derivation of \textsc{\small\olfpname} follows immediately from \TirName{\olprefpname} which we derive also in \Cref{fig:ol-postfp-proof} (bottom), 
thus concluding our algebraic proofs of \Cref{thm:olszewski}  and \Cref{cor}.

\vspace*{-.75\abovedisplayskip}%
\footnotesizerules%
\begin{gather*}
	\infexpandrule{
		\meetintroLrulet{
			\nextinfcommrulet{
				\infmonorulet{
					\nextsupcommrulet{
						\supmonorulet{
							\desciterrulet{
								\Next \fora \tss \fora
							}{
								\Next \Fas \fora
								\tss  
								\Fa \Fas \fora
							}
						}{
							\Sup \Next \Fas \fora 
							\tss 
							\Sup \Fa \Fas \fora
						}
					}{
						\colmark{\Next \Sup \Fas \fora}
						\tss 
						\Sup \Fa \Fas \fora
					}
				}{
					\Inf \Next \Sup \Fas \fora 
					\tss 
					\Inf \Sup \Fa \Fas \fora
				}
			}{
				\colmark{\Next \Inf \Sup \Fas \fora}
				\tss 
				\Inf \Sup \Fa \Fas \fora
			}
		}{
			\Sup \Fas \fora \llmin \Next \Inf \Sup \Fas \fora 
			\tss 
			\Inf \Sup \Fa \Fas \fora
		}
	}{
		\colmark{\Inf \Sup \Fas \fora}
		\tss 
		\Inf \Sup \Fa \Fas \fora
	}~~~~
	\qquad\qquad
		\infmonorule{
		\supquasiprefprulet{
			\fora
			\tss 
			\Next \fora
		}{
			\Sup \Fa \Fas \fora
			\tss 
			\Sup \Fas \fora
		}
	}{
		\Inf \Sup \Fa \Fas \fora
		\tss 
		\Inf \Sup \Fas \fora
	}\tag*{\normalsize\qed}
	\end{gather*}%
\normalrules%
\normalsize%
\subsection{Latticed $\boldsymbol{k}$-Induction}
\label{sec:case_studies:kinduction}

\noindent%
Latticed $k$-induction \cite{kind_cav} is a fixed point-theoretic generalization of the prominent $k$-induction verification principle \cite{DBLP:conf/fmcad/SheeranSS00} for hard- and software model checking \cite{sofware_k_induction,boosting_k_induction,property_directed_k_induction,k_induction_without_unrolling}. 
Latticed $k$-induction has successfully been employed for the verification of probabilistic programs~\cite{kind_cav,DBLP:journals/corr/abs-2403-17567} and probabilistic pushdown systems \cite{DBLP:conf/tacas/WinklerK23}. 

Fix an $\omega$-continuous $\funca\colon \ladomain \to \ladomain$ and an element $\laelema \in \ladomain$. We call 
%
	\mbox{$\funckind_{\laelema} \colon  \ladomain \to \ladomain$} defined as 
	$\funckind_{\laelema}(\laelemb) = \funca(\laelemb) \lameet \laelema$
%
the \emph{$k$-Induction operator}. 
The principle of $k$-induction is a generalization of Park induction: 
\begin{theorem}[Latticed $k$-Induction~\textnormal{\cite{kind_cav}}]
	\label{thm:kind_prose}
	For every $k\in\Nats$, 
	\begin{align*}
		   \funca\Bigl(\funckind_{\laelema}^k (\laelema)\Bigr) \llaord \laelema 
		  \qqimplies
		  \lfp \funca \llaord \laelema~.
	\end{align*}
\end{theorem}

\paragraph*{AIC formalization of $\boldsymbol{k}$-Induction (\Cref{thm:kind_prose})}

Analogously to \Cref{sec:case_studies:kleene}, we associate the fixed lattice element $\laelema$ with a fixed $\forb$. 
To formalize the $k$-fold application of the $k$-induction operator, we define for all $k \in\Nats$ the term \mbox{$\Fkinditer{k} \forb$ recursively by}
\[
	\Fkinditer{k} \forb \eeq 
	\begin{cases}
		 \forb &\text{if $k=0$} \\
		 \Fa \Fkinditer{k-1} \forb ~~\lmin~~\forb  &\text{otherwise}~.
	\end{cases}
\]
The above is meant to define \emph{\underline{s}y\underline{ntactic}} equality of terms, so e.g.\ defining that $\Fkinditer{2}\forb$ is a name for the term $\Fa \Fkinditer{1} \forb \lmin\forb$.

\begin{wrapfigure}[3]{r}{.56\columnwidth}
\vspace*{-1\intextsep - .75\baselineskip}%
\begin{minipage}{.999\linewidth}
\begin{infbox}%
	\smallrules%
	\begin{center}
		$\kindrule
	{ \metaassume{\Fa~\textsc{\omegacontname}}
		\\
		\Fa \Fkinditer{k} \forb \tss \forb}
	{\Next \forb \tss \forb}
	{ \Sup \Fas \bot \tss \forb}$
	\end{center}%
	\normalrules%
\end{infbox}%
\end{minipage}%
\end{wrapfigure}%
The algebraic version of \Cref{thm:kind_prose} then reads as a family of  qua\-si-in\-e\-qual\-i\-ties \TirName{\kindname} (shown right, parametrized in $k$).
We require the additional premise $\Next \forb \tss \forb$ for the same reasons as in the formalization of \Cref{thm:tkp}. 
Classical $k$-induction  (\Cref{thm:kind_prose}) is obtained from 
\TirName{\kindname} for $\forb = \constseq{\laelema}$ for which $\Next \forb \tss \forb$ holds vacuously. We prove \TirName{\kindname} enitrely within $\AIC_1$ for every $k\in\Nats$, constructing proofs inductively on $k$~(cf.\ \appref{sec:app-k-ind}).

\section{Mechanization}
\label{sec:isabelle}

We mechanized AIC in Isabelle/HOL to confirm soundness and usability~\cite{aic-afp}.
The mechanization is a shallow embedding of sequence algebras on complete
lattices (cf.\ \Cref{sec:synsem}).
A shallow embedding means we do not
separately formalize the algebra as an object or define its syntax. Instead, we
directly build on Isabelle/HOL's formalization of complete lattices and define
the operators of AIC as functions and constants in Isabelle. We use
functions from $\omega$ to lattice elements as the type of sequences.
With this, the definition of operators is straightforward, and integrates so
well into the existing Isabelle lattice setup that for instance our $\lmin$ and
$\lmax$ coincide with existing Isabelle type class operators $\sqcup$ and
$\sqcap$ \mbox{applied to functions}.

The axioms of AIC can then be derived as lemmas in Isabelle from the definition of the operators. 
This builds up the algebra in
the prover and at the same time constitutes a soundness proof of $\AIC_0$ as
well as any derived rules (e.g.\ the axioms in $\AIC_1 \setminus \AIC_0$).

Using our formalization, we proved all fixed point
theorems of \Cref{sec:case_studies}. 
The proofs follow
the paper presentation, with intermediate steps solved fully
automatically by Isabelle's sledgehammer.
Beyond that, if the extended axiom system $\AIC_1$ is available,
\emph{sledgehammer is even able to prove many high-level theorems such as
Tarski-Kantorovich or Park induction fully automatically}. 
We also found that
Isabelle's counterexample finder is effective in identifying misstated
lemmas. Both show promise for the use of the algebra and further exploration of
fixed point principles directly in the prover. 
Conversely, if we remove
important derivable $\AIC_1$ axioms such as $\TirName{\supexpandname}$ as explicit lemma statements and force the prover to work directly with expanded
definitions and indexed suprema and infima as in traditional fixed point
formalizations, automation tends to break down quickly. 
Sledgehammer is able to
recover some lemmas by reproving them inline, but higher-level theorems remain
out of reach. 
This means that $\AIC_1$ brings with it not only increased clarity, but also \mbox{improves proof search}.

\section{Completeness and Incompleteness}
\label{sec:completeness}

In this section, we investigate aspects of completeness of the axiom systems presented in \Cref{sec:calc}.
Completeness of $\AIC_0$ means that if a finitary quasiequation is valid for sequence algebras, then it is derivable from \(\AIC_0\).
In this section, we show that $\AIC_0$ is \emph{not} complete.
In fact, even if we restrict just to the \(\First\) and \(\Next\) operations, no finite axiom system of finitary quasiequations can be complete.
This situation changes when allowing \emph{infinitary} quasiequations and a slightly infinite notion of derivation, as we will see in \Cref{sec:infinitary-ax}.

\subsection{Incompleteness of Finitary Axiomatizations}
\label{sec:incompleteness}

We now show that the validity of finitary quasiequations that just involve the \( \Firstsymbol \) and \( \Nextsymbol \) operations has no sound and complete finite axiomatization. To this end, it suffices to consider a simplified version of sequence algebras that drops all lattice-theoretic structure.  

\begin{definition}[Discrete Sequence Algebras]
A \emph{discrete sequence algebra} is given by the set~$A^\omega$ of sequences in a set $A$ equipped with the two unary operations $\Firstsymbol,\Nextsymbol \colon A^\omega \to A^\omega$ defined exactly the same as in \Cref{def:seq-algebras} (\ref{def:head shift operations}).
\end{definition}%
All syntactic concepts for sequence algebras (terms, quasiequations, finite derivations) carry over to discrete sequence algebras by restricting all definitions to $\Firstsymbol$ and $\Nextsymbol$. 
We speak of \emph{discrete} terms, quasiequations, etc.\ for distinction. 
For instance, a discrete quasiequation is a quasi{\-}equation whose terms use only $\Firstsymbol$ and $\Nextsymbol$.

A discrete quasiequation is \emph{valid for discrete sequence algebras} if it is satisfied by every discrete sequence algebra. A discrete axiom system \(\Ax\) (i.e.~a set of discrete finitary quasiequations) is \emph{sound} for discrete sequence algebras if every quasiequation from \(\Ax\) is valid for discrete sequence algebras. It is \emph{complete} for discrete sequence algebras if every quasiequation that is valid for discrete sequence algebras has a finite derivation from \(\Ax\).%
\begin{restatable}[Incompleteness]{theorem}{incompleteness}\label{thm:incompleteness}
There is no finite discrete axiom system that is both sound and complete for discrete sequence algebras.
\end{restatable}%
Intuitively, the reason is that a finite discrete axiom system cannot derive non-trivial statements about periodic sequences with large periods, since the finitely many finitary axioms necessarily only contain statements about periods of bounded length. Specifically, for each positive integer $N$, consider the discrete quasiequation%
\begin{equation}\label{eq:not-provable} 
	\Nextsymbol^N\, \fora \eeq a \quad{\wedge}\quad \bigwedge_{i=0}^{N-1}~ \First\Nextsymbol^i\, a \eeq \First\Nextsymbol^{i+1}\, a
	\qmorespace{\implies}
	\Next a \eeq a
\end{equation}
stating that if a sequence is periodic with period length $N$ and its first $N$ elements are equal, then the sequence is flat.
While~\eqref{eq:not-provable} is valid for discrete sequence algebras, one can show that for every finite and sound discrete axiom system $\Ax$, this quasiequation is not derivable from~$\Ax$ for sufficiently large $N$. 
Thus $\Ax$ is not complete.
 
Despite the simple underlying idea, the proof of \Cref{thm:incompleteness} is fairly involved; it rests on the introduction of a suitable normal form for quasiequations and a careful analysis of possible derivations from normalized quasiequations. 
For details, see \appref{app:incompleteness}.

\subsection{A Complete Infinitary Axiomatization}
\label{sec:infinitary-ax}

As we have just seen, no finite set of finitary quasiequations can completely capture validity.
For achieving completeness, we need to add possibly \emph{infinitary} quasiequations, i.e.\ \mbox{\((\bigwedge_{i \in I} {^i}s = {^i}t) \implies s = t\)} where \(I\)~may be infinite.
As the size of $I$ determines the branching degree in derivation trees, we can no longer be satisfied with our finite derivation trees (cf.\ \Cref{def:finite provability}).
It turns out that a slightly different kind of finiteness provides just the right notion of derivation for infinitary quasiequational logic.%
\begin{definition}[Well-founded Derivation]
    \label{def:well-founded}
    An inference tree is \emph{well-founded} if every branch of the tree has finite height. 
    Given an axiom system \(\Ax\) and an quasiequation \(\textup{\textsc{q}}\), a \emph{well-founded} derivation of \(\textup{\textsc{q}}\) from $\Ax$ is defined exactly as in \Cref{def:finite provability}, except that the underlying inference tree is well-founded in lieu of finite.
    A quasiequation~$\quasieq$ is \emph{well-foundedly derivable 
    from \(\Ax\)}, denoted \(\Ax\vdash_{\wf} \quasieq\), if $\quasieq$ has a well-founded derivation from $\Ax$.
\end{definition}%
By K\H{o}nig's lemma, a finitely branching well-founded tree is finite.
Thus our previous notion of finite derivability (\Cref{def:finite provability}) coincides with well-founded derivability when all of the quasiequations from \(\Ax\) are finitary.
In particular, since all \(\AIC_0\) axioms are finitary, a quasiequation \(\quasieq\) is derivable from \(\AIC_0\) if and only if \(\AIC_0 \vdash_\wf \quasieq\).

Importantly, we also cannot be satisfied with inference trees of finite height, even with infinitary branching.
Indeed, \(\Ax\) could contain \((\bigwedge_{i \in I} {^i}s = {^i}t) \implies s = t\) with \(I = \omega\) and each identity \({^i}s = {^i}t\) could require a derivation of height \(i\), like so:

\vspace*{-1\baselineskip}%
\smallrules%
\begin{equation*}
	\inferrule{
		\inferrule*{
			\semilightgray{\text{\scriptsize (derivation height 0)}}
		}{
			{^0}s = {^0}t
		}
		\\
		\inferrule*{
			\inferrule*{
				\semilightgray{\text{\scriptsize (derivation height 1)}}
			}{
				\semilightgray{\cdots}
			}
		}{
			{^1}s = {^1}t
		}
		\\
		\inferrule*{
			\inferrule*{
				\inferrule*{
					\semilightgray{\text{\scriptsize (derivation height 2)}}
				}{
					\semilightgray{\cdots}
				}
			}{
				\semilightgray{\cdots}
			}
		}{
			{^2}s = {^2}t
		}
		\\
		\semilightgray{\cdots}
	}{
		s = t
	}
\end{equation*}%
\normalrules%
This is a derivation tree of infinite height, but it also well-founded, because every branch has finite height. 
To summarize, well-founded derivations are sufficient and (in general) necessary for reasoning with infinitary inference rules.
This is the content of the theorem below, which appears to be folklore.%
\footnote{
    For the sake of convenience, a proof can be found in the appendix.
}
We state the theorem for AIC-structures, but it extends to algebras over any finitary signature.%
\begin{restatable}[Soundness and Completeness of Well-founded Derivations]{theorem}{wellfoundedcompleteness}
    \label{thm:quasiequational completeness}
    Let \(\Ax\) be an axiom system and let $\quasieq$ be a quasiequation. 
    Then \(\Ax\vdash_\wf \quasieq\) if and only if for any AIC-structure \(\Alg \models \Ax\), we also have \(\Alg \models \quasieq\).
\end{restatable}%
\begin{figure}
    \centering
    \begin{adjustbox}{max width=1\linewidth}%
    \begingroup
    \def\arraystretch{2.5}
    \small
    \begin{tabular}[row sep=20cmem]{l @{\hspace{1em}} c}
        \toprule
        \textbf{\textsf{\large Lattice}}
        & \(\begin{gathered}
            \bot \ppreceq a 
            \qquad 
            a \ppreceq \top 
            \qquad
            a\lmax b \eeq b\lmax a \qquad a\lmin b \eeq b\lmin a
            \\
            a \lmin (a \lmax b) \eeq a \eeq a \lmax (a \lmin b)  
            \qquad 
            a \ppreceq b \morespace{\Longrightarrow} Fa \ppreceq Fb
            \\
            a \lmin (b \lmin c) \eeq (a \lmin b) \lmin c 
            \qquad
            a \lmax (b \lmax c) \eeq (a \lmax b) \lmax c
            \\
        \end{gathered}\)
        \\[1.5em]
        \midrule
        \textbf{\textsf{\large Heads}}
        & \(\begin{gathered}
            \hspace*{-1.25em}
            \First \bot \eeq \bot 
            \qquad
            \First \top \eeq \top
            \qquad
            \First (a \lmin b) \eeq \First a \lmin \First b 
            \qquad
            \First (a \lmax b) \eeq \First a \lmax \First b 
            \\
            \hspace*{-1.25em}
            \First F a \eeq F \First a 
            \qquad
            \First F^* a \eeq \Next\First a \eeq \Sup\First a
            \eeq \Inf\First a \eeq \First\First a \eeq \First a 
        \end{gathered}\) 
        \\[.5em]
        \midrule
        \textbf{\textsf{\large  Shifts}}
        & \(\begin{gathered}
            \hspace*{-1.25em}
            \Next \bot \eeq \bot 
            \qquad
            \Next \top \eeq \top 
            \qquad
			\Next (a\lmax b) \eeq \Next a \lmax \Next b
			\qquad
			\Next (a\lmin b) \eeq \Next a \lmin \Next b
	    \\
            \hspace*{-1.25em}
            \Next F a \eeq F \Next a 
            \qquad
            \Next F^* a \eeq F F^* \Next a
            \qquad
            \Next \Sup a \eeq \Sup \Next a 
            \qquad
            \Next \Inf a \eeq \Inf \Next a 
        \end{gathered}\)
        \\[.5em]
        \midrule
         \pbox{20cm}{\textbf{\textsf{\large Majora \&}}\\\textbf{\textsf{\large Minora}}}
        & \(\begin{gathered}
            a \ppreceq \Sup a
            \qquad
            \Inf a \ppreceq a
            \qquad 
            \Sup\Sup a \eeq \Sup a
            \qquad 
            \Inf\Inf a \eeq \Inf a
            \\
            \Next a \ppreceq a \morespace{\Longleftrightarrow} \Sup a \ppreceq a
            \qquad
            a \ppreceq \Next a \morespace{\Longleftrightarrow} a \ppreceq \Inf a
        \end{gathered}\)
        \\[.5em]
        \midrule
        \textbf{\textsf{\large Sequence}}
        &
        \(\begin{gathered}
            \bigwedge_{n \in \omega} \First \Next^n a \ppreceq b \morespace{\Longrightarrow} \First \Sup a \ppreceq b
            \qquad
            \bigwedge_{n \in \omega} b \ppreceq \First \Next^n a \morespace{\Longrightarrow} b \ppreceq \First \Inf a
            \\
            \bigwedge_{n \in \omega} \First \Next^n a \eeq \First \Next^n b \morespace{\Longrightarrow} a \eeq b
            \\
        \end{gathered}\)
        \\[2em]
        \bottomrule
    \end{tabular}
    \endgroup
    \end{adjustbox}
    \caption{
        \label{tab:complete AIC}
        The complete infinitary axiom system \({\AIC_\omega}\). Above, ${a,b,c\in \Vars}$ are fixed variables.%
    }
\end{figure}%
\noindent%
An \emph{infinitary} axiom system \(\AIC_\omega\) that is sound and complete for sequence algebras is given in \Cref{tab:complete AIC} (for economy of space, we present the axioms as implications rather than inference rules).
Note that there are exactly three infinitary axioms, namely the sequence axioms. 
To understand these axioms, it helps to recall (cf.\ \Cref{sec:sequence-algebras}) that \mbox{-- in} sequence \mbox{algebras --} $\First\Next^n a$ represents the~$n^{\textnormal{th}}$~element of a sequence $\fora$, $\First\Sup \fora$ the supremum of $\fora$, and $\First\Inf \fora$ its infimum. 
Intuitively, the first sequence axiom states that if every element of $\fora$ is bounded above by the entirety of a sequence~$\forb$ (cf.\ $\bigwedge_{n \in \omega} \First\Next^n a \preceq b$), then the supremum of~$\fora$ is also bounded above by the entirety of~$\forb$ (cf.\ $\First\Sup a \preceq b$). 
The second sequence axiom is the dual statement for infima. 
The third sequence axiom captures that if two sequences agree at every index (cf.\ $\bigwedge_{n \in\omega} \First\Next^n a = \First\Next^n b$), then the two sequences are equal (cf.\ $a = b$). 
This is needed in the embedding theorem, \cref{thm:embedding}.
\begin{theorem}[Soundness and Completeness of $\AIC_\omega$]
    \label{thm:completeness}
For any quasiequation $\quasieq$, \[ \text{$\AIC_\omega \vdash_{\wf} \quasieq$} \qqiff \text{$\quasieq$ is valid for sequence algebras}.\]
\end{theorem}%
The proof of \Cref{thm:completeness} uses two results:
The first one is \Cref{thm:quasiequational completeness}, which tells us that to show that a quasiequation is well-foundedly derivable in \(\AIC_\omega\), it suffices to show that it is true in all models (not just sequence algebras) of \(\AIC_\omega\).
The second result is the \emph{\(\AIC_\omega\) embedding theorem} below, which tells us that every AIC-structure that satisfies \(\AIC_\omega\) embeds into some sequence algebra.%
\begin{restatable}[\(\AIC_\omega\)-Embedding]{theorem}{embeddingtheorem}
    \label{thm:embedding}
    Let \(\Alg_0 \models \AIC_\omega\) be an AIC-structure.
    Then there exists a complete lattice \((L,\, {\sqsubseteq})\) and a sequence algebra \(\Alg\) carried by \(L^\omega\) such that \(\Alg_0\) \emph{embeds} into \(\Alg\), in the sense that there exists an injective algebra homomorphism \(\varepsilon \colon \Alg_0 \hookrightarrow \Alg\).
\end{restatable}%
\begin{proof}[Proof Sketch for \Cref{thm:embedding}]
    Given an AIC-structure \(\Alg_0 \models \AIC_\omega\), we first need to find a complete lattice \(L\), which later will carry the sequence algebra \(\Alg\).
    We start by defining a (possibly incomplete) lattice \(L_0\) as follows: let \(A_0\) carry \(\Alg_0\), and define \(L_0 = \{\Firstsymbol^{\gray{\Alg_0}} x \mid x \in A_0\}\).
    Note that \(L_0 \subseteq A_0\).
    From the first row of the \textsf{\textbf{\small Heads}} axioms in \Cref{tab:complete AIC}, we see that \(L_0\) is a lattice when equipped with the lattice operations of \(\Alg_0\): given \(x,y \in A_0\), we define 
	\begin{equation*}
		\top^{\gray{L_0}} 
		\eeq 
		\Firstsymbol^{\gray{\Alg_0}}\, \top^{\gray{\Alg_0}}
		\qqand
		\Firstsymbol^{\gray{\Alg_0}}\, x 
		\morespace{\sqcap^{\gray{L_0}}} 
		\Firstsymbol^{\gray{\Alg_0}}\, y
		\eeq
		\Firstsymbol^{\gray{\Alg_0}}\, (x \mathbin{{\lmin}\!{}^{\gray{\Alg_0}}} y)
    \end{equation*}
    and similarly for \(\bot{\!}^{\gray{L_0}}\) and \(\sqcup^{\gray{L_0}}\).
    Furthermore, for each \(F \in \Funcs\), the operation \(F^{\gray{\Alg_0}}\) restricts to a monotone map \(F^{\gray{L_0}} \colon L_0 \to L_0\) defined by \(
        F^{\gray{L_0}}(\First^{\gray{\Alg_0}} x) = \First^{\gray{\Alg_0}} F^{\gray{\Alg_0}}(x)
    \) 
    for each \(x \in A_0\) by the bottom-left equation in the \textsf{\textbf{\small Heads}} axioms of \Cref{tab:complete AIC}.

    The only issue with stopping here and declaring that \(L = L_0\) is that \(L_0\) may be incomplete. 
    This is where we make use of a theorem of Macneille~\cite{MacNeil36}, which states that every lattice~\(L_0\) admits an {order-embedding}
    into a complete lattice \(L\) called the \emph{Dedekind-Macneille completion of \(L_0\)}.
    So, let \(m \colon L_0 \hookrightarrow L\) be the Dedekind-Macneille completion of \(L_0\).
    Every monotone map \(F_0 \colon L_0 \to L_0\) extends to a (not necessarily unique) monotone map \(F \colon L \to L\), in the sense that \(F \circ m = m \circ F_0\); see~\cite[Proposition 3.3]{TheunissenV07}.

    Thus, we arrive at the sequence algebra \(\Alg\) and the injective algebra homomorphism \(\Alg_0 \hookrightarrow \Alg\) that we were looking for.
    The sequence algebra \(\Alg\) is carried by \(L^\omega\), where \(L\) is the Dedekind-Macneille completion of \(L_0\), and for each function symbol \(F \in \Funcs\), the operation~\(F^{\gray{\Alg}}\) is the one obtained from some extension of \(F^{\gray{L_0}} \colon L_0 \to L_0\) to a monotone map \(F^{\gray{L}} \colon L \to L\) (see \Cref{def:seq-algebras}).
    The embedding \(\varepsilon \colon A_0 \to L^\omega\) is then obtained by forming sequences using the head and shift operators: for any \(x \in A_0\) and \(n \in \omega\), we set
    \(
        \varepsilon(x)_n = m\big(\First \Next^n x\big)
    \).
    Since \(m\) is an order-embedding, it follows from the last axiom in the \textsf{\textbf{\small Sequence}} axioms of \Cref{tab:complete AIC} that \(\varepsilon\) is injective.
    Verifying that \(\varepsilon\) is an algebra homomorphism amounts to a long calculation; see \appref{app:infinitary-ax} for details.
\end{proof}

\begin{proof}[Proof of~\Cref{thm:completeness}.]
    Soundness (\(\Longrightarrow\)) holds by \Cref{thm:quasiequational completeness} and the observation that all axioms in \(\AIC_\omega\) are valid for sequence algebras. For completeness (\(\Longleftarrow\)), consider any quasiequation~$\quasieq$ and
    suppose that for any sequence model \(\Alg\), \(\Alg \models \quasieq\).
    We need to show that~\mbox{\(\AIC_\omega \vdash_{\wf} \quasieq\)}.
    By \Cref{thm:quasiequational completeness}, it suffices to argue that \(\Alg_0 \models \quasieq\) for any AIC-structure~\(\Alg_0\) that satisfies \(\AIC_\omega\).
    But given any such \(\Alg_0\), by \Cref{thm:embedding} there is an injective algebra homomorphism \(\varepsilon \colon \Alg_0 \hookrightarrow \Alg\) into some sequence algebra \(\Alg\). 
    Quasiequations are universally quantified, so whenever an algebra satisfies some quasiequation, so do all its subalgebras. Therefore, from our assumption that \(\Alg \models \quasieq\) it follows that \(\Alg_0 \models \quasieq\), as desired.
\end{proof}
One could also ask for an axiomatization of the class of sequence algebras itself: is there a set $\Ax$ of quasiequations such that an AIC-structure satisfies $\Ax$ iff it is a sequence algebra? Such an axiomatization does not exist: by the quasivariety theorem~\cite[Thm.~2.25]{BurrisSankappanavar1981}, any class of algebras with a quasiequational axiomatization is closed under subalgebras. But sequence algebras are not; e.g., every sequence algebra has the AIC-subalgebra given by all eventually constant sequences, which is generally not isomorphic to a sequence algebra.

\section{Limitations and Outlook}
\label{sec:concl}

AIC currently cannot reason about \emph{transfinite} fixed point iterations. 
If it could, this might enable reasoning about fixed points of \emph{arbitrary monotonic} endomaps, e.g.\ by algebraizing the constructive version of the Knaster-Tarski theorem by Cousot \& Cousot~\cite{cousot1979constructive}.

As for a connection to modal logics, linear time temporal logic (\textsf{LTL}) \cite{DBLP:conf/focs/Pnueli77} seems to be closely related to the sequence-part of AIC.
While the purpose of \textsf{LTL} is completely different and it does not have iteration modalities, it would still be interesting to investigate how strong the connection between \textsf{LTL} and AIC really is.
The fact that both AIC and \textsf{LTL} admit expansion laws as well as absorption laws a l\`{a} $\Sup \Inf \Sup \fora = \Inf \Sup \fora$ suggests a strong connection.

In verification, one could investigate whether Olszewski-type fixed points admit a temporal interpretation a l\`{a} \emph{\enquote{\,$\Inf \Sup \Fas \fora$ are the states that always eventually reach $\fora$ by taking $\funca$-transitions}}.
Another endeavor in verification is to study whether Olszewski-type fixed points can help to better understand why techniques like $k$-induction~\cite{DBLP:conf/fmcad/SheeranSS00,kind_cav} work rather well and how generalization in IC3~\cite{DBLP:conf/vmcai/Bradley11} can be understood from a fixed point theoretic point of view in order to transfer this technique to more general settings~\cite{DBLP:conf/cav/KoriUKSH22,DBLP:conf/cav/BatzJKKMS20}.

Another promising direction for future work is to study non-standard (i.e.\ non-sequence) models of AIC. From \Cref{thm:completeness} applied to the axiom system \(\AIC_0\) and the incompleteness result of \Cref{sec:incompleteness}, it follows that many such models exist. We leave it as future work to discover interesting examples of non-standard models. Particularly interesting would be models where the inhabitants can also have a sort of \enquote{orientation} or \enquote{polarity} \`{a} la ascending/descending.
Possible candidates for such structures would be formal languages, trees, directed sets, continuous real-valued functions, or permutations.



\bibliography{literature}

\begin{thebibliography}{10}

\bibitem{abramsky1994domain}
Samson Abramsky and Achim Jung.
\newblock Domain theory.
\newblock In {\em Handbook of Logic In Computer Science}, volume~3. Clarendon
  Press, 1994.

\bibitem{amv07}
Jir{\'{\i}} Ad{\'{a}}mek, Stefan Milius, and Jir{\'{\i}} Velebil.
\newblock What are iteration theories?
\newblock In {\em {MFCS}}, Lecture Notes in Computer Science, pages 240--252.
  Springer, 2007.

\bibitem{DBLP:books/daglib/0020348}
Christel Baier and Joost{-}Pieter Katoen.
\newblock {\em Principles of model checking}.
\newblock {MIT} Press, 2008.

\bibitem{DBLP:journals/lmcs/BaldanEKP23}
Paolo Baldan, Richard Eggert, Barbara K{\"{o}}nig, and Tommaso Padoan.
\newblock Fixpoint theory - upside down.
\newblock {\em Log. Methods Comput. Sci.}, 19(2), 2023.

\bibitem{DBLP:journals/iandc/BaldanKP24}
Paolo Baldan, Barbara K{\"{o}}nig, and Tommaso Padoan.
\newblock Systems of fixpoint equations: Abstraction, games, up-to techniques
  and local algorithms.
\newblock {\em Inf. Comput.}, 301:105233, 2024.

\bibitem{kind_cav}
Kevin Batz, Mingshuai Chen, Benjamin~Lucien Kaminski, Joost-Pieter Katoen,
  Christoph Matheja, and Philipp Schr{\"{o}}er.
\newblock Latticed k-induction with an application to probabilistic programs.
\newblock In {\em {CAV} {(2)}}, volume 12760 of {\em Lecture Notes in Computer
  Science}, pages 524--549. Springer, 2021.

\bibitem{DBLP:conf/cav/BatzJKKMS20}
Kevin Batz, Sebastian Junges, Benjamin~Lucien Kaminski, Joost{-}Pieter Katoen,
  Christoph Matheja, and Philipp Schr{\"{o}}er.
\newblock Pric3: Property directed reachability for mdps.
\newblock In {\em {CAV} {(2)}}, Lecture Notes in Computer Science, pages
  512--538. Springer, 2020.

\bibitem{aic-conference-version}
Kevin Batz, Benjamin~Lucien Kaminski, Lucas Kehrer, Gerwin Klein, Todd Schmid,
  and Henning Urbat.
\newblock The algebra of iterative constructions.
\newblock In {\em {LICS}}, LIPIcs. Schloss Dagstuhl - Leibniz-Zentrum f{\"{u}}r
  Informatik, 2026.
\newblock [to appear].

\bibitem{aic-afp}
Kevin Batz, Benjamin~Lucien Kaminski, Lucas Kehrer, Gerwin Klein, Henning
  Urbat, and Todd Schmid.
\newblock The algebra of iterative constructions.
\newblock {\em Archive of Formal Proofs}, April 2026.
\newblock \url{https://isa-afp.org/entries/Iteration_Algebra.html}, Formal
  proof development.

\bibitem{boosting_k_induction}
Dirk Beyer, Matthias Dangl, and Philipp Wendler.
\newblock Boosting k-induction with continuously-refined invariants.
\newblock In {\em {CAV} {(1)}}, volume 9206 of {\em Lecture Notes in Computer
  Science}, pages 622--640. Springer, 2015.

\bibitem{bloom-esik93}
Stephen~L. Bloom and Zolt{\'{a}}n {\'{E}}sik.
\newblock {\em Iteration Theories - The Equational Logic of Iterative
  Processes}.
\newblock {EATCS} Monographs on Theoretical Computer Science. Springer, 1993.
\newblock \href {https://doi.org/10.1007/978-3-642-78034-9}
  {\path{doi:10.1007/978-3-642-78034-9}}.

\bibitem{DBLP:conf/vmcai/Bradley11}
Aaron~R. Bradley.
\newblock Sat-based model checking without unrolling.
\newblock In {\em {VMCAI}}, Lecture Notes in Computer Science, pages 70--87.
  Springer, 2011.

\bibitem{BurrisSankappanavar1981}
Stanley~N. Burris and H.~P. Sankappanavar.
\newblock {\em A Course in Universal Algebra}, volume~78 of {\em Graduate Texts
  in Mathematics}.
\newblock Springer, New York, 1981.

\bibitem{DBLP:conf/focs/Clarke77}
Edmund~M. Clarke.
\newblock Program invariants as fixed points.
\newblock In {\em {FOCS}}, pages 18--29. {IEEE} Computer Society, 1977.

\bibitem{cousot1979constructive}
Patrick Cousot and Radhia Cousot.
\newblock Constructive versions of tarski’s fixed point theorems.
\newblock {\em Pacific journal of Mathematics}, 82(1):43--57, 1979.

\bibitem{sofware_k_induction}
Alastair~F. Donaldson, Leopold Haller, Daniel Kroening, and Philipp
  R{\"{u}}mmer.
\newblock Software verification using k-induction.
\newblock In {\em {SAS}}, volume 6887 of {\em Lecture Notes in Computer
  Science}, pages 351--368. Springer, 2011.

\bibitem{DBLP:journals/corr/abs-1907-10381}
Frank M.~V. Feys and Helle~Hvid Hansen.
\newblock Arrow's theorem through a fixpoint argument.
\newblock In {\em {TARK}}, volume 297 of {\em {EPTCS}}, pages 175--188, 2019.

\bibitem{gunter1992semantics}
Carl~A Gunter.
\newblock {\em Semantics of programming languages: structures and techniques}.
\newblock MIT press, 1992.

\bibitem{k_induction_without_unrolling}
Arie Gurfinkel and Alexander Ivrii.
\newblock K-induction without unrolling.
\newblock In {\em {FMCAD}}, pages 148--155. {IEEE}, 2017.

\bibitem{Hirschmonomaps}
M.~W. Hirsch and Hal Smith.
\newblock Monotone maps: a review.
\newblock {\em Journal of Difference Equations and Applications},
  11(4-5):379--398, 2005.
\newblock \href {https://doi.org/10.1080/10236190412331335445}
  {\path{doi:10.1080/10236190412331335445}}.

\bibitem{jachymski2000tarski}
Jacek Jachymski, Leslaw Gajek, and Piotr Pokarowski.
\newblock The tarski--kantorovitch prinicple and the theory of iterated
  function systems.
\newblock {\em Bulletin of the Australian Mathematical Society},
  61(2):247--261, 2000.

\bibitem{property_directed_k_induction}
Dejan Jovanovic and Bruno Dutertre.
\newblock Property-directed k-induction.
\newblock In {\em {FMCAD}}, pages 85--92. {IEEE}, 2016.

\bibitem{kiwerski2024arithmeticinterpolationfactorizationamalgams}
Tomasz Kiwerski and Jakub Tomaszewski.
\newblock Arithmetic, interpolation and factorization of amalgams, 2024.
\newblock URL: \url{https://arxiv.org/abs/2401.05526}, \href
  {https://arxiv.org/abs/2401.05526} {\path{arXiv:2401.05526}}.

\bibitem{DBLP:conf/csl/000124}
Barbara K{\"{o}}nig.
\newblock Approximating fixpoints of approximated functions (invited talk).
\newblock In {\em {CSL}}, volume 288 of {\em LIPIcs}, pages 4:1--4:1. Schloss
  Dagstuhl - Leibniz-Zentrum f{\"{u}}r Informatik, 2024.

\bibitem{DBLP:conf/cav/KoriUKSH22}
Mayuko Kori, Natsuki Urabe, Shin{-}ya Katsumata, Kohei Suenaga, and Ichiro
  Hasuo.
\newblock The lattice-theoretic essence of property directed reachability
  analysis.
\newblock In {\em {CAV} {(1)}}, Lecture Notes in Computer Science, pages
  235--256. Springer, 2022.

\bibitem{kozenKA}
D.~Kozen.
\newblock A completeness theorem for kleene algebras and the algebra of regular
  events.
\newblock {\em Information and Computation}, 110(2):366--390, 1994.
\newblock \href {https://doi.org/10.1006/inco.1994.1037}
  {\path{doi:10.1006/inco.1994.1037}}.

\bibitem{kozen82}
Dexter Kozen.
\newblock Results on the propositional {\(\mathrm{\mu}\)}-calculus.
\newblock In Mogens Nielsen and Erik~Meineche Schmidt, editors, {\em Automata,
  Languages and Programming, 9th Colloquium, Aarhus, Denmark, July 12-16, 1982,
  Proceedings}, volume 140 of {\em Lecture Notes in Computer Science}, pages
  348--359. Springer, 1982.
\newblock \href {https://doi.org/10.1007/BFB0012782}
  {\path{doi:10.1007/BFB0012782}}.

\bibitem{DBLP:journals/ipl/LassezNS82}
Jean-Louis Lassez, V.~L. Nguyen, and Liz Sonenberg.
\newblock Fixed point theorems and semantics: {A} folk tale.
\newblock {\em Information Processing Letters}, 14(3):112--116, 1982.

\bibitem{MacNeil36}
H.~M. MacNeille.
\newblock Partially ordered sets.
\newblock 1936.
\newblock \href {https://doi.org/10.1090/S0002-9947-1937-1501929-X}
  {\path{doi:10.1090/S0002-9947-1937-1501929-X}}.

\bibitem{DBLP:journals/ipl/Misra01}
Jayadev Misra.
\newblock A walk over the shortest path: Dijkstra's algorithm viewed as
  fixed-point computation.
\newblock {\em Inf. Process. Lett.}, 77(2-4):197--200, 2001.

\bibitem{doi:10.1073/pnas.36.1.48}
John~F. Nash.
\newblock Equilibrium points in $n$-person games.
\newblock {\em Proceedings of the National Academy of Sciences}, 36(1):48--49,
  1950.
\newblock \href {https://doi.org/10.1073/pnas.36.1.48}
  {\path{doi:10.1073/pnas.36.1.48}}.

\bibitem{olszewski21_order}
Wojciech Olszewski.
\newblock On convergence of sequences in complete lattices.
\newblock {\em Order}, 38:251--255, 2021.

\bibitem{Ol:omegaSequences}
Wojciech Olszewski.
\newblock On sequences of iterations of increasing and continuous mappings on
  complete lattices.
\newblock {\em Games and Economic Behavior}, 126:453--459, 2021.

\bibitem{park}
David Park.
\newblock Fixpoint induction and proofs of program properties.
\newblock {\em Machine Intelligence}, 5, 1969.

\bibitem{DBLP:conf/focs/Pnueli77}
Amir Pnueli.
\newblock The temporal logic of programs.
\newblock In {\em 18th Annual Symposium on Foundations of Computer Science,
  Providence, Rhode Island, USA, 31 October - 1 November 1977}, pages 46--57.
  {IEEE} Computer Society, 1977.
\newblock \href {https://doi.org/10.1109/SFCS.1977.32}
  {\path{doi:10.1109/SFCS.1977.32}}.

\bibitem{DBLP:journals/pacmpl/RosselSKSG26}
Marcus Rossel, Rudi Schneider, Thomas Koehler, Michel Steuwer, and Andr{\'{e}}s
  Goens.
\newblock Towards pen-and-paper-style equational reasoning in interactive
  theorem provers by equality saturation.
\newblock {\em Proc. {ACM} Program. Lang.}, 10({POPL}):718--747, 2026.

\bibitem{SANTACRUZHIDALGO2024110490}
Alejandro {Santacruz Hidalgo} and Gord Sinnamon.
\newblock Core decreasing functions.
\newblock {\em Journal of Functional Analysis}, 287(4):110490, 2024.
\newblock \href {https://doi.org/10.1016/j.jfa.2024.110490}
  {\path{doi:10.1016/j.jfa.2024.110490}}.

\bibitem{DBLP:conf/fmcad/SheeranSS00}
Mary Sheeran, Satnam Singh, and Gunnar St{\aa}lmarck.
\newblock Checking safety properties using induction and a sat-solver.
\newblock In {\em {FMCAD}}, volume 1954 of {\em Lecture Notes in Computer
  Science}, pages 108--125. Springer, 2000.

\bibitem{sinnamon2007monotonicity}
Gord Sinnamon.
\newblock Monotonicity in banach function spaces.
\newblock {\em Nonlinear Analysis, Function Spaces and Applications}, pages
  205--240, 2007.

\bibitem{TheunissenV07}
Mark {Theunissen} and Yde {Venema}.
\newblock Macneille completions of lattice expansions.
\newblock 2007.
\newblock \href {https://doi.org/10.1007/s00012-007-2033-1}
  {\path{doi:10.1007/s00012-007-2033-1}}.

\bibitem{walukiewicz95}
Igor Walukiewicz.
\newblock Completeness of kozen's axiomatisation of the propositional
  mu-calculus.
\newblock In {\em Proceedings, 10th Annual {IEEE} Symposium on Logic in
  Computer Science, San Diego, California, USA, June 26-29, 1995}, pages
  14--24. {IEEE} Computer Society, 1995.
\newblock \href {https://doi.org/10.1109/LICS.1995.523240}
  {\path{doi:10.1109/LICS.1995.523240}}.

\bibitem{DBLP:conf/tacas/WinklerK23}
Tobias Winkler and Joost-Pieter Katoen.
\newblock Certificates for probabilistic pushdown automata via optimistic value
  iteration.
\newblock In {\em {TACAS} {(2)}}, volume 13994 of {\em Lecture Notes in
  Computer Science}, pages 391--409. Springer, 2023.

\bibitem{DBLP:journals/corr/abs-2403-17567}
Tengshun Yang, Hongfei Fu, Jingyu Ke, Naijun Zhan, and Shiyang Wu.
\newblock Piecewise linear expectation analysis via k-induction for
  probabilistic programs.
\newblock {\em CoRR}, abs/2403.17567, 2024.

\end{thebibliography}

\appendix


\newpage
\section{Appendix: Soundness of $\boldsymbol{\AIC_0}$}
\label{sec:app-soundness}
%
%
\normalsize%
In this section, we prove the soundness of the $\AIC_0$ axiom system (see \Cref{fig:axioms-aic0}).

\thmAICZeroSoundness*

\begin{proof}

For every axiom $\quasieq \in \AIC_0$, we have to prove that it is valid under standard interpretations (sequence algebras), i.e.\ valid according to the semantics given in \Cref{tab:semantics}.
That is, if $\quasieq$ is of the form%
%
\begin{align*}
	\nirule{q}{~\leftsuperscript{1}{\terma}{} \eeq \leftsuperscript{1}{\termb}{} \\ {\cdots} \\ \leftsuperscript{k}{\terma}{} \eeq \leftsuperscript{k}{\termb}~}{\terma \eeq \termb}~, 
\end{align*}%
we have to prove that for an arbitrary interpretation $\interpret$, we have that if $\nsemin{\leftsuperscript{i}{\terma}{}} = \nsemin{\leftsuperscript{i}{\termb}{}}$ for all $i \in \{1,\, \ldots,\, k\}$ and all $n \in \omega$, then also $\nsemin{\terma} = \nsemin{\termb}$ must hold.

In the following, we go over every axiom in $\AIC_0$.
Assume that $n$ is arbitrary but fixed.
Furthermore, we assume that the interpretation of the variables and function symbols $\interpret$ is also arbitrary but fixed and we will mostly drop the superscript ${}^\interpret$ for readability.

\bigskip%
\begin{axbox}%
	\vspace*{-1.25\abovedisplayskip}%
	\begin{align*}
		\axbot{\fora} 
		\qquad\quad
		\axtop{\fora} 
		\tag*{\verylightgray{\hyperref[fig:axioms-aic0]{\Cref{fig:axioms-aic0}, \textsf{\textbf{\small Bounded Lattice}}}}}
	\end{align*}%
	\normalsize%
\end{axbox}%
	For the first few axioms, there are no premises. We only have to prove that the conclusion always holds.
	\begin{align*}
		\semn{\bot} &\eeq \bot \tag{by \Cref{tab:semantics}} \\
		&\ssqsubseteq \semn{\fora} \tag{by $\bot$ being the least element of the underlying lattice}
	\end{align*}
	The proof of \TirName{\topname} is entirely dual by $\top$ being the greatest element of the underlying lattice.

\bigskip%
\normalsize%
\begin{axbox}
	\vspace*{-1.25\abovedisplayskip}%
	\begin{gather*}
		\axjoincomm{\fora}{\forb}
					\qquad
					\axmeetcomm{\fora}{\forb}
		\tag*{\verylightgray{\hyperref[fig:axioms-aic0]{\Cref{fig:axioms-aic0}, \textsf{\textbf{\small Bounded Lattice}}}}}
	\end{gather*}
\end{axbox}
	\begin{align*}
		\semn{\fora \lmax \forb} 	
		&\eeq \semn{\fora} \ssqcup \semn{\forb} 	\tag{by \Cref{tab:semantics}}		\\
		&\eeq \semn{\forb} \ssqcup \semn{\fora}	\tag{by commutativity of $\sqcup$}	\\
		&\eeq \semn{\forb \lmax \fora} 			\tag{by \Cref{tab:semantics}}
	\end{align*}
	The proof of \TirName{\meetcommname} is entirely analogous by commutativity of $\sqcap$.

\bigskip%
\begin{axbox}
	\vspace*{-1.25\abovedisplayskip}%
	\begin{gather*}
		\axjoinabsorb{\fora}{\forb}
		\qquad\quad
		\axmeetabsorb{\fora}{\forb}
		\tag*{\verylightgray{\hyperref[fig:axioms-aic0]{\Cref{fig:axioms-aic0}, \textsf{\textbf{\small Bounded Lattice}}}}}
	\end{gather*}
\end{axbox}
	\begin{align*}
		\semn{\fora \lmax (\fora \lmin \forb)}  
		&\eeq \semn{\fora} \ssqcup \bigl(\semn{\fora} \sqcap \semn{\forb} \bigr)	\tag{by \Cref{tab:semantics}}		\\
		&\eeq \semn{\fora} \tag{by absorption laws of complete lattices}
	\end{align*}
	The proof of \TirName{\meetabsorbname} is entirely dual, also by the absorptions laws of complete lattices.

\bigskip%
\begin{axbox}
	\vspace*{-1.25\abovedisplayskip}%
	\begin{gather*}
		\axjoinassoc{\fora}{\forb}{\forc}
		\axmeetassoc{\fora}{\forb}{\forc}
		\tag*{\verylightgray{\hyperref[fig:axioms-aic0]{Fig.\ \ref{fig:axioms-aic0}, \textsf{\textbf{\small Bound.\ Lat.}}}}}
	\end{gather*}
\end{axbox}
	\begin{align*}
		\semn{\fora \lmax (\forb \lmax \forc)}  
		&\eeq \semn{\fora} \ssqcup \bigl(\semn{\forb} \ssqcup \semn{\forc}\bigr)	\tag{by \Cref{tab:semantics}}				\\
		&\eeq \bigl( \semn{\fora} \ssqcup \semn{\forb} \bigr) \ssqcup \semn{\forc}	\tag{by associativity of $\sqcup$}	\\
		&\eeq \semn{(\fora \lmax \forb) \lmax \forc}						\tag{by \Cref{tab:semantics}}
	\end{align*}
	The proof of \TirName{\meetassocname} is entirely analogous by associativity of $\sqcap$.

\bigskip%
\begin{axbox}%
	\vspace*{-1\abovedisplayskip}%
	\begin{align*}
		\nextmonorule{\fora \tss \forb}{\Next \fora \tss \Next \forb}
		\tag*{\verylightgray{\hyperref[fig:axioms-aic0]{\Cref{fig:axioms-aic0}, \textsf{\textbf{\small Shifts}}}}}
	\end{align*}%
\end{axbox}%
This axiom has a premise, namely $\fora \preceq \forb$.
We thus assume that
\begin{align*}
	\sem{\fora}{k} \ssqsubseteq \sem{\forb}{k} \tag*{\asmcolor{(\textasm)}}
\end{align*}%
holds for all $k \in \omega$, in particular for any $n,\, n + 1 \in \omega$.
Then consider the following:%
\belowdisplayskip=0pt%
\begin{align*}
		\semn{\Next\fora}										
		&\eeq \sem{\fora}{n+1}			\tag{by \Cref{tab:semantics}}	\\
		&\ssqsubseteq \sem{\forb}{n+1}	\tag{by \textasm}			\\
		&\eeq \semn{\Next\forb}			\tag{by \Cref{tab:semantics}}
\end{align*}%
\normalsize%

\bigskip%
\begin{axbox}
	\vspace*{-1.25\abovedisplayskip}%
	\begin{gather*}
		\axnextbot{}
		\quad
		\axnexttop{}
		\tag*{\verylightgray{\hyperref[fig:axioms-aic0]{\Cref{fig:axioms-aic0}, \textsf{\textbf{\small Shifts}}}}}
	\end{gather*}
\end{axbox}
	\begin{align*}
		\semn{\Next \bot} 
		\eeq \sem{\bot}{n+1}	
		&\eeq \bot				\tag{by \Cref{tab:semantics}}		\\
		&\ssqsubseteq \bot		\tag{by reflexivity of $\sqsubseteq$}	\\
		&\eeq \semn{\bot} 		\tag{by \Cref{tab:semantics}}
	\end{align*}
	The proof of \TirName{\nexttopname} is entirely analogous.

\bigskip%
\begin{axbox}
	\vspace*{-1.25\abovedisplayskip}%
	\begin{align*}
		\axnextoverjoin{\fora}{\forb}
		\qquad
		\axnextovermeet{\fora}{\forb}
		\tag*{\verylightgray{\hyperref[fig:axioms-aic0]{\Cref{fig:axioms-aic0}, \textsf{\textbf{\small Shifts}}}}}
	\end{align*}%
\end{axbox}%
\belowdisplayskip=0pt%
\begin{align*}
	\semn{\Next(\fora \lmax \forb)} 
	\eeq \sem{\fora \lmax \forb}{n+1}	 
	\eeq \sem{\fora}{n+1} \ssqcup \sem{\forb}{n+1} 	
	\eeq \semn{\Next \fora} \ssqcup \semn{\Next \forb}
	\eeq \semn{\Next \fora \lmax \Next \forb} 
	\tag{all by \Cref{tab:semantics}}
\end{align*}%
\normalsize%
The proof of \TirName{\nextovermeetname} is entirely analogous.

\bigskip%
\begin{axbox}%
	\vspace*{-1.25\abovedisplayskip}%
	\begin{align*}
		\axsupinflate{\fora}
		\qquad
		\axinfdeflate{\fora}
		\tag*{\verylightgray{\hyperref[fig:axioms-aic0]{\Cref{fig:axioms-aic0}, \textsf{\textbf{\small Majora \& Minora}}}}}
	\end{align*}%
\end{axbox}%
\begin{align*}
	\semn{\fora} 	
	&\ssqsubseteq \sup_{n\leq k} \, \sem{\fora}{k}  \\
	&\eeq \semn{\Sup \fora}
	\tag{by \Cref{tab:semantics}}
\end{align*} 
The proof of \TirName{\infdeflatename} is entirely dual.

\begin{axbox}
	\vspace*{-1.25\abovedisplayskip}%
	\begin{align*}
		\axsupidem{\fora}
		\qquad
		\axinfidem{\fora}
		\tag*{\verylightgray{\hyperref[fig:axioms-aic0]{\Cref{fig:axioms-aic0}, \textsf{\textbf{\small Majora \& Minora}}}}}
	\end{align*}
\end{axbox}
\begin{align*}
	\semn{\Sup\Sup \fora} 
	\eeq \sup_{n\leq k}\sem{\Sup \fora}{k}	
	&\eeq \sup_{n\leq k} \, \sup_{k\leq i} \, \sem{\fora}{i}	\tag{by \Cref{tab:semantics}}	\\
	&\eeq \sup_{n\leq i} \, \sem{\fora}{i} \\
	&\eeq \semn{\Sup \fora}					\tag{by \Cref{tab:semantics}}
\end{align*}
The proof of \TirName{\infidemname} is entirely analogous.

\bigskip%
\begin{axbox}%
	\vspace*{-1\abovedisplayskip}%
	\begin{align*}
		\supmonorule{\fora \tss \forb}{\Sup \fora \tss \Sup \forb}
		\qquad 
		\infmonorule{\fora \tss \forb}{\Inf \fora \tss \Inf \forb}
		\tag*{\verylightgray{\hyperref[fig:axioms-aic0]{\Cref{fig:axioms-aic0}, \textsf{\textbf{\small Majora \& Minora}}}}}
	\end{align*}%
\end{axbox}%
By the premise of both rules, we assume that $\fora \preceq \forb$, i.e.\
\begin{align*}
	\sem{\fora}{k} \ssqsubseteq \sem{\forb}{k} \tag*{\asmcolor{(\textasm)}}
\end{align*}%
holds for all $k \in \omega$.
Then, consider the following:
\begin{align*}
		\semn{\Sup \fora} 
		& \eeq \sup_{n\leq k} \, \sem{\fora}{k}			\tag{by \Cref{tab:semantics}}	\\
		&\ssqsubseteq \sup_{n\leq k} \, \sem{\forb}{k} 	\tag{by \textasm}			\\
		&\eeq \semn{\Sup \forb}					\tag{by \Cref{tab:semantics}}
\end{align*}
The proof of \TirName{\infmononame} is entirely analogous.

\bigskip
\begin{axbox}%
	\vspace*{-1.25\abovedisplayskip}%
	\begin{align*}
		\axnextsupcomm{\fora}
		\quad
		\axnextinfcomm{\fora}
		\tag*{\verylightgray{\hyperref[fig:axioms-aic0]{\Cref{fig:axioms-aic0}, \textsf{\textbf{\small Majora \& Minora}}}}}
	\end{align*}%
\end{axbox}%
		\begin{align*}
			\semn{\Next \Sup \fora} 
			\eeq \sem{\Sup \fora}{n + 1}
			&\eeq \!\!\!\sup_{n+1 \leq k}\,\sem{\fora}{k}		\tag{by \Cref{tab:semantics}}	\\
			&\eeq \sup_{n \leq k}\,\sem{\fora}{k+1}								\\
			&\eeq \sup_{n \leq k}\,\sem{\Next \fora}{k}	
			\eeq \semn{\Sup \Next \fora}				\tag{by \Cref{tab:semantics}}
		\end{align*}
The proof of \TirName{\nextinfcommname} is entirely analogous.

\bigskip%
\begin{axbox}%
	\vspace*{-1\abovedisplayskip}
	\begin{align*}
		{
			\mprset{fraction={===}}%
			\nextindrule{\Next \fora \tss \fora}{\Sup \fora \tss \fora}%
			\qquad
			\nextcoindrule{\fora \tss \Next \fora}{\fora \tss \Inf \fora}%
		}
		\tag*{\verylightgray{\hyperref[fig:axioms-aic0]{\Cref{fig:axioms-aic0}, \textsf{\textbf{\small Majora \& Minora}}}}}
	\end{align*}%
\end{axbox}%
$\TirName{\nextindname}$ and $\TirName{\nextcoindname}$ both have double bars, meaning that both rules shall work in both directions.
We thus have four subproofs -- two directions for each of the two rules.

\proofsubparagraph{$\boldsymbol{\Supsymbol}$-\textrm{\scriptsize IND} from top to bottom:}
	By the premise of $\TirName{\nextindname}$ (top), we assume $\Next \fora \preceq \fora$, i.e.%
	\begin{align}
		\label{eq:ind:ass}
		\sem{\fora}{i+1} 
		\stackrel{\textnormal{\tiny \Cref{tab:semantics}}}{\qeq} 
		\sem{\Next a}{i}
		\qmorespace{\sqsubseteq} \sem{\fora}{i}
		\tag*{\asmcolor{(\textasm)}}
	\end{align}%
	for all $i \in \omega$. 
	To prove the conclusion of \TirName{\nextindname} (bottom), first notice that
	\begin{align*}
		\Sup \fora \tss \fora
		&\qqiff
		 \forall{n}\colon~ \sup_{n \leq k}~ \sem{\fora}{k} \llaord \sem{\fora}{n}
		 \tag{by \Cref{tab:semantics}}	\\
		 &\qqiff
		 \orange{\forall{n,j}\colon~  \sem{\fora}{n+j} \llaord \sem{\fora}{n}}
		 \tag{by definition of suprema}
	\end{align*}
	Now let $n$ be arbitrary but fixed. 
	We prove that \orange{the latter statement} holds by induction on~$j$. 
	The base case $j=0$ is trivial by reflexivity of $\sqsubseteq$. 
	For the induction step, we have
	\[
		\sem{\fora}{n+j+1} 
		\stackrel{\smash{\ref{eq:ind:ass}}}{\llaord}
		 \sem{\fora}{n+j}
		 \stackrel{\smash{\textnormal{I.H.}}}{\llaord}
		  \sem{\fora}{n}~.
	\]

\proofsubparagraph{$\boldsymbol{\Supsymbol}$-\textrm{\scriptsize IND} from bottom to top:}
	By the premise of $\TirName{\nextindname}$ (bottom), we assume $\Sup \fora \preceq \fora$, i.e.%
	\begin{align*}
		\sem{\Sup a}{i}
		\ssqsubseteq
		\sem{a}{i}
		\tag*{\asmcolor{(\textasm)}}
	\end{align*}%
	Then, to prove the conclusion of $\TirName{\nextindname}$ (top), consider%
	\begin{align*}
		\semn{\Next a} 
		&\eeq
		\sem{\fora}{n+1}			\tag{by \Cref{tab:semantics}}	\\
		&\ssqsubseteq
		\sup_{n \leq k}~ \sem{\fora}{k}	\tag{by choosing $k = n+1$}	\\
		&\eeq \semn{\Sup a}			\tag{by \Cref{tab:semantics}}	\\
		&\ssqsubseteq \semn{\fora}~. \tag{by \textasm}	
	\end{align*}%

\proofsubparagraph{$\boldsymbol{\Infsymbol}$-\textrm{\scriptsize COIND} both directions:}
The proofs of \TirName{\nextcoindname} are entirely dual for both directions.

\bigskip%
\begin{axbox}%
	\vspace*{-1.25\abovedisplayskip}%
	\begin{align*}
		\Fmonorulet{\fora \tss \forb}{\Fa \fora \tss \Fa \forb}%
		\tag*{\verylightgray{\hyperref[fig:axioms-aic0]{\Cref{fig:axioms-aic0}, \textsf{\textbf{\small Function Applications \& Orbits}}}}}
	\end{align*}%
\end{axbox}%
The rule \TirName{\Fmononame} is valid for sequence algebras \emph{if and only if} the endomap $\funca$ on the complete lattice is monotonic.
For a proof, see \Cref{sec:proof:FMono}.

\bigskip%
\begin{axbox}%
	\vspace*{-1\abovedisplayskip}%
	\begin{align*}
		\Fsmonorulet{\fora \tss \forb}{\Fs \fora \tss \Fs \forb}%
		\tag*{\verylightgray{\hyperref[fig:axioms-aic0]{\Cref{fig:axioms-aic0}, \textsf{\textbf{\small Function Applications \& Orbits}}}}}
	\end{align*}%
\end{axbox}%
By the premise of \TirName{\Fsmononame}, we assume that
\begin{align*}
	\sem{\fora}{n} \ssqsubseteq \sem{\forb}{n} \tag*{\asmcolor{(\textasm)}}
\end{align*}%
holds for all $n \in \omega$.
We then prove by induction on $k$ that for any $n$%
\begin{align*}
	\forall\, k\colon F^k \bigl(\semn{\fora} \bigr) \ssqsubseteq F^k \bigl(\semn{\forb} \bigr) \tag{$\dagger$}
\end{align*}%
holds.
The induction base $k=0$ is trivial by \textasm\ and reflexivity of $\sqsubseteq$.
As induction hypothesis, assume that $F^k \bigl(\semn{\fora} \bigr) \sqsubseteq F^k \bigl(\semn{\forb} \bigr)$ holds for arbitrary but fixed $k$.
Then we have for the induction step%
\begin{align*}
	&F^k \bigl(\semn{\fora} \bigr) \ssqsubseteq F^k \bigl(\semn{\forb} \bigr)	\tag{by I.H.}					\\
	\lqqimplies
	&F\bigl(F^k \semn{\fora} \bigr) \ssqsubseteq F\bigl(F^k \semn{\forb}\bigr)	\tag{by monotonicity of $\funca$}	\\
	\lqqiff
	&F^{k+1} \bigl(\semn{\fora} \bigr) \ssqsubseteq F^{k+1} \bigl(\semn{\forb} \bigr)
\end{align*}
which concludes the proof of ($\dagger$).
We now choose $k=n$ in ($\dagger$) and thus obtain
\begin{align*}
	& F^{n} \bigl(\semn{\fora} \bigr) \ssqsubseteq F^{n} \bigl(\semn{\forb} \bigr)							\\
	\lqqiff
	& \semn{\Fs \fora} \ssqsubseteq \semn{\Fs \forb}	\tag{by \Cref{tab:semantics}}
\end{align*}%
which is the conclusion of \TirName{\Fsmononame}.

\bigskip%
\begin{axbox}
	\vspace*{-1.25\abovedisplayskip}%
	\begin{align*}
		\axFNextcomm{\fora}%
		\tag*{\verylightgray{\hyperref[fig:axioms-aic0]{\Cref{fig:axioms-aic0}, \textsf{\textbf{\small Function Applications \& Orbits}}}}}
	\end{align*}
\end{axbox}%
	\belowdisplayskip=0pt%
	\begin{align*}
		\semn{\Next \funca \fora} 
		\eeq \sem{\funca \fora}{n+1} 
		\eeq \funca \bigl( \sem{\fora}{n+1} \bigr) 
		\eeq \funca \bigl( \semn{\Next \fora} \bigr) 
		\eeq \semn{\funca \Next \fora}
		\tag{all by \Cref{tab:semantics}}
	\end{align*}
	\normalsize%

\bigskip%
\begin{axbox}%
		\vspace*{-1.25\abovedisplayskip}%
		\begin{align*}
			\axFFscomm{\fora}%
		\tag*{\verylightgray{\hyperref[fig:axioms-aic0]{\Cref{fig:axioms-aic0}, \textsf{\textbf{\small Function Applications \& Orbits}}}}}
		\end{align*}%
	\end{axbox}%
	\belowdisplayskip=0pt%
	\begin{align*}
		\semn{\Fa \Fas \fora} 
		\eeq \Fa \bigl( \semn{\Fas \fora} \bigr) 
		&\eeq \Fa \bigl( \Fa^n \semn{\fora} \bigr)	\tag{by \Cref{tab:semantics}}	\\
		&\eeq \Fa^n \bigl( \Fa \semn{\fora} \bigr)							\\
		&\eeq \Fa^n \bigl( \semn{\Fa \fora} \bigr) 
		\eeq \semn{\Fas \Fa \fora} 
		\tag{by \Cref{tab:semantics}}
	\end{align*}%
	\normalsize%

\bigskip%
\begin{axbox}%
	\vspace*{-1.25\abovedisplayskip}%
	\begin{align*}
		\axiter{\fora}%
		\tag*{\verylightgray{\hyperref[fig:axioms-aic0]{\Cref{fig:axioms-aic0}, \textsf{\textbf{\small Function Applications \& Orbits}}}}}
	\end{align*}%
\end{axbox}%
	\begin{align*}
		\semn{\Next \Fas \fora} 
		\eeq \sem{\Fas \fora}{n+1}
		\eeq \Fa^{n+1} \bigl(\sem{\fora}{n+1}	\bigr)
		&\eeq \Fa^{n+1} \bigl(\semn{\Next \fora} \bigr)		\tag{by \Cref{tab:semantics}}	\\
		&\eeq \Fa \Bigl( \Fa^n \bigl(\semn{\Next \fora} \bigr) \Bigr)						\\
		&\eeq \Fa \Bigl(\semn{\Fas \fora} \Bigr)
		\eeq \semn{\Fa \Fas \Next \fora}				\tag{by \Cref{tab:semantics}}	
	\end{align*}

\begin{axbox}%
		\vspace*{-1\abovedisplayskip}%
		\begin{align*}
			\Findrule{\Fa \fora \tss \fora}{\Fas \fora \tss \fora} 
			\qquad
			\Fcoindrule{\fora \tss \Fa \fora}{\fora \tss \Fas \fora}%
		\tag*{\verylightgray{\hyperref[fig:axioms-aic0]{\Cref{fig:axioms-aic0}, \textsf{\textbf{\small Function Applications \& Orbits}}}}}
		\end{align*}%
	\end{axbox}
\proofsubparagraph{Proof of $\boldsymbol{\funca}$-\textrm{\scriptsize IND}:}

By the premise of \TirName{\Findname}, we assume that
\begin{align*}
	\sem{\Fa \fora}{n} \ssqsubseteq \sem{\fora}{n}\\
	\lqqimplies \funca\bigl( \sem{\fora}{n} \bigr) \ssqsubseteq \sem{\fora}{n}
	\tag*{(by \Cref{tab:semantics}) \quad \asmcolor{(\textasm)}}
\end{align*}%
holds for all $n \in \omega$.
We then prove by induction on $k$ that for any $n$%
\begin{align*}
	\forall\, k\colon F^k \bigl(\semn{\fora} \bigr) \ssqsubseteq \sem{\fora}{n} \tag{$\dagger$}
\end{align*}%
holds.
The induction base $k=0$ is trivial by \textasm\ and reflexivity of $\sqsubseteq$.
As induction hypothesis, assume that $F^k \bigl(\semn{\fora} \bigr) \sqsubseteq \semn{\fora}$ holds for arbitrary but fixed $k$.
Then we have for the induction step%
\begin{align*}
	&F^k \bigl(\semn{\fora} \bigr) \ssqsubseteq \semn{\fora}	\tag{by I.H.}																										\\
	\lqqimplies
	&F\bigl(F^k \semn{\fora} \bigr) \stackrel{\textnormal{mono}}{\sqsubseteq} F\bigl(\semn{\fora}\bigr) \stackrel{\hphantom{.}\textasm}{\sqsubseteq} \semn{\fora}	\tag{by monotonicity of $\funca$ and \textasm}	\\
	\lqqimplies
	&F\bigl(F^k \semn{\fora} \bigr) \ssqsubseteq \semn{\fora}																						\\
	\lqqiff
	&F^{k+1} \bigl(\semn{\fora} \bigr) \ssqsubseteq \semn{\fora}
\end{align*}
which concludes the proof of ($\dagger$).
We now choose $k=n$ in ($\dagger$) and thus obtain
\begin{align*}
	& F^{n} \bigl(\semn{\fora} \bigr) \ssqsubseteq \semn{\fora}						\\
	\lqqiff
	& \semn{\Fs \fora} \ssqsubseteq \semn{\fora}	\tag{by \Cref{tab:semantics}}
\end{align*}%
which is the conclusion of \TirName{\Findname}.

\proofsubparagraph{Proof of $\boldsymbol{\funca}$-\textrm{\scriptsize COIND}:}
	The proof of \TirName{\Fcoindname} is entirely dual.
\end{proof}


\section{Appendix: Axioms Pertaining to Properties of Functions}
\label{sec:app-functions}
%
%
\normalsize%
In this section, we provide omitted proofs on axioms pertaining to properties of functions like monotonicity and continuity.

\subsection{Monotonicity}
\label{sec:proof:FMono}

\begin{axbox}%
	\vspace*{-1.25\abovedisplayskip}%
	\begin{align*}
		\Fmonorulet{\fora \tss \forb}{\Fa \fora \tss \Fa \forb}%
		\tag*{\verylightgray{\hyperref[fig:axioms-aic0]{\Cref{fig:axioms-aic0}, \textsf{\textbf{\small Function Applications \& Orbits}}}}}
	\end{align*}
\end{axbox}%

\begin{theorem}[Precise Capturing of Monotonicity]
	The axiom \textnormal{\TirName{\Fmononame}} is valid for sequence algebras \emph{if and only} if the endomap~$\funca$ on the underlying complete lattice is monotonic.
\end{theorem}
\begin{proof}
	We prove both directions.
	
	\proofsubparagraph{If $\boldsymbol{\funca}$ is monotonic, then $\boldsymbol{\funca}$-\textrm{\scriptsize MONO} is valid:}
	
	Consider the following:%
	\begin{align}
					&\semn{\fora} \ssqsubseteq \semn{\forb}								\tag{by assuming the premise of \TirName{\Fmononame}}	\\
		\lqqimplies	&\funca\bigl(\semn{\fora} \bigr) \ssqsubseteq \funca \bigl(\semn{\forb} \bigr)	\tag{by assuming monotonicity of $\funca$}			\\
		\lqqiff		&\semn{\Fa \fora} \ssqsubseteq \semn{\Fa \forb}	\tag{by \Cref{tab:semantics}}
	\end{align}

	\proofsubparagraph{If $\boldsymbol{\funca}$-\textrm{\scriptsize MONO} is valid, then $\boldsymbol{\funca}$ is monotonic:}
	Let $\ell$ and $m$ be two arbitrary lattice elements with~\mbox{$\ell \laord m$}.
	From these two, form the two constant sequences $\overline{\ell} = \constseq{\ell}$ and \mbox{$\overline{m} = \constseq{m}$}.
	Then consider the following:%
	\begin{align*}
		&\ell \laord m																												\\
		\qqiff
		&\forall n\colon \semn{\,\overline{\ell}\,} \sqsubseteq \semn{\,\overline{m}\,}					\tag{by definition of $\overline{\ell}$ and $\overline{m}$}		\\
		\qqimplies
		&\forall n\colon \semn{\Fa\overline{\ell}\,} \sqsubseteq \semn{\Fa\overline{m}\,}						\tag{by assuming validity of \TirName{\Fmononame}}	\\
		\qqiff
		&\forall n\colon \funca\bigl(\semn{\,\overline{\ell}\,} \bigr) \ssqsubseteq \funca \bigl(\semn{\,\overline{m}\,} \bigr)	\tag{by \Cref{tab:semantics}}			\\
		\qqiff
		&\funca(\ell) \ssqsubseteq \funca (m)	\tag{by definition of $\overline{\ell}$ and $\overline{m}$}
	\end{align*}
	thus proving that $\funca$ is monotonic.
\end{proof}

\subsection{$\boldsymbol{\omega}$-continuity and $\boldsymbol{\omega}$-cocontinuity}
\label{sec:proof:omegaCont}

\begin{axbox}
	\normalsize%
	\begin{align*}
		\omegacontrule{\fora \tss \Next \fora}{\F \Sup \fora \tss \Sup \F \fora}%
		\qquad
		\omegacocontrule{\Next \fora \tss \fora}{\Inf \F \fora \tss \F \Inf \fora}
		\tag*{\verylightgray{\hyperref[sec:continuity-axioms]{\Cref{sec:aic1}, \textsf{\textbf{\small Continuous Functions}}}}}
	\end{align*}
\end{axbox}%
\begin{theorem}[Precise Capturing of $\omega$-(co)continuity]
$~$
	\begin{romanenumerate}
		\item The axiom \textnormal{\TirName{\omegacontname}} is valid for sequence algebras \emph{if and only} if the endomap~$\funca$ on the underlying complete lattice is $\omega$-continuous.
		\item The axiom \textnormal{\TirName{\omegacocontname}} is valid for sequence algebras \emph{if and only} if the endomap~$\funca$ on the underlying complete lattice is $\omega$-cocontinuous.
	\end{romanenumerate}
\end{theorem}%
\begin{proof}
	We only prove the case for \textnormal{\TirName{\omegacontname}}. The case for \textnormal{\TirName{\omegacocontname}} is entirely dual.
	
	\proofsubparagraph{If $\boldsymbol{\funca}$ is $\boldsymbol{\omega}$-continuous, then $\boldsymbol{\omega}$-\textrm{\scriptsize CONT} is valid:}
	
	By the premise of \TirName{\omegacontname}, we assume that%
	\begin{align*}
		\sem{\fora}{k}
		\ssqsubseteq
		\sem{\Next \fora}{k}
		\morespace{\qmorespace{\stackrel{\mathclap{\text{\tiny by \Cref{tab:semantics}}}}{=}}}
		\sem{\fora}{k+1}
	\end{align*}%
	holds for any $k \in \omega$, which captures precisely that $\bigl( \sem{\fora}{k} \bigr)_{k \in \omega}$ is an ascending chain.
	If $\bigl( \sem{\fora}{k} \bigr)_{k \in \omega}$ is an ascending chain, it is clear that for any $n \in \omega$ also%
	\begin{align*}
		\bigl( \sem{\fora}{k} \bigr)_{n \leq k}
		\qquad
		\textnormal{is an ascending chain.}
		\tag*{\asmcolor{(\textasm)}}
	\end{align*}%
	Then consider the following:%
	\begin{align*}
		\semn{\funca \Sup \fora}
		\eeq \funca \bigl( \semn{\Sup \fora} \bigr)
		&\eeq \funca \biggl( \sup_{n\leq k}~ \sem{\fora}{k} \biggr)	\tag{by \Cref{tab:semantics}}								\\
		&\ssqsubseteq \sup_{n\leq k}~ \funca \bigl( \semn{\fora} \bigr)  		\tag{by \textasm\ and assuming $\omega$-continuity of $\funca$}	\\
		&\eeq \sup_{n\leq k}~ \semn{\funca \fora} 		
		\eeq \semn{\Sup \funca \fora}						\tag{by \Cref{tab:semantics}}	
	\end{align*}%

	\proofsubparagraph{If $\boldsymbol{\omega}$-\textrm{\scriptsize CONT} is valid, then $\boldsymbol{\funca}$ is $\boldsymbol{\omega}$-continuous:}
	Assume that $\fora$ is a variable that stands for an arbitrary ascending chain, i.e.\ $\bigl( \semn{\fora} \bigr)_{n \in \omega}$ is an arbitrary ascending chain.
	Then for any $n \in \omega$, we have%
	\begin{align*}
		\sem{\fora}{n}
		{\qqmorespace{\stackrel{\mathclap{\text{\tiny $\fora$ ascending}}}{\sqsubseteq}}}
		\sem{\fora}{n+1}
		\morespace{\qmorespace{\stackrel{\mathclap{\text{\tiny by \Cref{tab:semantics}}}}{=}}}
		\sem{\Next \fora}{n}
		\tag*{\asmcolor{(\textasm)}}
	\end{align*}%
	thus fulfilling the premise of \TirName{\omegacontname}.
	
	We now have to prove that $\funca \bigl( \sup_{n \in \omega}~ \semn{\fora} \bigr) \sqsubseteq \sup_{n \in \omega}~ \funca \bigl( \semn{\fora} \bigr)$.
	Assuming that \TirName{\omegacontname} is valid, we reason for this as follows:
	\begin{align*}
		\funca \biggl( \sup_{n \in \omega}~ \semn{\fora} \biggr)
		&\eeq \funca \biggl( \sup_{0 \leq n}~ \semn{\fora} \biggr)															\\
		&\eeq \funca \bigl( \sem{\Sup \fora}{0} \bigr)			\tag{by \Cref{tab:semantics}}									\\
		&\eeq \sem{\Fa \Sup \fora}{0} 						\tag{by \Cref{tab:semantics}}									\\
		& \ssqsubseteq \sem{\Sup \Fa \fora}{0}  				\tag{by \textasm\ and assuming validity of \TirName{\omegacontname}}	\\
		& \eeq \sup_{0 \leq n}~ \sem{\Fa \fora}{n} 				\tag{by \Cref{tab:semantics}}										\\
		& \eeq \sup_{0 \leq n}~ \funca\bigl(\semn{\fora}\bigr) 		\tag{by \Cref{tab:semantics}}								\\
		& \eeq \sup_{n \in \omega} \funca\bigl(\semn{\fora}\bigr) 	\tag*{\qedhere}
	\end{align*}
\end{proof}

\subsection{Countable Continuity and Countable Cocontinuity}
\label{sec:proof:alephCont}
\begin{axbox}
	\normalsize%
	\vspace*{-1.25\abovedisplayskip}%
	\begin{align*}
		\axalephcont{\fora}
		\qquad
		\axalephcocont{\fora}
		\tag*{\verylightgray{\hyperref[sec:continuity-axioms]{\Cref{sec:aic1}, \textsf{\textbf{\small Continuous Functions}}}}}
	\end{align*}
\end{axbox}%
\begin{theorem}[Precise Capturing of Countable (Co)continuity]
$~$
	\begin{romanenumerate}
		\item The axiom \textnormal{\TirName{\alephcontname}} is valid for sequence algebras \emph{if and only} if the endomap~$\funca$ on the underlying complete lattice is countably continuous.
		\item The axiom \textnormal{\TirName{\alephcocontname}} is valid for sequence algebras \emph{if and only} if the endomap~$\funca$ on the underlying complete lattice is countably cocontinuous.
	\end{romanenumerate}
\end{theorem}%
\begin{proof}
	We only prove the case for \textnormal{\TirName{\alephcontname}}. The case for \textnormal{\TirName{\alephcocontname}} is entirely dual.
	
	\proofsubparagraph{If $\boldsymbol{\funca}$ is countably continuous, then \textrm{\scriptsize C-CONT} is valid:}
	
	For any $\fora$, clearly the set %
	\begin{align*}
		\setcomp{\sem{\fora}{k}}{n \leq k} \qquad \text{is countable}	\tag*{(\asmcolor{\textasm})}
	\end{align*}%
	Now, consider the following:%
	\begin{align*}
		\semn{\funca \Sup \fora}
		\eeq \funca \bigl( \semn{\Sup \fora} \bigr)
		&\eeq \funca \biggl( \sup_{n\leq k}~ \sem{\fora}{k} \biggr)	\tag{by \Cref{tab:semantics}}								\\
		&\ssqsubseteq \sup_{n\leq k}~ \funca \bigl( \semn{\fora} \bigr)  		\tag{by assuming countable continuity of $\funca$}	\\
		&\eeq \sup_{n\leq k}~ \semn{\funca \fora} 		
		\eeq \semn{\Sup \funca \fora}						\tag{by \Cref{tab:semantics}}	
	\end{align*}%

	\proofsubparagraph{If \textrm{\scriptsize C-CONT} is valid, then $\boldsymbol{\funca}$ is countably continuous:}
	Assume that $\fora$ is a variable that enumerates an arbitrary countable set of lattice elements $\setcomp{ \semn{\fora} }{n \in \omega}$.
	For this arbitrary set, we have to prove that $\funca \bigl( \sup_{n \in \omega}~ \semn{\fora} \bigr) \sqsubseteq \sup_{n \in \omega}~ \funca \bigl( \semn{\fora} \bigr)$.
	Assuming that \TirName{\alephcontname} is valid, we reason for this as follows:
	\begin{align*}
		\funca \biggl( \sup_{n \in \omega}~ \semn{\fora} \biggr)
		&\eeq \funca \biggl( \sup_{0 \leq n}~ \semn{\fora} \biggr)															\\
		&\eeq \funca \bigl( \sem{\Sup \fora}{0} \bigr)			\tag{by \Cref{tab:semantics}}									\\
		&\eeq \sem{\Fa \Sup \fora}{0} 						\tag{by \Cref{tab:semantics}}									\\
		& \ssqsubseteq \sem{\Sup \Fa \fora}{0}  				\tag{by \textasm\ and assuming validity of \TirName{\alephcontname}}	\\
		& \eeq \sup_{0 \leq n}~ \sem{\Fa \fora}{n} 				\tag{by \Cref{tab:semantics}}										\\
		& \eeq \sup_{0 \leq n}~ \funca\bigl(\semn{\fora}\bigr) 		\tag{by \Cref{tab:semantics}}								\\
		& \eeq \sup_{n \in \omega} \funca\bigl(\semn{\fora}\bigr) 	\tag*{\qedhere}
	\end{align*}
\end{proof}


\section{Appendix: Soundness of $\boldsymbol{\AIC_1}$}
\label{sec:app-soundness-AIC1}
%
%
\normalsize%
In this section, we prove the soundness of the $\AIC_1$ axiom system (see \Cref{fig:axioms-aic1}).

\thmAICOneSoundness*

\begin{proof}

All axioms in $\AIC_0$ are included in $\AIC_1$ and all axioms in $\AIC_0$ are valid for sequence algebras. 
Every axiom in $\AIC_1 \setminus \AIC_0$ (see \Cref{fig:axioms-aic1}) is finitely derivable from $\AIC_0$.
Thus, in the following we only have to show these derivations for the axioms in \Cref{fig:axioms-aic1}. 

For the sake of being explicit, we will (un)fold the definitions of inequations $\fora \tss \forb $ into equations by the following rules:

\vspace*{-1\baselineskip}%
\smallrules%
\begin{align*}
    \leequivrule{\forb \lmax \fora \hqeq \fora}{\forb \tss \fora}
    \quad
    \leequivrule{\forb \hqeq\fora \lmin \forb}{\forb \tss \fora}
    \qquad
    \equivlerule{\forb \tss \fora}{\forb \lmax \fora \hqeq \fora}
    \quad
    \equivlerule{\forb \tss \fora}{\forb \hqeq \fora \lmin \fora}
\end{align*}%
\normalrules%
Note that these are really just purely presentational rewritings -- not inference rules or quasiequations.
We hence use a dotted inference bar to make that distinction.

We now go over the axioms of \Cref{fig:axioms-aic1} and show derivations from $\AIC_0$ and already-derived axioms in $\AIC_1$.

\bigskip%
\begin{infbox}%
	\vspace*{-1.25\abovedisplayskip}%
	\begin{align*}
		\axreflex{\fora} 
		\tag*{\verylightgray{\hyperref[fig:axioms-aic1]{\Cref{fig:axioms-aic1}, \textsf{\textbf{\small Additional Partial Order \& Lattice Axioms}}}}}
	\end{align*}%
	\normalsize%
\end{infbox}%

\vspace*{-1\baselineskip}%
\scriptsizerules%
\belowdisplayskip=0pt%
\begin{align*}
    \leequivrulet{
        \eqtransrulet{
            \leqcongrulet{
                \axeqreflex{\fora}
            }{
                    \lequivlerule{
                        \axtop{\fora}
                    }{\fora \hqeq\top \lmin \fora}
            }{\fora \lmax \fora \hqeq \fora \lmax (\top \lmin \fora)}
        }{
        \eqtransrulet{
            \leqcongrulet{
                \axeqreflex{\fora}
            }{
                \axjoincomm{\top}{\fora}
            }{\fora \lmax (\top \lmin \fora) \hqeq \fora \lmax (\fora \lmin \top)}
        }
        {\axjoinabsorb{\fora}{\top}}
        {\fora \lmax (\top \lmin \fora) \hqeq \fora}
        }{\fora \lmax \fora \hqeq \fora}
    }
    {\fora \tss \fora}
\end{align*}
\normalrules%
\normalsize%

\bigskip%
\begin{infbox}%
	\vspace*{-1.25\abovedisplayskip}%
	\begin{align*}
		\cutrule{\fora \tss \forb}{\forb \tss \forc}{\fora \tss \forc} 
		\tag*{\verylightgray{\hyperref[fig:axioms-aic1]{\Cref{fig:axioms-aic1}, \textsf{\textbf{\small Additional Partial Order \& Lattice Axioms}}}}}
	\end{align*}%
	\normalsize%
\end{infbox}%

\vspace*{-1\baselineskip}%
\smallrules%
\begin{align*}
    \leequivrule{
        \eqtransrule{
            \fora \lmax \forc \hqeq a \lmax (\forb \lmax \forc)
        }{
            \fora \lmax (\forb \lmax \forc) \hqeq \forc           
        }{\fora \lmax \forc \hqeq \forc}
    }{\fora \tss \forc}
\end{align*}%
\normalrules%
\normalsize%
We provide derivations for each leaf separately. For the left leaf, we have:

\vspace*{-1\baselineskip}%
\smallrules%
\begin{align*}
    \eqcongrule{
                \axeqreflex{\fora}
            }{
                \eqsymmrulet{
                    \equivlerulet{
                        \sassume{\forb}{\forc}
                    }{\forb \lmax \forc \hqeq \forc}
                }{\forc \hqeq \forb \lmax \forc}
            }{\fora \lmax \forc \hqeq a \lmax (\forb \lmax \forc)}
\end{align*}%
\normalrules%
\normalsize%
For the right leaf, we have:%

\smallrules%
\vspace*{-1\baselineskip - 1\abovedisplayskip}%
\belowdisplayskip=0pt%
\begin{align*}
    \leqtransrulet{
            \leqtransrulet{
                \leqsymmrulet{
                    \axjoinassoc{\fora}{\forb}{\forc}
                }{\fora \lmax (\forb \lmax \forc) \hqeq (\fora \lmax \forb) \lmax \forc}
            }{
            \leqcongrule{
                \lequivlerulet{
                    \sassume{\fora}{\forb}
                }{\fora \lmax \forb \hqeq \forb}
            }{
                \axeqreflex{\forc}
            }{(\fora \lmax \forb) \lmax \forc \hqeq \forb \lmax \forc}
            }{\fora \lmax (\forb \lmax \forc) \hqeq \forb \lmax \forc}
            }{
            \equivlerulet{
                \sassume{\forb}{\forc}
            }{\forb \lmax \forc \hqeq \forc}
            }{\fora \lmax (\forb \lmax \forc) \hqeq \forc}
\end{align*}%
\normalrules%
\normalsize%

\pagebreak%
\begin{infbox}%
	\vspace*{-1.25\abovedisplayskip}%
	\begin{align*}
        \inequalequivrulet{
						\fora \tss \forb
						\\
						\forb \tss \fora
					}{
						\fora \hqeq \forb
					}
		\tag*{\verylightgray{\hyperref[fig:axioms-aic1]{\Cref{fig:axioms-aic1}, \textsf{\textbf{\small Additional Partial Order \& Lattice Axioms}}}}}
	\end{align*}%
	\normalsize%
\end{infbox}%

\smallrules%
\vspace*{-1\baselineskip}%
\belowdisplayskip=0pt%
\begin{align*}
    \leqtransrule{
        \lequivlerulet{
            \sassume{\fora}{\forb}
        }{\fora \hqeq \forb \lmin \fora}
    }{
        \eqsymmrulet{
            \eqtransrulet{
                \equivlerule{
                    \sassume{\forb}{\fora}
                }{\forb \hqeq \fora \lmin \forb}
            }{
                \axmeetcomm{\fora}{\forb}
            }{\forb \hqeq \forb \lmin \fora}
        }{\forb \lmin \fora \hqeq \forb}
    }{\fora \hqeq \forb}
\end{align*}%
\normalrules%
\normalsize%

\bigskip%
\begin{infbox}%
	\vspace*{-1.25\abovedisplayskip}%
	\begin{align*}
		\axjoinidem{\fora}
		\qquad\quad
        \axmeetidem{\fora}
		\tag*{\verylightgray{\hyperref[fig:axioms-aic1]{\Cref{fig:axioms-aic1}, \textsf{\textbf{\small Additional Partial Order \& Lattice Axioms}}}}}
	\end{align*}%
	\normalsize%
\end{infbox}%

\smallrules%
\vspace*{-1\baselineskip}%
\belowdisplayskip=0pt%
\begin{align*}
    \equivlerule{
        \axreflex{\fora}
    }{\fora \lmax \fora \hqeq \fora}
    \qquad\qquad
	\eqsymmrulet{
		\equivlerulet{
			\axreflex{\fora}
		}{
			\fora \hqeq \fora \lmin \fora
		}
	}{
		\fora \lmin \fora \hqeq \fora
	}
\end{align*}%
\normalrules%
\normalsize%

\bigskip%
\begin{infbox}%
	\vspace*{-1.25\abovedisplayskip}%
	\begin{align*}
		\equivinequalRrule{
						\fora \hqeq \forb
					}{
						\fora \tss \forb
					}
		\qquad\quad
        \equivinequalLrule{
						\fora \hqeq \forb
					}{
						\forb \tss \fora
					}
		\tag*{\verylightgray{\hyperref[fig:axioms-aic1]{\Cref{fig:axioms-aic1}, \textsf{\textbf{\small Additional Partial Order \& Lattice Axioms}}}}}
	\end{align*}%
	\normalsize%
\end{infbox}%
We only show a derivation for \TirName{\equivinequalRname}; the one for \TirName{\equivinequalLname} is analogous.%

\smallrules%
\vspace*{-1\baselineskip}%
\belowdisplayskip=0pt%
\begin{align*}
    \leequivrule{
        \eqtransrule{
            \eqcongrule{
                \seassume{\fora}{\forb}
            }{
                \axeqreflex{\forb}
            }{\fora \lmax \forb \hqeq \forb \lmax \forb}
        }{
            \axjoinidem{\forb}
        }{\fora \lmax \forb \hqeq \forb }
    }{\fora \tss \forb}
\end{align*}%
\normalrules%
\normalsize%

\bigskip%
\begin{infbox}%
	\vspace*{-1.25\abovedisplayskip}%
	\begin{align*}
		\joinintroRrule{
            \fora \tss \forb
        }{
            \fora \tss \forb \lmax \forc
        }%
		\qquad\quad
        \meetintroLrule{
            \fora \tss \forc
        }{\fora \lmax \forb \tss \forc}
		\tag*{\verylightgray{\hyperref[fig:axioms-aic1]{\Cref{fig:axioms-aic1}, \textsf{\textbf{\small Additional Partial Order \& Lattice Axioms}}}}}
	\end{align*}%
	\normalsize%
\end{infbox}%
We only show a derivation for \TirName{\joinintroRname}; the one for \TirName{\meetintroLname} is analogous.%

\smallrules%
\vspace*{-1\baselineskip}%
\belowdisplayskip=0pt%
\begin{align*}
    \leequivrulet{
        \eqtransrulet{
            \leqsymmrule{
                \axjoinassoc{\fora}{\forb}{\forc}
            }{\fora \lmax (\forb \lmax \forc) \hqeq (\fora \lmax \forb) \lmax \forc}
        }{
            \eqcongrulet{
                \equivlerule{
                    \sassume{\fora}{\forb}
                }{\fora \lmax \forb \hqeq \forb}
            }{
                \axeqreflex{\forc}
            }{(\fora \lmax \forb) \lmax \forc \hqeq \forb \lmax \forc}
        }{\fora \lmax (\forb \lmax \forc) \hqeq \forb \lmax \forc}
    }{\fora \tss \forb \lmax \forc}
\end{align*}%
\normalrules%
\normalsize%

\pagebreak
\begin{infbox}%
	\vspace*{-1.25\abovedisplayskip}%
	\begin{align*}
		\joinintroLrule{
            \fora \tss \forc
        }{
            \forb \tss \forc
        }{
            \fora \lmax \forb \tss \forc
        }%
		\qquad\quad
        \meetintroRrule{
            \fora \tss \forb
        }{
            \fora \tss \forc
        }{
            \fora \tss \forb \lmin \forc
        }
		\tag*{\verylightgray{\hyperref[fig:axioms-aic1]{\Cref{fig:axioms-aic1}, \textsf{\textbf{\small Additional Partial Order \& Lattice Axioms}}}}}
	\end{align*}%
	\normalsize%
\end{infbox}%
We only show a derivation for \TirName{\joinintroLname}; the one for \TirName{\meetintroRname} is analogous.%

\smallrules%
\vspace*{-1\baselineskip}%
\belowdisplayskip=0pt%
\begin{align*}
   \leequivrulet{
        \eqtransrulet{
            \axjoinassoc{\fora}{\forb}{\forc}
        }{
            \eqtransrulet{
                \eqcongrule{
                    \axeqreflex{\fora}
                }{
                    \equivlerulet{
                        \sassume{\forb}{\forc}
                    }{\forb \lmax \forc \hqeq \forc}
                }{\fora \lmax (\forb \lmax \forc) \hqeq \fora \lmax \forc}
            }{
                \equivlerulet{
                    \sassume{\fora}{\forc}
                }{\fora \lmax \forc \hqeq \forc}
            }{\fora \lmax (\forb \lmax \forc) \hqeq \forc}
        }{(\fora \lmax \forb) \lmax \forc \hqeq \forc}
    }{\fora \lmax \forb \tss \forc}
\end{align*}%
\normalrules%
\normalsize%

\bigskip%
\begin{infbox}%
	\vspace*{-1.25\abovedisplayskip}%
	\begin{align*}
		\joinelimrule{
            \fora \lmax \forb \tss \forc
        }{\forb \tss \forc}%
		\qquad\quad
        \meetelimrule{
            \forb \tss \forb \lmin \forc
        }{\fora \tss \forb}
		\tag*{\verylightgray{\hyperref[fig:axioms-aic1]{\Cref{fig:axioms-aic1}, \textsf{\textbf{\small Additional Partial Order \& Lattice Axioms}}}}}
	\end{align*}%
	\normalsize%
\end{infbox}%
We only show a derivation for \TirName{\joinelimname}, the one for \TirName{\meetelimname} is analogous.
First, we note that we can derive a \TirName{\joincommname} rule for $\preceq$ (similar to the one for $=$) as follows:

\smallrules%
\vspace*{-1\baselineskip}%
\begin{align*}
            \leequivrulet{
                \eqtransrulet{
                    \leqcongrulet{
                        \axjoincomm{\forb}{\fora}
                    }{
                        \axeqreflex{\fora \lmax \forb}
                    }{(\forb \lmax \fora) \lmax (\fora \lmax \forb) \hqeq (\fora \lmax \forb) \lmax (\fora \lmax \forb)}
                }{
                    \nequivaxiom{\joinidemname}{(\fora \lmax \forb) \lmax (\fora \lmax \forb)}{\fora \lmax \forb}
                }{(\forb \lmax \fora) \lmax (\fora \lmax \forb) \hqeq \fora \lmax \forb}
            }{\forb \lmax \fora \tss \fora \lmax \forb}
\end{align*}%
\normalrules%
\normalsize%
We then proceed as follows:%

\smallrules%
\vspace*{-1\baselineskip}%
\belowdisplayskip=0pt%
\begin{align*}
    \lcutrule{
        \lcutrulet{
        \ljoinintroRrule{
            \axreflex{\forb}
        }{\forb \tss \forb \lmax \fora}
        }{\nsaxiom{\text{$\lmax$-comm}}{\forb \lmax \fora}{\fora \lmax \forb}
        }{
            \forb \tss \fora \lmax \forb
        }
    }{
        \sassume{\fora \lmax \forb}{\forc}
    }{\forb \tss \forc}
\end{align*}%
\normalrules%
\normalsize%

\bigskip%
\begin{infbox}%
	\vspace*{-1.25\abovedisplayskip}%
	\begin{align*}
        \antisymmrule{
						\terma\subst{\fora}{\termc} \tss \termb\subst{\fora}{\termc}
					}{
						\termc \hqeq \termd
					}{
						\terma\subst{\fora}{\termd} \tss \termb\subst{\fora}{\termd}
					}
		\tag*{\verylightgray{\hyperref[fig:axioms-aic1]{\Cref{fig:axioms-aic1}, \textsf{\textbf{\small Additional Partial Order \& Lattice Axioms}}}}}
	\end{align*}%
	\normalsize%
\end{infbox}%
The following axiom follows by induction on the term structure of $\terma$, using \TirName{\eqcongname} and the basic equational lemmas.%
\begin{align*}
    \nirule{}{\termd \hqeq \termc}{ \terma\subst{\fora}{\termd} \hqeq \terma\subst{\fora}{\termc}}
\end{align*}
We then proceed as follows:

\scriptsizerules%
\vspace*{-1\baselineskip}%
\belowdisplayskip=0pt%
\begin{align*}
    \lleequivrulet{
        \leqtransrulet{
            \leqcongrulet{
                \nirulet{}{\termd \hqeq \termc}{\terma\subst{\fora}{\termd} \hqeq \terma\subst{\fora}{\termc}}
            }{
                \nirulet{}{\termd \hqeq \termc}{\termb\subst{\fora}{\termd} \hqeq \termb\subst{\fora}{\termc}}
            }{\terma\subst{\fora}{\termd} \lmax \termb\subst{\fora}{\termd} \hqeq \terma\subst{\fora}{\termc} \lmax \termb\subst{\fora}{\termc}}
        }{
            \quad\eqtransrulet{
                \lequivlerulet{
                    \terma\subst{\fora}{\termc} \tss \termb\subst{\fora}{\termc}
                }{\terma\subst{\fora}{\termc} \lmax \termb\subst{\fora}{\termc} \hqeq \termb\subst{\fora}{\termc}}
            }{
                \nirulet{}{\termd \hqeq \termc}{\termb\subst{\fora}{\termc} \hqeq \termb\subst{\fora}{\termd}}
            }{\terma\subst{\fora}{\termc} \lmax \termb\subst{\fora}{\termc} \hqeq \termb\subst{\fora}{\termd}}
        }{\terma\subst{\fora}{\termd} \lmax \termb\subst{\fora}{\termd} \hqeq \termb\subst{\fora}{\termd}}
    }{\terma\subst{\fora}{\termd} \tss \termb\subst{\fora}{\termd}}
\end{align*}%
\normalrules%
\normalsize%


\pagebreak%
\begin{infbox}
		\vspace*{-1.25\abovedisplayskip}%
		\begin{align*}
			\supintroRrule{\fora \tss \forb}{\fora \tss \Sup \forb} 
			\qquad 
			\supelimrule{\Sup \fora \tss \forb}{\fora \tss \forb} 
		\end{align*}%
		\normalsize%
		\begin{align*}
			\infintroLrule{\fora \tss \forb}{\Inf \fora  \tss \forb} 
			\qquad 
			\infelimrule{\fora \tss \Inf \forb}{\fora \tss \forb} 
            \tag*{\verylightgray{\hyperref[fig:axioms-aic1]{\Cref{fig:axioms-aic1}, \textsf{\textbf{\small Additional Majora \& Minora Axioms}}}}}
		\end{align*}%
        \normalsize
	\end{infbox}
	\label{sec:proof:SupInf:IntroBasic}
		For \textsc{$\Supsymbol$-intro} and \textsc{$\Supsymbol$-elim}, consider the following two derivations:
		
		\smallrules%
\vspace*{-1\baselineskip}%
		\begin{align*}
			\cutrule{\sassume{\fora}{\forb}}{\axsupinflate{\forb}}{\fora \tss \Sup \forb}
		\qquad\qquad
			\cutrule{\axsupinflate{\fora}}{\sassume{\Sup \fora}{\forb}}{\fora \tss \forb}
		\end{align*}%
\normalrules%
\normalsize%
		For \textsc{$\Infsymbol$-intro} and \textsc{$\Infsymbol$-elim}, consider these two derivations:
		
		\smallrules%
\vspace*{-1\baselineskip}%
\belowdisplayskip=0pt%
		\begin{align*}
			\cutrule{\axinfdeflate{\fora}}{\sassume{\fora}{\forb}}{\Inf \fora \tss \forb}
		\qquad\qquad
			\cutrule{\sassume{\fora}{\Inf \forb}}{\axinfdeflate{\forb}}{\fora \tss \forb}
		\end{align*}%
\normalrules%
\normalsize%

\bigskip%
\begin{infbox}%
		\vspace*{-1.25\abovedisplayskip}%
		\begin{align*}
			\axsupexpand{\fora}%
			\qquad\qquad
			\axinfexpand{\fora} 
            \tag*{\verylightgray{\hyperref[fig:axioms-aic1]{\Cref{fig:axioms-aic1}, \textsf{\textbf{\small Additional Majora \& Minora Axioms}}}}}
		\end{align*}%
        \normalsize
	\end{infbox}%
    We only show a derivation for \textsc{\supexpandname}; the one for \textsc{\supexpandname} is analogous.
    
    \scriptsizerules%
\vspace*{-1\baselineskip}%
\belowdisplayskip=0pt%
    \begin{align*}
		\inequalequivrulet{
		\suptightrule{
			\ljoinintroRrulet{
				\axreflex{\fora}
			}{
				\fora \tss \fora \llmax \Next \Sup \fora
			}%
		}{
			\joincommrulet{
				\joinintroRrulet{
					\nextmonorulet{
						\joinintroLrulet{
							\axsupinflate{\fora}
						}{
							\axsupdesc{\fora}
						}{
							\fora \llmax \Next \Sup \fora
							\tss 
							\Sup \fora
						}
					}{
						\Next (\fora \llmax \Next \Sup \fora)
						\tss 
						\Next \Sup \fora
					}
				}{
					\Next (\fora \llmax \Next \Sup \fora) 
					\tss 
					\Next \Sup \fora \llmax \fora
				}
			}{
				\Next (\fora \llmax \Next \Sup \fora) 
				\tss 
				\colmark{\fora \llmax \Next \Sup \fora}
			}
		}{
			\Sup \fora \tss \fora \llmax \Next \Sup \fora
		}%
		\\
		\joinintroLrulet{
			\axsupinflate{\fora}%
		}{
			\nsaxiom{$\Supsymbol$-desc}{\Next \Sup \fora}{\Sup \fora}
		}{
			\fora \llmax \Next \Sup \fora \tss \Sup \fora
		}%
		}{
			\fora \llmax \Next \Sup \fora \hqeq \Sup \fora
		}
	\end{align*}%
\normalrules%
\normalsize%

\bigskip%
\begin{infbox}
	\vspace*{-1\abovedisplayskip}%
	\begin{align*}
		\suptightrule{\fora \tss \forb}{\Next \forb \tss \forb }{\Sup \fora \tss \forb}%
	\quad
		\inftightrule{\fora \tss \Next \fora }{\fora \tss \forb}{ \fora \tss \Inf \forb}%
        \tag*{\verylightgray{\hyperref[fig:axioms-aic1]{\Cref{fig:axioms-aic1}, \textsf{\textbf{\small Additional Majora \& Minora Axioms}}}}}
	\end{align*}%
    \normalrules%
\end{infbox}

\smallrules%
\vspace*{-1\baselineskip}%
\belowdisplayskip=0pt%
\begin{align*}
	\cutrule{
		\lsupmonorulet{\fora \tss \forb}{\Sup \fora \tss \Sup \forb}
	}{
		\nextindrulet{\Next \forb \tss \forb}{\Sup \forb \tss \forb} 
	}{
		\Sup \fora \tss \forb
	}%
\qquad\qquad\qquad\qquad
	\cutrulet{
		\lnextcoindrulet{\fora \tss \Next \fora}{\fora \tss \Inf \fora}
	}{
		\infmonorulet{\fora \tss \forb}{\Inf \fora \tss \Inf \forb} 
	}{ 
		\fora \tss \Inf \forb
	}%
\end{align*}%
\normalrules%
\normalsize%

\bigskip%
\begin{infbox}%
		\vspace*{-1.25\abovedisplayskip}%
		\begin{align*}
			\axsupdesc{\fora}%
			\qquad\qquad
			\axinfasc{\fora} 
            \tag*{\verylightgray{\hyperref[fig:axioms-aic1]{\Cref{fig:axioms-aic1}, \textsf{\textbf{\small Additional Majora \& Minora Axioms}}}}}
		\end{align*}%
        \normalsize
	\end{infbox}%
    We only show the derivation for \TirName{\supdescname}; the one for \TirName{\infascname} is analogous.
	
	\smallrules%
\vspace*{-1\baselineskip}%
\belowdisplayskip=0pt%
    \begin{align*}
        \nextindrule{
            \axsupidem{\fora}
        }{\Next \Sup \fora \tss \Sup \fora}
\end{align*}%
\normalrules%
\normalsize%
    

\pagebreak
\begin{infbox}
	\vspace*{-1.25\abovedisplayskip}%
	\begin{align*}
		\axsemicont{\fora}
		\qquad
		\axsemicocont{\fora}
        \tag*{\verylightgray{\hyperref[fig:axioms-aic1]{\Cref{fig:axioms-aic1}, \textsf{\textbf{\small Additional Function Application \& Orbit Axioms}}}}}
	\end{align*}
\end{infbox}%
We only show a derivation for \TirName{\semicontname}; the one for \TirName{\semicocontname} is analogous.

\smallrules%
\vspace*{-1\baselineskip}%
\belowdisplayskip=0pt%
\begin{align*}
	\linftightrule{
		\lFNextcommrulet{
			\lFmonorulet{
				\axinfasc{\fora}
			}{ 
				\Fa \Inf \fora \tss \Fa \Next \Inf \fora
			}
		}{ 
			\Fa \Inf \fora \tss \Next \Fa \Inf \fora
		}
	}{
		\Fmonorulet{
			\axinfdeflate{\fora}
		}
		{\Fa \Inf \fora \tss \Fa \fora}
	}{ 
		\Fa \Inf \fora \tss  \Inf \Fa \fora
	}
\end{align*}%
\normalrules%
\normalsize%

\bigskip%
\begin{infbox}
	\vspace*{-1.25\abovedisplayskip}%
	\begin{align*}
		\asciterrule{
				\fora \tss \Next \fora
			}{
				\Fa \Fas \fora 
				\tss  
				\Next \Fas \fora
			}
			\qquad
			\desciterrule{
				\Next \fora \tss \fora
			}{
				\Next \Fas \fora
				\tss  
				\Fa \Fas \fora 
			}
        \tag*{\verylightgray{\hyperref[fig:axioms-aic1]{\Cref{fig:axioms-aic1}, \textsf{\textbf{\small Additional Function Application \& Orbit Axioms}}}}}
	\end{align*}
\end{infbox}%

\smallrules%
\vspace*{-1\baselineskip}%
\belowdisplayskip=0pt%
\begin{align*}
	\iterrule{
		\Fmonorulet{
			\Fsmonorulet{
				\fora 
				\tss  
				\Next \fora
			}{
				\Fas \fora 
				\tss  
				\Fas \Next \fora
			}
		}{
			\Fa \Fas \fora 
			\tss  
			\Fa \Fas \Next \fora
		}
	}{
		\Fa \Fas \fora 
		\tss  
		\colmark{\Next \Fas \fora}
	}
\qquad\qquad
	\iterrulet{
		\Fmonorulet{
			\Fsmonorulet{
				\Next \fora 
				\tss  
				\fora
			}{
				\Fas \Next \fora
				\tss  
				\Fas \fora 
			}
		}{
			\Fa \Fas \Next \fora
			\tss  
			\Fa \Fas \fora 
		}
	}{
		\colmark{\Next \Fas \fora}
		\tss  
		\Fa \Fas \fora 
	}
\end{align*}%
\normalrules%
\normalsize%

\bigskip%
\begin{infbox}
	\vspace*{-1.25\abovedisplayskip}%
	\begin{align*}
		\orbitascrule{
				\fora \tss \Fa \fora \\ \fora \tss \Next \fora
			}{
				\Fas \fora 
				\tss  
				\Next \Fas \fora
			}
			\qquad
			\orbitdescrule{
				\Next \fora \tss \fora \\ \Fa \fora \tss \fora
			}{
				\Next \Fas \fora 
				\tss  
				\Fas \fora
			}
        \tag*{\verylightgray{\hyperref[fig:axioms-aic1]{\Cref{fig:axioms-aic1}, \textsf{\textbf{\small Additional Function Application \& Orbit Axioms}}}}}
	\end{align*}
\end{infbox}%
We only show a derivation for \TirName{\orbitascname}; the one for \TirName{\orbitdescname} is analogous.

\smallrules%
\vspace*{-1\baselineskip}%
\belowdisplayskip=0pt%
\begin{align*}
    \iterrule{
        \FFscommrulet{
            \Fsmonorulet{
                \cutrulet{
                    \sassume{\fora}{\Fa \fora}
                }{
                    \Fmonorulet{
                        \sassume{\fora}{\Next \fora}
                    }{\Fa \fora \tss \Fa \Next \fora}
                }{\fora \tss \Fa \Next \fora}
            }{\Fas \fora \tss \Fas \Fa \Next \fora}
        }{\Fas \fora \tss \colmark{\Fa \Fas \Next \fora}}
    }{\Fas \fora \tss \colmark{\Next \Fas \fora}}
\end{align*}%
\normalrules%
\normalsize%

\bigskip%
\begin{infbox}
	\vspace*{-1.25\abovedisplayskip}%
	\begin{align*}
		\FsintroLrule{\fora \tss \forb}{\Fa \forb \tss \forb }{\Fas \fora \tss \forb}%
			\qquad
			\FsintroRrule{\fora \tss \Fa \fora }{\fora \tss \forb}{ \fora \tss \Fas \forb}%
        \tag*{\verylightgray{\hyperref[fig:axioms-aic1]{\Cref{fig:axioms-aic1}, \textsf{\textbf{\small Additional Function Application \& Orbit Axioms}}}}}
	\end{align*}
\end{infbox}%

\smallrules%
\vspace*{-1\baselineskip}%
\belowdisplayskip=0pt%
\begin{align*}
	\cutrule
	{\lFsmonorulet{\iassume{\fora \tss \forb}}{\Fas \fora \tss \Fas \forb}}
	{\Findrulet{\iassume{\funca \forb \tss \forb}}{\Fas \forb \tss \forb}}
	{\Fas \fora \tss \forb}
\qquad\qquad\qquad
	\cutrule
	{\lFcoindrulet{\iassume{\fora \tss \funca\fora}}{\fora \tss \Fas \fora}}
	{\Fsmonorulet{\iassume{\fora \tss \forb}}{\Fas \fora \tss \Fas \forb}}
	{\fora \tss \Fas \forb}
	\tag*{\qedhere}
\end{align*}%
\normalrules%
\normalsize%
\end{proof}


\section{Appendix: Majora-Minora Alternation Depth Hierarchy Collapse}
\label{sec:app-collapse}

Just like in LTL \cite[Figure 5.7, absorption laws]{DBLP:books/daglib/0020348}, the depth of alternating $\Supsymbol$'s and $\Infsymbol$'s collapses at level 2 in the following sense:
\begin{restatable}[Majora-Minora Alternation Depth Hierarchy Collapse]{theorem}{thmAlternationDepth}
	The majora-minora alternation depth hierarchy in AIC collapses at level 2, i.e.\ 
	\begin{align*}
		\mathsf{M}_1\, \mathsf{M}_2\, \cdots\, \mathsf{M_k}\, \Sup \cdots\, \Sup \Inf \cdots\, \Inf \fora 
		&\qeq
		\Sup \Inf \fora
		\qquad \textnormal{and}\\
		\mathsf{M}_1\, \mathsf{M}_2\, \cdots\, \mathsf{M_k}\, \Inf \cdots\, \Inf \Sup \cdots\, \Sup  \fora 
		&\qeq
		\Inf \Sup \fora
	\end{align*}
	hold, where the leading $\mathsf{M}_1\, \cdots\, \mathsf{M}_k$ are any sequence of $\Supsymbol$'s and $\Infsymbol$'s.
\end{restatable}%

\begin{proof}
	We only prove the first equation; the other one is analogous.
	Consider $b = \mathsf{M}_1\, \mathsf{M}_2\, \cdots\, \mathsf{M_k}\, \Sup \cdots\, \Sup \Inf \cdots\, \Inf \fora$.
	First, by repeatedly applying \textsc{\supidemname} and \textsc{\infidemname}, we can bring~$\forb$ into the form $\mathsf{M}_1\, \mathsf{M}_2\, \cdots\, \mathsf{M_k}\, \Sup \Inf \fora$.
	
	Consider now the innermost subterm $\mathsf{M_k}\, \Sup \Inf \fora$.
	If \mbox{$\mathsf{M_k} = \Supsymbol$}, then we can again apply \textsc{\supidemname} to obtain the equivalent term $\Sup \Inf \fora$.
	If $\mathsf{M_k} = \Infsymbol$, we can also obtain the equivalent term $\Sup \Inf \fora$ as follows:
	\vspace*{-1\abovedisplayskip}%
	\begin{align*}
		\inequalequivrule{
			\axinfdeflate{\Sup \Inf \fora}
			\\
			\nextcoindrulet{
				\nextsupcommrulet{
					\supmonorulet{
						\axinfasc{\fora}
					}{
						\Sup \Inf \fora \tss \Sup \Next \Inf \fora
					}
				}{
					\Sup \Inf \fora \tss \colmark{\Next \Sup \Inf \fora}
				}
			}{
				\Sup \Inf \fora \tss \Inf \Sup \Inf \fora
			}
		}{
			\Inf \Sup \Inf \fora \hqeq \Sup \Inf \fora
		}
	\end{align*}
	So in any case we could show that $\mathsf{M_k}\, \Sup \Inf \fora = \Sup \Inf \fora$, thus obtaining overall the term $\mathsf{M}_1\, \mathsf{M}_2\, \cdots\, \mathsf{M_{k-1}}\, \Sup \cdots\, \Sup \Inf \cdots\, \Inf \fora$, which is equivalent to $\forb$ but has one fewer majorum or minorum.
	Repeatedly applying this proof procedure to $\mathsf{M_{k-1}}$ through $\mathsf{M_1}$ finally leaves us with the desired equation $\Sup \Inf \fora = \forb$.
\end{proof}%
From this observation, we can immediately conclude that $\Sup \Inf \fora$ is flat (and so is $\Inf \Sup \fora$):%
\begin{align*}
			&\Sup \Sup \Inf \fora \tss \Inf \Sup \Inf \fora	\tag{expressing that $\Sup \Inf \fora$ is flat}	\\
	\lqqiff 	&\Sup \Inf \fora \tss \Inf \Sup \Inf \fora	\tag{by \TirName{\supidemname}}				\\
	\lqqiff 	&\Sup \Inf \fora \tss \Sup \Inf \fora~.	\tag{by alternation depth collapse}				
\end{align*}

\section{Appendix: Further Derivable Axioms}
\label{sec:app-further-axioms}

In this section, we present and derive further useful axioms that are derivable from $\AIC_1$.

\bigskip
\begin{infbox}%
	\vspace*{-1\abovedisplayskip}%
	\begin{align*}
		\mprset{fraction={===}}%
		\shiftpointrule{
			\Sup \fora \tss \Inf \fora
		}{
			\Next \fora \hqequiv \fora
		}
	\end{align*}%
\end{infbox}%
This axiom describes that both $\Sup \fora \preceq \Inf \fora$ as well as $\Next \fora = \fora$ capture equivalently that $\fora$ is flat.
For the direction from top to bottom, consider the following derivation:
	
\smallrules%
\vspace*{-1\baselineskip}%
	\begin{align*}
		\inequalequivrule{
		\nirule{$\Nextsymbol$-ind}{
			\nirulet{$\Infsymbol$-elim}{
				\sassume{\Sup \fora}{\Inf \fora}
			}{
				\Sup \fora \tss \fora
			}
		}{
			\Next \fora \tss \fora
		}
		\\
		\nextcoindrulet{
			\nirulet{$\Supsymbol$-elim}{
				\sassume{\Sup \fora}{\Inf \fora}
			}{
				\fora \tss \Inf \fora
			}
		}{
			\fora \tss  \Next \fora
		}
		}{
			\Next \fora \hqequiv \fora
		}
	\end{align*}%
\normalrules%
\normalsize%
	From bottom to top, consider the following derivation:%
	
	\smallrules%
\vspace*{-1\baselineskip}%
\belowdisplayskip=0pt%
	\begin{align*}
		\nirule{cut}{
			\nlirulet{$\Nextsymbol$-ind}{
				\lequivinequalLrulet{
					\seassume{\Next \fora}{\fora}
				}{
					\Next \fora \tss \fora
				}
			}{
				\Sup \fora \tss \fora
			}
			\\
			\nirulet{$\Nextsymbol$-coind}{
				\equivinequalLrulet{
					\seassume{\fora}{\Next \fora}
				}{
					\fora \tss \Next \fora
				}
			}{
				\fora \tss \Inf \fora
			}
		}{
			\Sup \fora \tss \Inf \fora	
		}
	\end{align*}%
\normalrules%
\normalsize%

\bigskip%
\begin{infbox}%
	\vspace*{-1\abovedisplayskip}%
	\smallrules
	\begin{align*}
		\supascpointrule{\fora \tss \Next \fora}{\Sup \Sup \fora \tss \Inf \Sup \fora}
		\qquad
		\infdescpointrule{\Next \fora \tss \fora}{\Sup \Inf \fora \tss \Inf \Inf \fora}
	\end{align*}%
	\normalrules
\end{infbox}%
This axiom describes that the majorum of an ascending chain is flat; and dually that the minorum of a descending chain is flat.
For derivations, consider the following:

\smallrules%
\vspace*{-1\baselineskip}%
\belowdisplayskip=0pt%
	\begin{align*}
		\supidemrule{
			\nextcoindrulet{
				\nextsupcommrulet{
					\supmonorulet{
						\sassume{\fora}{\Next \fora}
					}{
						\Sup \fora \tss \Sup\Next \fora
					}
				}{
					\Sup \fora \tss \colmark{\Next \Sup \fora}
				}
			}{
				\Sup \fora \tss \Inf \Sup \fora
			}
		}{
			\colmark{\Sup \Sup \fora} \tss \Inf \Sup \fora
		}
		\qquad
		\infidemrulet{
			\nextcoindrulet{
				\nextinfcommrulet{
					\infmonorulet{
						\sassume{\Next \fora}{\fora}
					}{
						\Inf \Next  \fora
						\tss 
						\Inf \fora
					}
				}{
					\colmark{\Next \Inf \fora} 
					\tss 
					\Inf \fora
				}
			}{
				\Sup \Inf \fora \tss \Inf \fora
			}
		}{
			\Sup \Inf \fora \tss \colmark{\Inf \Inf \fora}
		}
		\tag*{\normalsize\qedhere}
	\end{align*}%
	\normalrules%
	\normalsize%

\bigskip%
\begin{infbox}
	\vspace*{-1\abovedisplayskip}%
	\begin{align*}
		\supinfinfsupmonorule{\fora \tss \forb}{\Sup \Inf \fora \tss \Inf \Sup \forb}
	\end{align*}%
\end{infbox}%

\smallrules%
\vspace*{-1\baselineskip}%
\belowdisplayskip=0pt%
\begin{align*}
	\suptightrulet{
		\linfmonorulet{
			\lsupintroRrulet{
				\sassume{\fora}{\forb}
			}{
				\fora \tss \Sup \forb
			}
		}{
			\Inf \fora \tss \Inf \Sup \forb
		}
	}{
		\nextinfcommrulet{
			\infmonorulet{
				\axsupdesc{\forb}
			}{
				\Inf \Next \Sup \forb
				\tss
				\Inf \Sup \forb
			}
		}{
			\colmark{\Next \Inf \Sup \forb}
			\tss
			\Inf \Sup \forb
		}
	}{
		\Sup \Inf \fora \tss \Inf \Sup \forb
	}
	\tag*{\normalsize\qedhere}
\end{align*}%
	\normalrules%
	\normalsize%

\bigskip%
\begin{infbox}
	\vspace*{-1.25\abovedisplayskip}%
	\begin{align*}
		\axsupoverjoin{\fora}{\forb}
		\qquad\qquad
		\axinfovermeet{\fora}{\forb}
	\end{align*}%
\end{infbox}%
We only show a derivation for \TirName{\supoverjoinname}; the one for \TirName{\infovermeetname} is analogous.
We prove $\preceq$ and $\succeq$ separately.

\smallrules%
\vspace*{-1\baselineskip - \abovedisplayskip}%
\begin{align*}
		\suptightrulet
		{
			\ljoinintroLrulet
			{\lsupintroRrulet{\,}{\fora\tss \Sup \fora}}
			{\supintroRrule{\,}{\forb\tss \Sup \forb}}
			{\fora \lmax \forb \tss \Sup \fora \lmax \Sup \forb}
		}
		{
			\nextoverjoinrulet
			{
				\joinintroLrulet
				{\axsupdesc{\fora}}
				{\axsupdesc{\forb}}
				{\Next\Sup \fora \lmax \Next\Sup \forb \tss \Sup \fora \lmax \Sup \forb}
			}
			{\Next(\Sup \fora \lmax \Sup \forb) \tss \Sup \fora \lmax \Sup \forb}
		}
		{\Sup (\fora \lmax \forb) \tss \Sup \fora \lmax \Sup \forb}
	\end{align*}
	\begin{align*}
		\ljoinintroLrule
		{
			\lsupmonorulet{
				\ljoinintroRrulet{\axreflex{\fora}}
				{\fora \tss \fora \lmax \forb}
			}
			{
				\Sup \fora \tss \Sup (\fora \lmax \forb)
			}
		}
		{
			\supmonorulet{
				\joinintroRrulet{\axreflex{\forb}}
				{\forb \tss \fora \lmax \forb}
			}
			{
				\Sup \forb \tss \Sup (\fora \lmax \forb)
			}
		}
		{\Sup \fora \lmax \Sup \forb\tss\Sup (\fora \lmax \forb)}
	\end{align*}%
	\normalrules%
	\normalsize%

\bigskip%
\begin{infbox}
	\vspace*{-1.25\abovedisplayskip}%
	\begin{align*}
		\axsupovermeet{\fora}{\forb}
		\qquad
		\axinfoverjoin{\fora}{\forb}
	\end{align*}%
\end{infbox}%
We only show a derivation for \TirName{\infoverjoinname}; the one for \TirName{\supovermeetname} is analogous.

\smallrules%
\vspace*{-1\baselineskip}%
\belowdisplayskip=0pt%
	\begin{align*}
		\ljoinintroLrule
		{
			\linfmonorulet{
				\ljoinintroRrulet{\axreflex{\fora}}
				{\fora \tss \fora \lmax \forb}
			}
			{
				\Inf \fora \tss \Inf (\fora \lmax \forb)
			}
		}
		{
			\infmonorulet{
				\joincommrulet{
					\joinintroRrulet{\axreflex{\forb}}
						{\forb \tss \forb \lmax \fora}
				}{
					\forb \tss \colmark{\fora \lmax \forb}
				}
			}
			{
				\Inf \forb \tss \Inf (\fora \lmax \forb)
			}
		}
		{\Inf \fora \llmax \Inf \forb\tss\Inf (\fora \lmax \forb)}
	\end{align*}%
	\normalrules%
	\normalsize%
Note that $\Supsymbol$ indeed only \emph{sub}distributes over $\lmin$ and so does $\Inf$ over $\lmax$.
The converse inequalities do \emph{not} hold in general.

\section{Appendix: Case Studies}
\label{sec:app-case-studies}

\subsection{Olszewski Fixed Point Theorems (\Cref{sec:case_studies:ol})}
\label{app:proof:olprose}

\thmOlszewski*

\begin{proof}
Consider the following computation:%
\begin{align*}
	\funca\left(
	\inf_{n \in \omega}~ \sup_{k\geq n}~ \funca^k(\laelem)
	\right) 
	&\eeq \inf_{n \in \omega}~ \funca\left(
	\sup_{k\geq n}~ \funca^k(\laelem)
	\right) 
	\tag{by assuming $\omega$-cocontinuity of $\funca$}
	\\
	&\ssqsupseteq \inf_{n \in \omega} ~ \sup_{k\geq n}~ F\left(F^k(\laelem) \right) 
	\tag{by any monotonic $\funca$ being semi-continuous}
	\\
	&\eeq \inf_{n \in \omega}~ \sup_{k\geq n}~ F^{k+1}(\laelem)
	\\
	&\eeq \inf_n \sup_{k\geq n} F^k(\laelem)
\end{align*}%
%
For the last step, observe that the sequences $(\sup_{k\geq n}F^k(\laelem))_{n\in \omega}$ and $(\sup_{k\geq n}F^{k+1})_{n\in \omega}$ are both non-increasing, and the latter is a subsequence the former, so they have the same infimum.
\end{proof}

\subsection{$\boldsymbol{k}$-Induction (\Cref{sec:case_studies:kinduction})}
\label{sec:app-k-ind}
First recall that in the definition of the $k$-induction operator
\[
\Fkinditer{k} \forb \eeq 
\begin{cases}
	\forb &\text{if $k=0$} \\
	\Fa \Fkinditer{k-1} \forb ~~\lmin~~\forb  &\text{otherwise}~.
\end{cases}
\]
we mean to define \emph{\underline{s}y\underline{ntactic}} equality of terms, so e.g.\ defining that $\Fkinditer{2}\forb$ is a name for the term $\Fa \Fkinditer{1} \forb \lmin\forb$. 
For our proof, we require a few auxiliary lemmas. 

\bigskip%
\begin{infbox}%
	\vspace*{-1.25\abovedisplayskip}%
	\begin{align*}
        		\axkinddesc{\Fkinditer{k+1} \forb}{\Fkinditer{k} \forb}{k}
	\end{align*}%
	\normalsize%
\end{infbox}%
First, we need that applying~$\Fkinditer{k}$ once more to $\forb$ takes us down in the partial order (\TirName{\axkinddescname{k}}), which we prove by induction on $k$. 
For $k=0$, we have:

\vspace*{-1.5\abovedisplayskip}
\smallrules%
\begin{align*}
\lantisymmrule{
	\lantisymmrulet{
		\lmeetintroLrulet{
			\axreflex{\forb}
		}{
			\Fa {\Fkinditer{0} \forb} \llmin \forb \tss \forb
		}
	}{
					\Fkinditer{0} \forb \hqeq \forb
	}{
		\Fa {\Fkinditer{0} \forb} \llmin \forb \tss \colmark{\Fkinditer{0} \forb}
	}
}{
			\Fkinditer{1} \forb \hqeq \Fa \Fkinditer{0} \forb \llmin \forb
}{
	\colmark{\Fkinditer{1} \forb} \tss \Fkinditer{0} \forb
}
\end{align*}%
\normalrules%
\normalsize%
For the induction step, assume validity of \TirName{\axkinddescname{k}} for some arbitrary, but fixed, $k \geq 1$ and consider the following:
Below, we omit for conciseness of the proof tree the equalities of the form $\Fkinditer{k+1} = \Fa \Fkinditer{k} \lmin \forb$ on every \TirName{\antisymmname} rule.

\smallrules%
\vspace{-1\baselineskip}%
\begin{align*}
		\antisymmrulet{
			\lmeetintroRrulet{	
				\lantisymmrulet{
					\lmeetcommrulet
					{	\lmeetintroLrulet
						{	\lFmonorulet
							{\iIH{\Fkinditer{k}\forb \tss \Fkinditer{k-1} \forb}}
							{\Fa \Fkinditer{k} \forb \tss \Fa \Fkinditer{k-1} \forb}
						}{
							\forb \llmin \Fa \Fkinditer{k} \forb \tss \Fa \Fkinditer{k-1} \forb
						}
					}{
						\colmark{\Fa \Fkinditer{k} \forb \llmin \forb} \tss \Fa \Fkinditer{k-1} \forb
					}
				}
				{
				}{
					\colmark{\Fkinditer{k+1} \forb} 
					\tss 
					\Fa \Fkinditer{k-1} \forb
				}
			}{	
				\antisymmrulet
				{	\meetintroLrulet{
						\axreflex{\forb}
					}{
						\Fa \Fkinditer{k} \forb \lmin \forb \tss \forb
					}
				}{
				}{
					\colmark{\Fkinditer{k+1} \forb} \tss \forb
				}
			}{
				\Fkinditer{k+1} \forb 
				\tss 
				\Fa \Fkinditer{k-1} \forb \llmin \forb
			}
		}{
		}
		{\Fkinditer{k+1} \forb \tss \colmark{\Fkinditer{k} \forb}}
	\end{align*}%
	\normalrules%
	\normalsize%

\bigskip%
\begin{infbox}%
	\vspace*{-1.25\abovedisplayskip}%
	\begin{align*}
        		\axkinddescb{\Fkinditer{k} \forb}{\forb}{k}
	\end{align*}%
	\normalsize%
\end{infbox}%
Validity of \TirName{\axkinddescname{k}} (for every $k$) implies that $\Fkinditer{k} \forb \preceq \forb$.
Again, by induction on $k$: For $k=0$, the axiom is trivial. 
For the induction step, we have

\vspace*{-1\baselineskip}%
\smallrules%
\begin{align*}
	\cutrule{ 
		\axkinddesc{\Fkinditer{k+1} \forb}{\Fkinditer{k} \forb}{k}
	}{ 
		\iIH{\Fkinditer{k} \forb \tss \forb} 
	}{ 
		\Fkinditer{k} \forb \tss \forb 
	}
\end{align*}%
\normalrules%
\normalsize%

\bigskip%
\begin{infbox}%
	\vspace*{-1.25\abovedisplayskip}%
	\begin{align*}
        		\kindparkrule{ \Fa \Fkinditer{k} \forb \tss \forb}
				{\Fa \Fkinditer{k} \forb \tss \Fkinditer{k} \forb}
	\end{align*}%
	\normalsize%
\end{infbox}%
As third auxiliary lemma, we need that if $\Fa \Fkinditer{k} \forb$
%
is below $\forb$, then $\Fa \Fkinditer{k} \forb$ is even below $\Fkinditer{k} \forb$  (\TirName{\kindparkrulename}).
We prove this as by induction on $k$. 
The base case $k=0$ is trivial. 
For the induction step with $k\geq 1$, we have:

\vspace{-1\baselineskip}%
\smallrules%
\begin{align*}
			\antisymmrulet
			{	\meetintroRrulet
				{	\lFmonorulet
					{ \axkinddesc{\Fkinditer{k} \forb}{\Fkinditer{k-1} \forb }{k} }
					{\Fa \Fkinditer{k} \forb \tss \Fa\Fkinditer{k-1} \forb }
				}
				{	\iassume{\Fa \Fkinditer{k} \forb \tss \forb}}
				{\Fa \Fkinditer{k} \forb \tss \Fa \Fkinditer{k-1} \forb \lmin \forb}
			}
			{}
			{\Fa \Fkinditer{k} \forb \tss \colmark{\Fkinditer{k} \forb}}
\end{align*}%
\normalrules%
\normalsize%

\bigskip%
\begin{infbox}%
	\vspace*{-1.25\abovedisplayskip}%
	\begin{align*}
        		\kindpreserveascrule
				{\Next \forb \tss \forb}
				{\Next \Fkinditer{k} \forb \tss \Fkinditer{k} \forb}{k}
	\end{align*}%
	\normalsize%
\end{infbox}%
Finally, we require the axiom \TirName{$G_b^k$-asc-pres}, which we also prove by induction on $k$. 
The base case $k=0$ is trivial. 
For the induction step, we have: 

\vspace{-1.5\baselineskip}%
\smallrules%
		\begin{align*}
			\antisymmrulet
			{
				\lnextovermeetrulet
				{
					\lmeetintroRrulet
					{
						\lmeetintroLrulet
						{
							\lFNextcommrulet
							{
								\Fmonorulet
								{\iIH{\Next \Fkinditer{k} \forb  \tss  \Fkinditer{k} \forb}}
								{\Fa \Next \Fkinditer{k} \forb  \tss  \Fa \Fkinditer{k} \forb}
							}
							{\colmark{\Next \Fa \Fkinditer{k} \forb}  \tss  \Fa \Fkinditer{k} \forb}
						}
						{\Next \Fa \Fkinditer{k} \forb \lmin  \Next \forb \tss  \Fa \Fkinditer{k} \forb} 
					}
					{
						\meetintroLrulet
						{\iassume{\Next \forb \tss  \forb}}
						{\Next \Fa \Fkinditer{k} \forb \lmin  \Next \forb \tss  \forb}
					}
					{\Next \Fa \Fkinditer{k} \forb \lmin  \Next \forb \tss  \Fa \Fkinditer{k} \forb \lmin  \forb}
				}
				{\colmark{\Next (\Fa \Fkinditer{k} \forb \lmin \forb)} \tss  \Fa \Fkinditer{k} \forb \lmin  \forb}
			}
			{
				\Fkinditer{k+1} \forb \hqeq \Fa \Fkinditer{k} \forb \lmin \forb
			}
			{\Next \colmark{\Fkinditer{k+1} \forb} \tss \colmark{\Fkinditer{k+1} \forb}}
		\end{align*}%
\normalrules%
\normalsize%
Now, \textsc{\kindname} for $k=0$ is an immediate instance of \Cref{thm:generalized_park}. For $k\geq 1$, we employ the preceding lemmas, resulting in the following proof:

\vspace*{-1\baselineskip}
\smallrules%
\begin{align*}
				\qquad\cutrulet
				{
					\ltkpleastrulet
					{
						\axbot{\Fkinditer{k} \forb} 
						\\
						\kindparkrulet
						{\iassume{\Fa \Fkinditer{k} \forb \tss \forb}}
						{\Fa \Fkinditer{k} \forb \tss \Fkinditer{k} \forb}
						\\\quad
						\\
						\\
						\\
						\kindpreserveascrulet
						{\iassume{\Next \forb \tss \forb}}
						{\Next \Fkinditer{k} \forb \tss \Fkinditer{k} \forb \vphantom{\Bigl(} }{k}
					}
					{
						\Sup \Fas \bot \tss \Fkinditer{k} \forb
					}
				}
				{
						\axkinddescb{\Fkinditer{k} \forb}
						{\forb}
						{k} 
				}
				{
					\Sup \Fas \bot \tss \forb
				}
				\tag*{\qed}
			\end{align*}%



\section{Appendix: Incompleteness of Finitary Axiom Systems (\Cref{sec:incompleteness})}\label{app:incompleteness}
We give a detailed account of Theorem \ref{thm:incompleteness}, which states that the finitary quasiequational theory of the operators $\Firstsymbol$ and $\Nextsymbol$ on sequences has no finite axiomatization.

\subsection{Finitary Quasiequations in General}
We first recall some standard terminology and notation from universal algebra, generalizing the concepts introduced in \Cref{sec:synsem} from sequence algebras to algebras over any finitary signature.  

\begin{notation}
Fix a finitary algebraic signature $\Sigma$, i.e.\ a set of operation symbols with associated finite arities.
Moreover, fix a countably infinite set
\[\Vars \eeq \{ x,y,z,\ldots \} \] of variables, and let $\Sigmas\Vars$ denote the set of $\Sigma$-terms formed over the set $\Vars$ of variables. A $\Sigma$-algebra is a set $A$ equipped with an operation $\mathsf{f}^A\colon A^n\to A$ for every $n$-ary operation symbol in $\Sigma$. Given a $\Sigma$-algebra $A$, a variable interpretation $\interpret\colon \Vars\to A$ and a term $t\in \Sigmas\Vars$, we write $\interpretsem{t}\in A$ for the value of the term $t$ in the algebra $A$ under the given interpretation. Formally, $\interpretsem{t}$ is defined inductively by
\[ \interpretsem{x} \eeq \interpret(x) \qqand \interpretsem{\mathsf{f}(t_1,\ldots,t_n)} \eeq \mathsf{f}^A(\interpretsem{t_1},\ldots,\interpretsem{t_n})~,   \]
where $x\in \V$, $\mathsf{f}$ is an $n$-ary operation symbol from $\Sigma$, and $t_1,\ldots,t_n\in \Sigmas \V$. Finally, given a substitution (a map $\sigma\colon \Vars\to \Sigmas \Vars$) and a term $t\in \Sigmas \Vars$, we write $t[\sigma]\in \Sigmas\Vars$ for the term obtained by applying the substitution $\sigma$ to the variables of $t$.
\lipicsEnd
\end{notation}

\begin{definition}
\begin{enumerate}
\item A \emph{($\Sigma$-)identity} is an expression of the form
\[ s \eeq t~, \qquad \text{where $s,t\in \Sigmas\Vars$}. \]
A \emph{(finitary $\Sigma$-)quasiequation} is an expression of the form
\begin{equation}\label{eq:fin-quasieq} P \qmorespace{\implies} s \eeq  t~, \end{equation}
where $s=t$ is an identity and $P$ is a finite set of identities. The quasiequation \eqref{eq:fin-quasieq} is \emph{satisfied} in a $\Sigma$-algebra $A$ if for every variable interpretation $\interpret\colon \Vars\to A$, 
 \begin{equation}\label{eq:sound}  \interpretsem{u} \eeq \interpretsem{v} \text{ , for all $u=v$ in $P$,} \qimplies \interpretsem{s} \eeq \interpretsem{t}. \end{equation}
\item The \emph{substitution instance} of \eqref{eq:fin-quasieq} for a substitution $\sigma$ is the finitary quasiequation 
\[ P[\sigma] \qmorespace{\implies} s[\sigma] \eeq t[\sigma]~, \qquad \text{ where $P[\sigma] \eeq \{ u[\sigma]=v[\sigma] \mid (u=v)\in P \}$}.  \]
\end{enumerate}
 
\begin{notation}
We commonly write
\[ s_1=t_1 \,\wedge\, \,\ldots,\, \wedge \,s_n=t_n \qmorespace{\implies} s \eeq t \]
or simply
\[ s_1=t_1, \,\ldots,\, s_n=t_n \qmorespace\implies s \eeq t \]
for a quasiequation
\[ \{ s_1=t_1,\,\ldots,\, s_n=t_n \} \qmorespace\implies s \eeq t \]
This notation allows to freely permute or duplicate premises.
\lipicsEnd
\end{notation}

\end{definition}

\begin{example}
We denote by $\Ax_{=}$ the following set of quasiequations, where $\mathsf{f}$ ranges over operations in $\Sigma$ and $n$ is the arity of $\mathsf{f}$. 
\begin{align}
& \qmorespace\implies x=x \label{eq:refl}\\
x=y &\qmorespace\implies y=z \label{eq:sym}\\
x=y,\, y=z &\qmorespace\implies x=z \label{eq:trans}\\
x_1=y_1,\, \ldots\, x_n=y_n & \qmorespace\implies \mathsf{f}(x_1,\ldots,x_n)=\mathsf{f}(y_1,\ldots,y_n) \label{eq:cong1}
\end{align}
These correspond to axioms of standard equational logic and thus hold in every $\Sigma$-algebra.
\lipicsEnd
\end{example}

\begin{remark}\label{rem:subst}
If an algebra $A$ satisfies a quasiequation $\quasieq$, then $A$ satisfies all substitution instances of $\quasieq$ as well.
\lipicsEnd%
\end{remark}

\begin{definition}
Given a set $\Ax$ of finitary quasiequations, a \emph{proof} of a finitary quasiequation 
\[ P \implies s=t \]
from $\Ax$ is a finite tree with the following properties:
\begin{enumerate}
\item Every node is labelled with some identity $u=v$ where $u,v\in \Sigmas\Vars$.
\item The root is labelled with $s=t$.
\item For every inner node labelled $u=v$ with children labelled $u_1=v_1,\ldots,u_n=v_n$, the quasiequation
\[ u_1=v_1,\ldots,u_n=v_n\implies u=v \]
is a substitution instance of some quasiequation in $\Ax\cup \Ax_{=}$.
\item For every leaf node, the label $u=v$ is either an identity in $P$ or a substitution instance of some premise-free quasiequation in $\Ax\cup \Ax_{=}$.
\end{enumerate}
\item A quasiequation $P\implies s=t$ is \emph{$\Ax$-provable}, denoted 
\[ \Ax\vdash P\implies s=t, \] 
if there exists a proof of it from $\Ax$. It is \emph{provable in equational logic} if it is $\emptyset$-provable, denoted
\[ \vdash P\implies s=t. \] 
\end{definition}

\begin{remark}\label{rem:proof-props}
If $\Ax\vdash P\implies s=t$, then every $\Ax\cup \{ P \implies s=t\}$-provable quasiequation is $\Ax$-provable. In other words, one can treat quasiequations provable from $\Ax$ like additional axioms. Indeed, any use of a substitution instance $P[\sigma]\implies s[\sigma]=t[\sigma]$ in a proof from $\Ax\cup \{ P \implies s=t\}$ can be replaced with a proof of  $P[\sigma]\implies s[\sigma]=t[\sigma]$ from $\Ax$.
\lipicsEnd
\end{remark}

\subsection{Discrete Sequence Algebras}

\begin{definition}[Discrete Sequence Algebras]
A \emph{discrete sequence algebra} is given by the set $A^\omega$ of sequences over a set $A$ equipped with the unary operations
\[ \First\colon A^\omega\to A^\omega,\qquad \First(\varphi) = (\varphi(0),\varphi(0),\varphi(0),\ldots), \]
and
\[  
\shift\colon A^\omega \to A^\omega,\qquad \shift(\varphi)=(\varphi(1),\varphi(2),\varphi(3),\ldots). 
\]
Note that discrete sequence algebras are algebras for the algebraic signature $\Delta$ given by the unary operation symbols $\Firstsymbol$ and $\Nextsymbol$:
\[ \Delta = \{ \First:1,\, \shift:1\}. \]
\end{definition}

\begin{definition}
\begin{enumerate}
\item A \emph{$\Nextsymbol$-term} 
term is a term of the form $\shift^k(x)$, where $k\geq 0$ and $x\in \Vars$.
\item A \emph{$\Nextsymbol$-identity} is an identity $s=t$ between $\Nextsymbol$-terms.
\item A \emph{$\First\shift$-term} 
term is a term of the form $\First\shift^k(x)$, where $k\geq 0$ and $x\in \Vars$.
\item A \emph{$\First\shift$-identity} is an identity $s=t$ between $\First\shift$-terms.
\item A \emph{homogeneous identity} is either a $\Nextsymbol$-identity or a $\First\shift$-identity.
\end{enumerate}
\end{definition}

\begin{definition}
A finitary $\Delta$-quasiequation is \emph{valid for discrete sequence algebras} if it satisfied in every discrete sequence algebra.
\end{definition}

In the remainder of \Cref{app:incompleteness}, by a \emph{quasiequation} we always mean a finitary $\Delta$-quasiequation, and by a \emph{valid quasiequation} we mean one that is valid for discrete sequence algebras.
 
\begin{example}
The following quasiequations are valid:
\belowdisplayskip=-0\baselineskip%
\begin{align}
& \implies \First(x) = \First^2(x) \label{eq:theta-1}\\
& \implies \shift\First(x)=\First(x) \label{eq:theta-2}\\
x=\First(y)  & \implies \First(x)=\First(y) \label{eq:theta-3} \\
x=\First(y) & \implies \shift(x)=x \label{eq:theta-4} 
\end{align}%
\normalsize%
\lipicsEnd
\end{example}

\begin{remark}
 If $\Ax$ is a set of valid quasiequations, then every $\Ax$-provable quasiequation is valid. This is immediate from \Cref{rem:subst} by induction on the height of proof trees.
\lipicsEnd
\end{remark}
We start with a few technical results on quasiequations provable in equational logic.

\begin{lemma}\label{lem:hom-proof}
If $\vdash P\implies s=t$ and all identities in $P$ are homogeneous, then
\begin{enumerate}
\item $s$ is a $\Nextsymbol$-term iff $t$ is a $\Nextsymbol$-term.
\item $s$ is a $\First\shift$-term iff $t$ is a $\First\shift$-term.
\end{enumerate}
\end{lemma}

\begin{proof}
Immediate by induction on the height of a proof tree for $P\implies s=t$.
\end{proof}

Proofs of quasiequations in equational logic have a convenient graphical presentation:

\begin{definition}\label{def:graph}
To every set $P$ of homogeneous identities we associate the infinite undirected graph $G(P)$ whose nodes are all  pairs $(x,n)$ where $x\in \Vars$ and $n$ is a non-negative integer, and with two types of edges for $(x,n)\neq (y,k)$:
\begin{itemize}
\item If $\vdash P\implies {\shift}^n(x)=\shift^k(y)$, then there is a \emph{strong} edge between $(x,n)$ and $(y,k)$.
\item If $\First\shift^n(x)=\First\shift^k(y)$ lies in $P$ and the quasiequation $P\implies \shift^n(x)=\shift^k(y)$ is not provable in equational logic, then there is a \emph{weak} edge between $(x,n)$ and $(y,k)$.
\end{itemize}
By a \emph{strong path} in $G(P)$ we mean a path consisting of strong edges.
\end{definition}

\begin{proposition}\label{prop:proof-graph}
Let $P$ be a set of homogeneous identities, $n,k\geq 0$ and $x,y\in \Vars$.
\begin{enumerate}
\item $\vdash P\implies \shift^n(x)=\shift^k(y)$ iff there exists a strong path from $(x,n)$ to $(y,k)$ in $G(P)$. 
\item $\vdash P\implies \First\shift^n(x)=\First\shift^k(y)$ iff there exists a path from $(x,n)$ to $(y,k)$ in $G(P)$.
\end{enumerate}
\end{proposition}

\begin{proof}
\begin{enumerate}
\item ($\Longrightarrow$) Suppose that $\vdash P\implies \shift^n(x)=\shift^k(y)$. If $(x,n)=(y,k)$, there is a trivial path of length $0$ from $(x,n)$ to $(y,k)$. If $(x,n)\neq (y,k)$, there is a strong edge from $(x,n)$ to $(y,k)$, that is, an strong path of length $1$.

\mn ($\Longleftarrow$) Suppose that there exists a strong path from $(x,n)$ to $(y,k)$ in $G(P)$, say
\[ (x,n)=(x_0,n_0),\,(x_1,n_1)\,\ldots,\, (x_l,n_l)=(y,k).  \]
This means that \[\vdash P\implies \shift^{n_i}(x_i) = \shift^{n_{i+1}}(x_{i+1}) \qquad(0\leq i<l)\] and therefore \[\vdash P\implies \shift^n(x)=\shift^k(y)\] via reflexivity \eqref{eq:refl} and transitivity \eqref{eq:trans}.
\item 
($\Longleftarrow$) Suppose that there exists a path from $(x,n)$ to $(y,k)$ in $G(P)$, say
\[ (x,n)=(x_0,n_0),\,(x_1,n_1)\,\ldots,\, (x_l,n_l)=(y,k).  \]
Then \begin{equation}\label{eq:seqeq}\vdash P\implies \First\shift^{n_i}(x_i) = \First\shift^{n_{i+1}}(x_{i+1}) \qquad(0\leq i<l).\end{equation}
Indeed, for each $i$, if there is a strong edge from $(x_i,n_i)$ to $(x_{i+1},n_{i+1})$, then
\[ \vdash P\implies \shift^{n_i}(x_i) = \shift^{n_{i+1}}(x_{i+1}) \]
and therefore \[\vdash P\implies \First\shift^{n_i}(x_i) = \First\shift^{n_{i+1}}(x_{i+1})\]
by congruence \eqref{eq:cong1}. If there is a weak edge from $(x_i,n_i)$ to $(x_{i+1},n_{i+1})$, then the identity $\shift^{n_i}(x_i) = \shift^{n_{i+1}}(x_{i+1})$ lies in $P$ and therefore trivially
\[ \vdash P\implies \shift^{n_i}(x_i) = \shift^{n_{i+1}}(x_{i+1}). \]
From \eqref{eq:seqeq} we now get \[\vdash P\implies \First\shift^n(x)=\First\shift^k(y)\] via reflexivity \eqref{eq:refl} and transitivity \eqref{eq:trans}.

\mn ($\Longrightarrow$) Suppose that 
 \begin{equation}\label{eq:quasieq-pf} \vdash P\implies \First\shift^n(x)=\First\shift^k(y).\end{equation}
We prove that there exists a path from $(x,n)$ to $(y,k)$ in $G(P)$. The proof is by induction of the height of a proof tree for \eqref{eq:quasieq-pf}. For a proof of height $0$, the identity $\First\shift^n(x)=\First\shift^k(y)$ is either an element of $P$, in which case there is an edge from $(x,n)$ to $(y,k)$, or it is substitution instance of reflexivity \eqref{eq:refl}, i.e.~the terms $\First\shift^n(x)$ and $\First\shift^k(y)$ are identical and there is a path of length $0$.

Now consider a proof tree of positive height. If the last step is an substitution instance of the symmetry rule \eqref{eq:sym}, then there is a proof of
\[ P\implies\First\shift^k(y)= \First\shift^n(x)\]
of smaller height. By induction, there is a path from $(y,k)$ to $(x,n)$, hence from $(x,n)$ to $(y,k)$. If the last step is an instance of the transitivity rule \eqref{eq:trans}, then there are proofs of smaller height of 
\[ P\implies\First\shift^n(x)= u\qqand P\implies  u = \First\shift^k(y) \]
for some term $u$. By \Cref{lem:hom-proof}, the term $u$ must have the form $u=\First\shift^m(v)$
for some $m\geq 0$ and $v\in \Vars$. By induction, there is a path from $(x,n)$ to $(v,m)$ and from $(v,m)$ to $(y,k)$, hence from $(x,n)$ to $(y,k)$. If the last step is an instance of the congruence rule \eqref{eq:cong1}, then there exists a proof of smaller height of
\[ P\implies \shift^n(x) = \shift^k(y).  \]
Thus, there is a (strong) edge from $(x,n)$ to $(y,k)$, i.e.~a path of length $1$.\qedhere
\end{enumerate}
\end{proof}

\begin{lemma}\label{lem:eq-logic-1}
Let $P$ be a finite set of homogeneous identities and $x,y\in \Vars$. For all $k<l$ and $m,n\geq 0$,
\begin{align*} \vdash \shift^k(x)=\shift^l(x),\, \shift^m(x)=\shift^n(y) & \implies  \shift^{m'}(x)=\shift^n(y)  \\
 \vdash \shift^k(x)=\shift^l(x),\, \shift^{m'}(x)=\shift^n(y) &\implies  \shift^{m}(x)=\shift^n(y)  
\end{align*}
where \[m'=k+ ((m-k)\mathrel{\mathrm{mod}}(l-k)).\] 
\end{lemma}

\begin{proof}
By definition of $m'$, one has $m=p\cdot (l-k) + m'$ for some integer $p$. From this the statements are immediate.
\end{proof}

\begin{lemma}\label{lem:eq-logic-2}
Let $P$ be a finite set of homogeneous identities and $x\in \Vars$. If no quasiequation $P\implies\shift^k(x)=\shift^l(x)$ where $k\neq l$ is provable in equational logic, then the following holds:
\begin{enumerate}
\item There exist only finitely many pairs $k\neq l$ such that $\vdash P\implies \First\shift^k(x)=\First\shift^l(x)$.
\item If moreover $y\in \Vars$ and $\vdash P\implies \shift^m(y)=\shift^n(y)$ for some $m\neq n$, then for every fixed $q$ there exist only finitely many $k$ with $\vdash P\implies \First\shift^k(x)=\First\shift^q(y)$.
\end{enumerate}
\end{lemma}

\begin{proof}
 Suppose that no quasiequation $P\implies\shift^k(x)=\shift^l(x)$ where $k\neq l$ is not provable in equational logic. 
\begin{enumerate}
\item Consider the graph $G(P)$ associated to $P$. For every node $(y,n)$, there exists at most one node $(x,k)$ with a strong edge from $(y,n)$ to $(x,k)$. Indeed, if there exists a second such node $(x,l)$, $k\neq l$, we have $\vdash P\implies \shift^k(x)=\shift^l(x)$, a contradiction. Consider the set $C\seq \Nat\times \Nat$ given by
\begin{align*} (n,k)\in C & \iff 
\text{there exists $y\in \Var$ such that} \\
&\phantom{\iff l} \text{(i) there ex.\ a weak edge containing the node $(y,n)$, and }\\
&\phantom{\iff l} \text{(ii) there ex.\ a strong edge (equival.: a strong path) from $(y,n)$ to $(x,k)$.} 
\end{align*}
Since $P$ is finite and by the above argument, the set $C$ is finite. Pick an upper bound $K$ such that
\[ \forall (n,k)\in C. \, k\leq K. \]
Now suppose that 
\begin{equation}\label{eq:pkl} \vdash P\implies \First\shift^k(x)=\First\shift^l(x) \end{equation}
 where $k\neq l$. Then there exists a path from $(x,k)$ to $(x,l)$ in $G(P)$ At least one of its edges is weak; otherwise, by \Cref{prop:proof-graph}, we have $\vdash P\implies \shift^k(x)=\shift^l(x)$, a contradiction. Let $(y,n)$ be the end node of the last weak edge on the path. Then the subpath from $(y,n)$ to $(x,l)$ is strong, and so $(n,l)\in C$, which implies $l\leq K$. By considering the reversed path from $(x,l)$ to $(x,k)$ instead, we see that $k\leq K$. Hence there are only finitely many pairs $k\neq l$ with \eqref{eq:pkl}
\item Suppose that $\vdash P\implies \shift^m(y)=\shift^n(y)$ for some $m\neq n$. Then no quasiequation $P\implies \shift^i(x)=\shift^j(y)$ is provable in equational logic. (Otherwise, this proof could be used to construct a proof of some $P\implies \shift^k(x)=\shift^l(x)$ where $k\neq l$.) Now fix $q$ and let $\vdash P\implies \First\shift^k(x)=\First\shift^q(y)$. Then there exists a path from $(x,k)$ to $(y,q)$ in $D(P)$, and by the above argument this path is not strong.  Let $(z,p)$ be the start node of the first weak edge on the path.  Consider the set $C$ and the upper bound $K$ as above. Then $(p,k)\in C$ and so $k\leq K$, which implies that there only finitely many choices of $k$.
\end{enumerate}
\end{proof}

\begin{lemma}\label{lem:eq-logic-3}
Let $P$ be a finite set of homogeneous identities and $x,y\in \Vars$. If $x\neq y$ and no quasiequation $P\implies\shift^k(x)=\shift^l(x)$ or $P\implies\shift^k(y)=\shift^l(y)$ for $k\neq l$ is provable in equational logic, then the following holds:
\begin{enumerate}
\item\label{propa} For all $p,p',q,q'$, if $\vdash P\implies \shift^p(x)=\shift^q(y)$ and $\vdash P\implies \shift^{p'}(x)=\shift^{q'}(y)$ then $p-q\, =p'-q'$.
\item\label{propb} There exist only finitely many pairs $k,l$ such that $\vdash P\implies \First\shift^k(x)=\First\shift^l(y)$ and the quasiequation $P\implies \shift^k(x)=\shift^l(y)$ is not provable in equational logic. 
\end{enumerate}
\end{lemma}

\begin{proof}
Let $x\neq y$ and suppose that no quasiequation $P\implies\shift^k(x)=\shift^l(x)$ or $P\implies\shift^k(y)=\shift^l(y)$ where $k\neq l$ is provable in equational logic.
\begin{enumerate}
\item Suppose that $\vdash P\implies \shift^p(x)=\shift^q(y)$ and $\vdash P\implies \shift^{p'}(x)=\shift^{q'}(y)$. If $q=q'$, it follows by symmetry and transitivity that $\vdash P\implies\shift^p(x)=\shift^{p'}(x)$, and so $p=p'$ by hypothesis, i.e.~$p-q=p'-q'$. Now suppose that $q\neq q'$, say $q<q'$. Then $\vdash P\implies\shift^{p+(q'-q)}(x)=\shift^{q'}(x)$  by congruence and so $\vdash P\implies\shift^{p+(q'-q)}(x)=\shift^{p'}(x)$ by symmetry and transitivity. This implies $p+(q'-q)=p'$ by hypothesis, i.e.~$p-q=p'-q'$.
\item Suppose that 
\begin{equation}\label{eq:pkl2} \vdash P\implies \First\shift^k(x)=\First\shift^l(y)\end{equation} and the quasiequation $P\implies \shift^k(x)=\shift^l(y)$ is not provable in equational logic. Consider again the set $C$ and the upper bound $K$ as defined in the proof of \Cref{lem:eq-logic-2}. By hypothesis, there is a path from $(x,k)$ to $(y,l)$ in $G(P)$ but the path is not strong. Let $(z,n)$ be the start node of the first weak edge on the path. Then $(n,k)\in C$ and so $k\leq K$. A symmetric argument, swapping the roles of $x$ and $y$, shows that also $l$ is bounded from above by some constant $L$. This shows that there are only finitely many pairs $k,l$ satisfying \eqref{eq:pkl2}.\qedhere
\end{enumerate}
\end{proof}

\begin{proposition}\label{prop:basis}
Let $P$ be a finite set of homogeneous identities, and let $x$ and $y$ be distinct variables. Then there exists a finite set $B$ of homogeneous identities over the variables $x$ and $y$ such that 
\begin{equation}\label{eq:basis0}
(u=v)\in B \quad\text{implies}\quad \vdash P\implies u=v,
\end{equation}
and moreover, for any homogeneous identity $u=v$ where $u,v\in \Deltas\{x,y\}$,
\begin{equation}\label{eq:basis} \vdash P\implies u=v \qquad\text{iff}\qquad \vdash B\implies u=v. \end{equation}
We call such a set $B$ a \emph{basis} for $P$ w.r.t.~$x$ and $y$.
\end{proposition}

\begin{proof}
 Let $D(P)$ be the set of all homogeneous identities $u=v$ where $u,v\in \Deltas\{x,y\}$ and 
\[ \vdash P \implies u=v. \] For every finite subset $B\seq D(P)$, the implication \eqref{eq:basis0} and the right-to-left implication of \eqref{eq:basis} clearly hold. We show that there is a choice of $B$ for which also the left-to-right implication of \eqref{eq:basis} holds. For this purpose, we need to distinguish several cases:

\mn\textbf{\underline{Case 1:}} $D(P)$ contains two identities $\shift^k(x)=\shift^l(x)$ and $\shift^m(y)=\shift^n(y)$ where $k\neq l$ and $m\neq n$.

\smallskip\noindent Choose two such pairs $k,l$ and $m,n$, where w.l.o.g.~we assume $k<l$ and $m<n$. Let $B\seq D(P)$ be given by
\begin{romanenumerate}
\item all identities $\shift^p(x)=\shift^q(x)$ and $\First\shift^p(x)=\First\shift^q(x)$ in $D(P)$ where $k\leq p,q \leq l$;
\item all identities $\shift^p(y)=\shift^q(y)$ and $\First\shift^p(y)=\First\shift^q(y)$ in $D(P)$ where $m\leq p,q \leq n$;
\item all identities $\shift^p(x)=\shift^q(y)$ and $\First\shift^p(x)=\First\shift^q(y)$ in $D(P)$ where $k\leq p \leq l$ and $m\leq q \leq n$;
\end{romanenumerate}
Clearly, $B$ is a finite set. We claim that the left-to-right implication of \eqref{eq:basis} holds. Thus suppose that $\vdash P\implies u=v$. We only consider the case that $u=v$ is of the form $\shift^p(x)=\shift^q(y)$, since the other cases are entirely analogous. By \Cref{lem:eq-logic-1} there exists an identity $\shift^{p'}(x)=\shift^{q'}(y)$ in $B$ (namely, take $p'=k+((p-k) \mathrel{\mathrm{mod}} (l-k))$ and $q'=m+((q-m)\mathrel{\mathrm{mod}} (n-m))$) such that 
\[ \vdash \shift^k(x)=\shift^l(x),\, \shift^m(y)=\shift^n(y),\, \shift^{p'}(x)=\shift^{q'}(y) \implies \shift^{p}(x)=\shift^{q}(y).  \]
Since all three identities in the premise lie in $B$, we conclude $\vdash B\implies \shift^{p}(x)=\shift^{q}(y)$.

\mn\textbf{\underline{Case 2:}} $D(P)$ contains no identity $\shift^k(x)=\shift^l(x)$ where $k\neq l$ but contains some identity $\shift^m(y)=\shift^n(y)$ where $m\neq n$.

\mn Choose such a pair $m,n$, and assume w.l.o.g.~that $m<n$. Let $B\seq D(P)$ be given by
\begin{romanenumerate}
\item all identities $\First\shift^k(x)=\First\shift^l(x)$, $k\neq l$, in $D(P)$.
\item\label{case2-b} all identities $\First\shift^k(x)=\First\shift^q(y)$ in $D(P)$ where $m\leq q\leq n$;
\item\label{case2-c} all identities $\First\shift^p(y)=\First\shift^q(y)$ and $\First\shift^p(x)=\First\shift^q(x)$ in $D(P)$ where $m\leq p,q\leq n$;
\end{romanenumerate}
Then $B$ is a finite set by \Cref{lem:eq-logic-2}. Moreover, $\vdash P\implies u=v$ implies $\vdash B\implies u=v$:
\begin{itemize}
\item For $u=v$ of the form $\shift^k(x)=\shift^l(x)$, one has $k=l$ by hypothesis, so $\vdash B\implies u=v$ by reflexivity.
\item For $u=v$ of the form $\First\shift^k(x)=\First\shift^l(x)$, if $k=l$ then $\vdash B\implies u=v$ by reflexivity, and if $k\neq l$, then the identity $u=v$ is contained in $B$, so $\vdash B\implies u=v$ trivially.
\item $u=v$ cannot be of the form $\shift^k(x)=\shift^q(y)$, for otherwise from this and $\vdash P\implies \shift^m(y)=\shift^n(y)$ a non-trivial $P\implies \shift^p(x)=\shift^q(x)$ is provable.
\item For $u=v$ of the form $\First\shift^k(x)=\First\shift^q(y)$, one can assume w.l.o.g.~that $m\leq q\leq n$ by \Cref{lem:eq-logic-2}, and then $u=v$ lies in $B$, so $\vdash B\implies u=v$.
\item For $u=v$ of the form $\shift^p(y)=\shift^q(y)$ or $\First\shift^p(y)=\First\shift^q(y)$, one can assume w.l.o.g.~that $m\leq p,q\leq n$ by \Cref{lem:eq-logic-1}, and then $u=v$ lies in $B$, so $\vdash B\implies u=v$.
\end{itemize}

\mn\textbf{\underline{Case 3:}} $D(P)$ contains some identity $\shift^k(x)=\shift^l(x)$ where $k\neq l$ but no identity $\shift^m(y)=\shift^n(y)$ where $m\neq n$.

\smallskip\noindent Symmetric to Case 2.

\mn\textbf{\underline{Case 4.}} $D(P)$ contains no identity $\shift^k(x)=\shift^l(x)$ where $k\neq l$ and no identity $\shift^m(y)=\shift^n(y)$ where $m\neq n$. 

\smallskip\noindent We call an identity $\shift^k(x)=\shift^m(y)$ in $D(P)$ \emph{minimal} if there exists no identity $\shift^{k'}(x)=\shift^{m'}(y)$ with $k'<k$. By \Cref{lem:eq-logic-2}, there exists at most one minimal identity in $D(P)$. 

We let $B\seq D(P)$ be given by
\begin{romanenumerate}
\item all identities $\First\shift^k(x)=\First\shift^l(x)$, $k\neq l$, in $D(P)$;
\item the minimal identity in $D(P)$, if any;
\item all identities $\First\shift^k(x)=\First\shift^m(y)$ in $D(P)$ such that $\shift^k(x)=\shift^m(y)$ is not contained in $D(P)$.
\item all identities $\First\shift^m(y)=\First\shift^m(y)$, $m\neq n$, in $D(P)$.
\end{romanenumerate}
Then $B$ is a finite set by \Cref{lem:eq-logic-2} and \Cref{lem:eq-logic-3}. Moreover, $\vdash P\implies u=v$ implies $\vdash B\implies u=v$:
\begin{itemize}
\item For $u=v$ of the form $\shift^k(x)=\shift^l(x)$ or $\shift^k(y)=\shift^l(y)$, one has $k=l$ and so $\vdash B\implies u=v$ by reflexivity.
\item For $u=v$ of the form $\First\shift^k(x)=\First\shift^l(x)$ or $\First\shift^k(y)=\First\shift^l(y)$, the identity $u=v$ lies in $B$, so $\vdash B\implies u=v$ trivially.
\item For $u=v$ of the form $\shift^k(x)=\shift^m(y)$, the identity is provable via congruence from the minimal identity in $D(P)$ by \Cref{lem:eq-logic-3}, and so $\vdash B\implies u=v$ since the latter lies in $B$.
\item For $u=v$ of the form $\First\shift^k(x)=\First\shift^m(y)$, if $\shift^k(x)=\shift^m(y)$ in $D(P)$, then $u=v$ is provable from the minimal identity (which lies in $B$) via congruence; otherwise, the identity $u=v$ lies in $B$. In both cases, $\vdash B\implies u=v$. 
\end{itemize}
This concludes the proof of \eqref{eq:basis}. 
\end{proof}

\begin{proposition}\label{prop:basis-valid}
Let $P$ a finite set of homogeneous identities, and let $B$ be a finite basis for $P$ w.r.t.~$x$ and $y$. 
Then for any homogeneous identity $u=v$ where $u,v\in \Deltas\{x,y\}$,
\begin{equation}\label{eq:basis2} P\implies u=v \text{ is valid}\qquad\text{iff}\qquad B\implies u=v \text{ is valid}. \end{equation}
\end{proposition} 

\begin{proof}
We use the notation $D(P)$ from the proof of \Cref{prop:basis}.
Clearly, since $B\seq D(P)$, the right-to-left implication holds. Conversely, suppose that $P\implies u=v$ is valid, let $A$ be a set, and let $\interpret\colon \Vars\to A^\omega$ be a variable interpretation such that
\begin{equation}\label{eq:iasm} \interpretsem{s}=\interpretsem{t} \quad \text{for all $s=t$ in B}. \end{equation}
It suffices show that there exists an interpretation ${\ol{\interpret}}\colon \Vars\to A^\omega$ such that 
\begin{equation}\label{eq:j} {\ol{\interpret}}(x)=\interpret(x),\quad {\ol{\interpret}}(y)=\interpret(y),\quad \interpretsemol{s}=\interpretsemol{t} \quad \text{for all $s=t$ in P}. \end{equation}
Then, since  $P\implies u=v$ is valid and $u,v\in \Deltas\{x,y\}$, it follows that 
\[\interpretsem{u} = \interpretsemol{u} = \interpretsemol{v} = \interpretsem{v}\]
proving that $B\implies u=v$ is valid. For the definition of ${\ol{\interpret}}$, we consider the graph $G(P)$ (\Cref{def:graph}). For each connected component $C\seq \Vars\times \Nat$, we define an element $a_C\in A$ as follows:
\begin{itemize}
\item If $C$ contains a node $(z,n)\in C$ where $z\in \{x,y\}$, put 
\[ a_C:=\interpret(z)(n).\]
 Note that $a_C$ is independent of the choice of $(z,n)$: if $(w,m)$ is another element of $C$ where $w\in \{x,y\}$, then there exists a path from $(z,n)$ to $(w,m)$ in $G(P)$, so \[\vdash P\implies \First\shift^n(z)=\First\shift^m(w).\] 
This implies
\[\vdash B\implies \First\shift^n(z)=\First\shift^m(w)\] 
by \eqref{eq:basis}, in particular the quasiequation $B\implies \First\shift^n(z)=\First\shift^m(w)$ is valid. Therefore, from \eqref{eq:iasm} it follows that $\interpretsem{\First\shift^n(z)}=\interpretsem{\First\shift^k(w)}$, which means \[\interpret(z)(n)=\interpret(w)(k).\]
\item If $C$ does not contain any node $(z,n)$ where $z\in \{x,y\}$, choose an arbitrary element $a_C\in A$.
\end{itemize}
Now define ${\ol{\interpret}}\colon \Vars\to A^\omega$ by
\[ {\ol{\interpret}}(z)(n) = a_{C_{(z,n)}} \qquad \text{where $C_{(z,n)}$ is the connected component of $(z,n)$ in $G(P)$.}  \]
We show that ${\ol{\interpret}}$ satisfies the desired properties \eqref{eq:j}. Indeed, if $z\in \{x,y\}$, then
\[ {\ol{\interpret}}(z)(n) = a_{C_{(z,n)}} = \interpret(z)(n) \qquad \text{for all $n\geq 0$}\]
by definition of $a_{C_{(z,n)}}$, and so $\interpret(z)={\ol{\interpret}}(z)$. Now suppose that $s=t$ is an identity in $P$. If $s=t$ is a $\Nextsymbol$-identity $\shift^k(z)=\shift^l(w)$ then for every $n\in \Nat$ the pairs $(z,k+n)$ and $(w,l+n)$ lie in the same connected component of $G(P)$ since they are connected by a strong edge, and so
\begin{align*}
	\interpretsemol{s}(n) &=\interpretsemol{\shift^k(z)}(n) = {\ol{\interpret}}(z)(k+n) \\
	& = a_{C_{(z,k+n)}} = a_{C_{(w,l+n)}} \\
&= {\ol{\interpret}}(w)(l+n) = \interpretsemol{\shift^l(w)}(n) = \interpretsemol{t}(n)~, 
\end{align*}%
whence $\interpretsemol{s}=\interpretsemol{t}$. Similarly, if $s=t$ is a $\First\shift$-identity $\First\shift^k(z)=\First\shift^l(w)$, then the pairs $(z,k)$ and $(w,l)$ lie in the same connected component of $G(P)$ since they are connected by a (strong or weak) edge, whence
\[ \interpretsemol{s}=\interpretsemol{\First\shift^k(z)}  = ({\ol{\interpret}}(z)(k))^\omega = (a_{C_{(z,k)}})^\omega = (a_{C_{(w,l)}})^\omega = ({\ol{\interpret}}(w)(l))^\omega = \interpretsemol{\shift^l(w)} = \interpretsemol{t}. \]
This concludes the proof.
\end{proof}

The existence of finite bases allows us to reduce quasiequations to a very simple form:
\begin{definition}\label{def:normalized} Fix a pair of distinct variables $x\neq y$. A quasiequation $P\implies s=t$ is \emph{normalized (w.r.t.~$x$ and $y$)} if the following conditions hold:
\begin{enumerate}
\item\label{n1} All identities in $P$ are homogeneous.
\item\label{n2} The term $s$ and $t$ and all terms occurring in $P$ are elements of $\Deltas\{x,y\}$.
\item\label{n3} $P$ contains no identity of the form $\shift^l(y)=\shift^k(x)$.
\item\label{n4} All identities in $P$ of the form $\shift^k(x)=\shift^l(y)$ have the same degree $k$ on the left hand side.
\item\label{n5} If the identities $\shift^k(x)=\shift^l(y)$ and $\shift^k(x)=\shift^m(y)$ lie in $P$, then so does $\shift^l(y)=\shift^m(y)$.
\end{enumerate}
\end{definition}

\begin{definition}
Given a set $\Ax$ of quasiequations, a set $\Ax'$ of quasiequations \emph{strengthens} $\Ax$ if every quasiequation in $\Ax$ is $\Ax'$-provable (equivalently, by \Cref{rem:proof-props}, if every $\Ax$-provable quasiequation is $\Ax'$-provable).
\end{definition}

\begin{theorem}[Normalization]\label{prop:normalization}
For every finite set $\Ax$ of valid quasiequations, there exists a finite set $\Ax'$ of valid quasiequations such that $\Ax'$ strengthens $\Ax$ and all quasiequations contained in $\Ax'$ except \eqref{eq:theta-1}--\eqref{eq:theta-4} are normalized.  
\end{theorem}

\begin{proof}
We may assume that $\Ax$ contains the quasiequations \eqref{eq:theta-1}--\eqref{eq:theta-4}, which are all valid; otherwise, add them to $\Ax$. We show how to replace, step by step, non-normalized quasiequations in $\Ax$ (i.e.\ quasiequations \emph{not} satisfying one of the requirements \ref{n1}--\ref{n5}) with equivalent normalized ones, until a set of valid quasiequations is reached that strenghtens $\Ax$ and contains, except for \eqref{eq:theta-1}--\eqref{eq:theta-4}, only normalized quasiequations.

\begin{enumerate}[N1.]
\item Suppose that $\Ax$ contains, except from \eqref{eq:theta-1}--\eqref{eq:theta-4}, some quasiequation of the form
\begin{equation}\label{eq:defect-n1} P,\,\shift^k(z) = u \implies s=t  \end{equation}
where $u$ contains $\Firstsymbol$. (The case $u=\shift^k(z)$ is handled symmetrically.) Thus $u=v(\First\shift^m(w))$ for some unique term $v$ and variable $w$. Replace \eqref{eq:defect-n1} with  
\begin{equation}\label{eq:fix-n1} P,\, \shift^{k+1}(z)=\shift^k(z),\, \First \shift^k(z)=\First\shift^m(w) \implies s=t.  \end{equation}
This quasiequation is valid since the identity $\shift^k(z) = v(\First\shift^m(w))$ is semantically equivalent to the conjunction of $\shift^{k+1}(z)=\shift^k(z)$ and $\First \shift^k(z)=\First\shift^m(w)$. Moreover, the quasiequation \eqref{eq:defect-n1} is provable from \eqref{eq:fix-n1} using equational logic and \eqref{eq:theta-1}--\eqref{eq:theta-4}. Therefore, the resulting set of axioms strengthens $\Ax$ (\Cref{rem:proof-props}). Repeating this process yields a strengthening of $\Ax$ where all quasiequations except \eqref{eq:theta-1}--\eqref{eq:theta-4} satisfy \ref{n1}.
\item Suppose that all quasiequations in $\Ax$ except \eqref{eq:theta-1}--\eqref{eq:theta-4} satisfy \ref{n1}. First, by renaming variables if necessary, we may assume that for every $P\implies s=t$ in $\Ax$, one has $s,t\in \Deltas\{x,y\}$.  Replace every quasiequation 
\begin{equation}\label{eq:defect-n2} P\implies s=t \end{equation}
 in $\Ax$ not satisfying \ref{n2}, that is, with variables distinct from $x,y$ in the premise, by the quasiequation 
\begin{equation}\label{eq:fix-n2} B\implies s=t \end{equation}
 where $B$ is a finite basis of $P$ w.r.t.~$x$ and $y$ (\Cref{prop:basis}). Then \eqref{eq:fix-n2} is valid by \Cref{prop:basis-valid}. Moreover, since $\vdash P\implies u=v$ for all $u=v$ in $B$ by \eqref{eq:basis0}, the quasiequation \eqref{eq:defect-n2} is provable from \eqref{eq:fix-n2}. Thus we obtain a strengthening of $\Ax$ where all quasiequations except \eqref{eq:theta-1}--\eqref{eq:theta-4} satisfy \ref{n1} and \ref{n2}.
\item Suppose that all quasiequations in $\Ax$ except \eqref{eq:theta-1}--\eqref{eq:theta-4} satisfy \ref{n1} and \ref{n2}, and that $\Ax$ contains a quasiequation 
\begin{equation}\label{eq:defect-n3} P, \,\shift^l(y)=\shift^k(x) \implies s=t.\end{equation} 
Then replace this quasiequation with
\begin{equation}\label{eq:fix-n3} P, \,\shift^k(x)=\shift^l(y) \implies s=t.\end{equation}
This quasiequation is valid because $\shift^l(y)=\shift^k(x)$ and $\shift^k(x)=\shift^l(y)$ are semantically equivalent. Moreover, \eqref{eq:defect-n3} is provable from \eqref{eq:fix-n3} via symmetry \eqref{eq:sym}, so the resulting set of axioms strengthens $\Ax$. By repeating this process, we end up with a strengthening of $\Ax$ where all quasiequations except \eqref{eq:theta-1}--\eqref{eq:theta-4} satisfy \ref{n1}--\ref{n3}.
\item Suppose that all axioms in $\Ax$ except \eqref{eq:theta-1}--\eqref{eq:theta-4} satisfy \ref{n1}--\ref{n3}, and that $\Ax$ contains a quasiequation
\begin{equation}\label{eq:defect-n4} P, \shift^k(x) = u, \shift^l(x)=v \implies s=t  \end{equation}
where $u$ and $v$ are $\Nextsymbol$-terms over $\{y\}$ and $k\neq l$, say $k<l$. Replace \eqref{eq:defect-n4} with
\begin{equation}\label{eq:fix-n4} P, \shift^l(x) = \shift^{l-k}u,\, \First\shift^{k+i}(x)=\First\shift^i(u)\, (i=0,\ldots,l-k),\,      \shift^l(x)=v \implies s=t.  \end{equation}
This quasiequation is valid because $\shift^k(x) = u$ is semantically equivalent to the conjunction of the identities $\shift^l(x) = \shift^{l-k}u$ and $\First\shift^{k+i}(x)=\First\shift^i(u)$ for $i=0,\ldots,l-k$. Moreover \eqref{eq:defect-n4} is provable from \eqref{eq:fix-n4} using the congruence rule \eqref{eq:cong1} of equational logic. Thus the resulting set of axioms strengthens $\Ax$. By repeating this process, we obtain a strengthening of $\Ax$ where all quasiequations except \eqref{eq:theta-1}--\eqref{eq:theta-4} satisfy \ref{n1}--\ref{n4}.
\item Finally, suppose that all axioms in $\Ax$ except \eqref{eq:theta-1}--\eqref{eq:theta-4} satisfy \ref{n1}--\ref{n4}, and that $\Ax$ contains a quasiequation
\begin{equation}\label{eq:defect-n5} P, \shift^k(x)=\shift^l(y),\, \shift^k(x)=\shift^m(y) \implies s=t  \end{equation}
where $l\neq m$ and $\shift^l(y)=\shift^m(y)$ does not appear in $P$. Replace \eqref{eq:defect-n5} with
\begin{equation}\label{eq:fix-n5} P, \shift^k(x)=\shift^l(y),\, \shift^k(x)=\shift^m(y),\, \shift^l(y)=\shift^m(y) \implies s=t.  \end{equation}
This quasiequation is valid because the conjunction of $\shift^k(x)=\shift^l(y)$ and $\shift^k(x)=\shift^m(y)$ is semantically equivalent to the conjunction of $\shift^k(x)=\shift^l(y)$, $\shift^k(x)=\shift^m(y)$, and $\shift^l(y)=\shift^m(y)$. Moreover, \eqref{eq:defect-n5} is provable from \eqref{eq:fix-n5} using the symmetry \eqref{eq:sym} and transitivity rules \eqref{eq:trans} of equational logic. Thus the resulting set of axioms strengthens $\Ax$. Repeating this process yields a strengthening of $\Ax$ where all quasiequations except \eqref{eq:theta-1}--\eqref{eq:theta-4} satisfy \ref{n1}--\ref{n5}, and thus are normalized. \qedhere
\end{enumerate}
\end{proof}

\subsection{Finite Axiomatizability}

We are now ready to prove the non-existence of finite axiomatizations for the quasiequational theory of discrete sequence algebras. A discrete axiom system \(\Ax\) (i.e.~a set of finitary $\Delta$-quasiequations) is \emph{sound} for discrete sequence algebras if every quasiequation from \(\Ax\) is valid for discrete sequence algebras, or equivalently, if every quasiequation provable from $\Ax$ is valid for discrete sequence algebras. The axiom system \(\Ax\) is \emph{complete} for discrete sequence algebras if every quasiequation that is valid for discrete sequence algebras has a finite derivation from \(\Ax\).

\incompleteness*

\begin{proof}
Suppose towards a contradiction that there exists a finite discrete axiom system $\Ax$ that is sound and complete for discrete sequence algebras. By \Cref{prop:normalization}, we may assume w.l.o.g.~that all quasiequations in $\Ax$ except \eqref{eq:theta-1}--\eqref{eq:theta-4} are normalized w.r.t.~the variables $x$ and $y$. The key to the proof is to exploit that, since $\Ax$ is finite and all quasiequations in $\Ax$ are finitary, for $\Nextsymbol$-identities $\shift^k(z)=\shift^l(z)$, $z\in \{x,y\}$, arising in the premise of any quasiequation in $\Ax$, the difference $|k-l|$ of the degrees of the $\Nextsymbol$-terms is necessarily bounded from above. In other words, we can choose a natural number $N$ such that
\begin{equation}\label{eq:defN} (P\implies s=t)\in \Ax \qand (\shift^k(z)=\shift^l(z))\in P \qquad\text{implies}\qquad |k-l|<N. \end{equation}
Fix a variable $a\in \Vars$ and consider the following quasiequations:
\begin{align}
\shift^N(a)=a,\, \First\shift^k(a)=\First\shift^{k+1}(a)\, (0\leq k<N) & \implies \shift^i(a)=\shift^j(a) \qquad (N\nmid i-j) \label{eq:critical-eq1}\\
\shift^N(a)=a,\, \First\shift^k(a)=\First\shift^{k+1}(a)\, (0\leq k<N) &\implies s=t \label{eq:critical-eq2}   
\end{align}
where $i,j\geq 0$, and $s,t$ ranges over all pairs of terms in $\Deltas\{a\}$ such that one of the two terms is a $\Nextsymbol$-term and the other contains $\Firstsymbol$. Note that all the above quasiequations have the same premise $\ol{P}$. Moreover, they are all valid for discrete sequence algebras since the premise $\ol{P}$ semantically implies that $a$ represents a flat sequence.
Therefore, by completeness of $\Ax$, each of the quasiequations \eqref{eq:critical-eq1} and \eqref{eq:critical-eq2} is $\Ax$-provable. Among them, choose a quasiequation
\begin{equation}\label{eq:minimal-qe} \ol{P} \implies s_{\min} = t_{\min} \end{equation}
with a proof from $\Ax$ of minimum height $h_{\min}$. Fix a proof tree of \eqref{eq:minimal-qe} of height $h_{\min}$, and let
\begin{equation}\label{eq:proof-ax} P \implies s = t \end{equation}
be the axiom from $\Ax\cup \Ax_{=}$ used in the last proof step. We now consider several cases, corresponding to the type of axiom used, and demonstrate that in each case a contradiction arises.

\mn\textbf{\underline{Case 1:}} \eqref{eq:proof-ax} is the reflexivity axiom \eqref{eq:refl} from $\Ax_{=}$.

\mn Then $s_{\min}$ and $t_{\min}$ are identical terms, but this not the case for the conclusion of any of the quasiequations \eqref{eq:critical-eq1} and \eqref{eq:critical-eq2}, a contradiction.

\mn\textbf{\underline{Case 2:}} \eqref{eq:proof-ax} is the symmetry axiom \eqref{eq:sym} from $\Ax_{=}$.

\mn Then there exists a proof of $\ol{P}\implies t_{\min}=s_{\min}$ from $\Ax$ of height $h_{\min}-1$. This is a quasiequation of the form \eqref{eq:critical-eq1} or \eqref{eq:critical-eq2}, since \eqref{eq:minimal-qe} is, contradicting the minimal choice of \eqref{eq:minimal-qe}. 

\mn\textbf{\underline{Case 3:}} \eqref{eq:proof-ax} is the transitivity axiom \eqref{eq:sym} from $\Ax_{=}$. Then there exists a term $u$ such that \begin{equation}\label{eq:proof-trans}\ol{P}\implies s_{\min}=u \qqand \ol{P}\implies u=t_{\min}\end{equation}
have a proof from $\Ax$ of height less than $h_{\min}$. By definition of \eqref{eq:critical-eq1} and \eqref{eq:critical-eq2}, one of the terms $s_{\min}$ or $t_{\min}$ is a $\Nextsymbol$-term, say $s_{\min}=\shift^k(a)$ (the other case is symmetric). Note that $u$ must be a term in the variable $a$: otherwise, since the premise $\ol{P}$ contains only terms in $a$, the $\Ax$-provable quasiequation $\ol{P}\implies s_{\min}=u$ is clearly not valid, contradicting soundness of $\Ax$. We consider two subcases:

\mn\textbf{\underline{Subcase 3.1:}} $u$ contains $\Firstsymbol$.

\mn Then $\ol{P}\implies s_{\min}=u$ is one of the quasiequations \eqref{eq:critical-eq2}, and it has a proof of height less than $h_{\min}$, which is impossible.

\mn\textbf{\underline{Subcase 3.2:}} $u=\shift^l(a)$ for some $l\geq 0$.

\mn Then the $\Ax$-provable quasiequations \eqref{eq:proof-trans} take the form
\[ \ol{P}\implies \shift^k(a)=\shift^l(a) \qqand \ol{P}\implies \shift^l(a)=t_{\min}. \] 
Then $t_{\min}$ cannot contain $\Firstsymbol$, since otherwise the second quasiequation above is of type \eqref{eq:critical-eq2} and has a proof tree of height less than $h_{\min}$, which is impossible. Thus $t_{\min}=\shift^m(a)$ for some $m\geq 0$. Since the above two quasiequations have proof trees of height less than $h_{\min}$, they cannot be of the form \eqref{eq:critical-eq1}, so $N$ divides both $k-l$ and $l-m$. Thus $N$ divides $k-m$. But this is impossible: since \eqref{eq:minimal-qe} is the quasiequation
\[ \ol{P} \implies \shift^k(a)=\shift^m(a) \]
and thus of the form \eqref{eq:critical-eq1}, we have that $N$ does not divide $k-m$.

\mn\textbf{\underline{Case 4:}} \eqref{eq:proof-ax} is the congruence axiom \eqref{eq:cong1} from $\Ax_{=}$.

\mn Then there exists a quasiequation 
\begin{equation}\label{eq:puv1} \ol{P}\implies u=v\end{equation} with a proof of height $h_{\min}-1$ such that $\mathsf{f}(u)=s_{\min}$ and $\mathsf{f}(v)=t_{\min}$ for some $\mathsf{f}\in \{\First, \shift\}$. Since $s_{\min}$ or $t_{\min}$ is a $\Nextsymbol$-term (say $s_{\min}$), we have $\mathsf{f}=\shift$, and moreover $u$ is a $\Nextsymbol$-term. If the term $v$ contains $\Firstsymbol$, then \eqref{eq:puv1} is a quasiequation of type \eqref{eq:critical-eq2}, which is impossible because it has a proof of height less than $h_{\min}$. If $v$ is a $\Nextsymbol$-term, then \eqref{eq:puv1} takes the form
\[\ol{P}\implies \shift^i(a)=\shift^j(a)\]
for some $i,j\geq 0$, and $s_{\min}=\shift^{i+1}(a)$ and $t_{\min}=\shift^{j+1}(a)$.  Since $i-j=(i+1)-(j+1)$, the quasiequation \eqref{eq:puv1} is of type \eqref{eq:critical-eq1}. This is impossible since it has a proof of height less than $h_{\min}$.

\mn\textbf{\underline{Case 5:}} \eqref{eq:proof-ax} is the one of the quasiequations \eqref{eq:theta-1}--\eqref{eq:theta-4}.

\mn We first note that no substitution instance of the conclusion of \eqref{eq:theta-1}--\eqref{eq:theta-3} is a $\Nextsymbol$-identity or an equation where one term is a $\Nextsymbol$-term and the other contains $\Firstsymbol$. Therefore, neither of these quasiequations can be used in the last step of a proof of \eqref{eq:minimal-qe}. It thus only remains to consider the case where \eqref{eq:proof-ax} is \eqref{eq:theta-4}, i.e.\
\[ x=\First(y) \implies \shift(x)=x. \]
Then there exists a substitution $\sigma$ such that
\[
s_{\min} = \shift(\sigma(x)) \qqand t_{\min}=\sigma(x).
\]
Then $\sigma(x)$ is necessarily a $\Nextsymbol$-term in $\Deltas\{a\}$ (otherwise \eqref{eq:minimal-qe} cannot be of the form \eqref{eq:critical-eq1} or \eqref{eq:critical-eq2}). Therefore, we have a proof of 
\begin{equation}\label{eq:shorter} \ol{P} \implies \sigma(x)=\First(\sigma(y)) \end{equation}
from $\Ax$ of height $h_{\min}-1$, where $\sigma(x)$ is a $\Nextsymbol$-term in $\Deltas\{a\}$ and $\First(\sigma(y))$ contains $\Firstsymbol$. Note that $\sigma(y)\in \Deltas\{a\}$, for otherwise the quasiequation \eqref{eq:shorter} is not valid, contradicting soundness of $\Ax$. Therefore, \eqref{eq:shorter} is a quasiequation of type \eqref{eq:critical-eq2}, contradicting the minimal choice of \eqref{eq:minimal-qe}.

\mn\textbf{\underline{Case 6:}} \eqref{eq:proof-ax} is a quasiequation from $\Ax$ and none of \eqref{eq:theta-1}--\eqref{eq:theta-4}.

\mn Then \eqref{eq:proof-ax} is normalized, i.e.\ satisfies \ref{n1}--\ref{n5} of \Cref{def:normalized}. Moreover, there exists a substitution $\sigma$ such that
\begin{equation}\label{eq:smin}
s_{\min} = s[\sigma] \qqand t_{\min}=t[\sigma],
\end{equation}
and 
\begin{equation}\label{eq:usvs}\Ax\vdash \ol{P}\implies u[\sigma]=v[\sigma] \qquad\text{for all $(u=v)\in P$} \end{equation}
with respective proofs from $\Ax$ of height less than $h_{\min}$. We start with a simple observation; in the following, we write $\shift^{*}$ for $\shift^p$ if the specific exponent $p$ does not matter . 

\mn($\star$) \emph{Claim.} If $z\in \{x,y\}$ and $\sigma(z)=\shift^{*}(a)$, then $P$ contains no identity $\shift^k(z)=\shift^l(z)$ where $k\neq l$. 

\mn\emph{Proof.} Suppose that $P$ contains some identity $\shift^k(z)=\shift^l(z)$ where $k,l\geq 0$, and let $\sigma(z)=\shift^j(a)$. Then, by \eqref{eq:usvs} we have that $\Ax\vdash \ol{P}\implies \shift^{k+j}(a) = \shift^{l+j}(a)$ with a proof of height less than $h_{\min}$. This means that this quasiequation is not of type \eqref{eq:critical-eq1}, so $N$ divides $(k+j)-(l+j)=k-l$. On the other hand, we have $|k-l|<N$ by the choice \eqref{eq:defN} of the upper bound $N$, and so necessarily $k-l=0$, i.e.~$k=l$. This proves ($\star$).

\mn Since \eqref{eq:minimal-qe} is a quasiequation of type \eqref{eq:critical-eq1} or \eqref{eq:critical-eq2}, at least one of the terms $s_{\min}$ or $t_{\min}$ is a $\Nextsymbol$-term; by symmetry, we assume that $s_{\min}$ is a $\Nextsymbol$-term. Since $s_{\min}=s[\sigma]$, also $s$ is a $\Nextsymbol$-term: 
\begin{equation}\label{eq:defiz} s=\shift^i(z) \qquad \text{for some $i\geq 0$ and $z\in \{x,y\}$}.\end{equation} Moreover, we have \[\sigma(z)=\shift^j(a) \qquad\text{for some $j\geq 0$}\,\]
so that 
\[s_{\min}=\shift^{i+j}(a).\]

\mn Choose an upper bound to the degrees of $\First\shift$-identities in $P$, that is, a natural number $K$ such that, for all $v,w\in \{x,y\}$, 
\begin{equation}\label{eq:defK} (\First\shift^k(v)=\First\shift^l(w))\in  P \qquad\text{implies}\qquad k,l\leq K.\end{equation}
 We now consider several subcases, corresponding to whether the variable $z\in \{x,y\}$ of \eqref{eq:defiz} is $x$ or $y$, and to which kinds of identities occur in $P$. 

\mn\textbf{\underline{Subcase 6.1:}} $z=x$, and $P$ contains some identity $\shift^{k}(x)=\shift^{m}(y)$ where $k,m\geq 0$.

\mn Let us make the following observations:

\begin{enumerate}[O1.]
\item\label{o1} $P$ contains no non-trivial identity $\shift^{*}(x)=\shift^{*}(x)$. This follows from ($\star$) since $\sigma(x)=\shift^{*}(a)$.
\item\label{o2} We have $\sigma(y)=\shift^{*}(a)$. Indeed, $\sigma(y)$ cannot contain a variable $b\neq a$, since otherwise the quasiequation $\ol{P}\implies \shift^{k}(\sigma(x))=\shift^{m}(\sigma(y))$, which is $\Ax$-provable by \eqref{eq:usvs}, is clearly not valid, contradicting the soundness of $\Ax$. Moreover, it is not possible that $\sigma(y)=u$ for a term $u$ using $\Firstsymbol$, since then we have a proof of $\ol{P}\implies \shift^k(\sigma(x))=\shift^m(\sigma(y))$ of height less than $h_{\min}$ and this quasiequation is of type \eqref{eq:critical-eq2}, a contradiction.  
\item\label{o3} $P$ contains no non-trivial identity $\shift^{*}(y)=\shift^{*}(y)$. This follows from O\ref{o2} and ($\star$).
\item\label{o4} The identity $\shift^{k}(x)=\shift^{m}(y)$ is the unique identity of type $\shift^{*}(x)=\shift^{*}(y)$  in~$P$. Indeed, suppose that there is another such identity $\shift^{l}(x)=\shift^{n}(y)$ where $(k,m)\neq (l,n)$. By the normalization condition \ref{n4} we know that $k=l$. Therefore $m\neq n$, and by \ref{n5} the identity $\shift^m(y)=\shift^n(y)$ lies in $P$, contradicting O\ref{o3}.
\end{enumerate}
In summary, we have shown that $P$ contains only one non-trivial $\Nextsymbol$-identity, namely $\shift^{k}(x)=\shift^{m}(y)$. Fix a countably infinite set 
\begin{equation}\label{eq:defA} A = \{ \bot , a_1, a_2, a_3,\ldots \} \end{equation}
and define the interpretation $\interpret\colon \{x,y\}\to A^\omega$ by
\begin{equation} \interpret(x)=\bot^{K+k}a_1a_2a_3\cdots \qqand \interpret(y)=\bot^{K+m}a_1a_2a_3\cdots  \end{equation}
where $K$ is the upper bound on the degrees of $\First\shift$-identities in $P$ given by \eqref{eq:defK}. Then all identities in $P$ are satisfied under $\interpret$: for the identity $\shift^k(x)=\shift^m(y)$ this is clear, and all $\First\shift$-identities in $P$ are satisfied because they refer to the first $K$ elements of $\interpret(x)$ and $\interpret(y)$, which have value $\bot$.

 We now look at the conclusion of the axiom \eqref{eq:proof-ax}, recalling that $s=\shift^i(x)$:

\mn\textbf{\underline{Subsubcase 6.1.1:}} \eqref{eq:proof-ax} is of the form $P\implies \shift^i(x)=\shift^p(x)$ where $p\geq 0$.

\mn Then \eqref{eq:minimal-qe} is the quasiequation $\ol{P}\implies \shift^i(\sigma(x))=\shift^p(\sigma(x))$, i.e.~$\ol{P}\implies \shift^{i+j}(a)=\shift^{p+j}(a)$ since $\sigma(x)=\shift^{j}(a)$. Thus it is a quasiequation of type \eqref{eq:critical-eq1}, and so $N$ does not divide $(i+j)-(p+j)=i-p$. On the other hand, since all identities in $P$ are satisfied under $\interpret$ and the axiom \eqref{eq:proof-ax} is valid, it follows that $\shift^i(\interpret(x)) = \shift^p(\interpret(x))$. But by definition of $\interpret(x)$, this is only possible if $i=p$. In particular $N$ divides $i-p$, a contradiction.

\mn\textbf{\underline{Subsubcase 6.1.2:}} \eqref{eq:proof-ax} is of the form $P\implies \shift^i(x)=\shift^p(y)$ where $p\geq 0$.

\mn Since all identities in $P$ are satisfied under $\interpret$ and the axiom \eqref{eq:proof-ax} is valid, it follows that $\shift^i(\interpret(x)) = \shift^p(\interpret(y))$. By definition of $\interpret(x)$ and $\interpret(y)$, this is only possible if
\[ k-m = i-p. \]
 Now consider the quasiequation \eqref{eq:minimal-qe}, i.e.~$\ol{P}\implies \shift^i(\sigma(x))=\shift^p(\sigma(y))$. We have $\sigma(x)=\shift^j(a)$ and $\sigma(y)=\shift^q(a)$ for some $q\geq 0$, the latter by O\ref{o2}. Thus \eqref{eq:minimal-qe} is of the form \eqref{eq:critical-eq1}, which means that $N$ does not divide $(i+j)-(p+q)$. Now by \eqref{eq:usvs} the quasiequation $\ol{P}\implies \shift^k(\sigma(x))=\shift^m(\sigma(y))$ has a proof from $\Ax$ of height less than $h_{\min}$. This quasiequation is equal to
\begin{equation}\label{eq:Pijpq} \ol{P}\implies \shift^{k+j}(a) = \shift^{m+q}(a), \end{equation}
and we have
\[ (k+j)-(m+q) = (k-m) + (j-q) = (i-p) + (j-q) = (i+j)-(p+q).  \]
Thus $N$ also does not divide $(k+j)-(m+q)$, which means that \eqref{eq:Pijpq} is a quasiequation of type \eqref{eq:critical-eq1} with a proof of height less than $h_{\min}$, a contradiction.

\mn\textbf{\underline{Subsubcase 6.1.3:}} \eqref{eq:proof-ax} is of the form $P\implies \shift^i(x)=t$ where $t$ contains $\Firstsymbol$.

\mn Since all identities in $P$ are satisfied under $\interpret$ and \eqref{eq:proof-ax} is valid, it follows that $\shift^i(\interpret(x))=\interpretsem{t}$. But this is impossible: the sequence $\interpretsem{t}$ is flat (since $t$ contains $\Firstsymbol$) but the sequence $\shift^i(\interpret(x))$ is not.

\mn\textbf{\underline{Subcase 6.2:}} $z=y$, and $P$ contains some identity $\shift^{k}(x)=\shift^{m}(y)$ where $k,m\geq 0$.

\mn Using arguments symmetric to Subcase 6.1, one shows that%
\begin{enumerate}[O1.]
\item\label{o1p} $P$ contains no non-trivial identity $\shift^{*}(y)=\shift^{*}(y)$.
\item\label{o2p} $\sigma(x)=\shift^{*}(a)$.
\item\label{o3p} $P$ contains no non-trivial identity $\shift^{*}(x)=\shift^{*}(x)$.
\item\label{o4p} The identity $\shift^{k}(x)=\shift^{m}(y)$ is the unique identity of type $\shift^{*}(x)=\shift^{*}(y)$ in $P$. 
\end{enumerate}%
Therefore $P$ contains only one non-trivial $\Nextsymbol$-identity. The rest of the argument is exactly like in Subcase 6.1.

\mn\textbf{\underline{Subcase 6.3:}} $z=x$, and $P$ contains no identity $\shift^{k}(x)=\shift^{m}(y)$ where $k,m\geq 0$.

\mn We first observe that%
\begin{enumerate}[O1.]
\item\label{o1pp} $P$ contains no non-trivial identity $\shift^{*}(x)=\shift^{*}(x)$. This follows from ($\star$) since $\sigma(x)=\shift^j(a)$.
\end{enumerate}
Take the set $A$ given by \eqref{eq:defA} and define the interpretation ${\ol{\interpret}}\colon \{x,y\}\to A^\omega$ by 
\begin{equation}\label{eq:defJ} 
{\ol{\interpret}}(x)=\bot^Ka_1a_2a_3\cdots  \qqand {\ol{\interpret}}(y)=\bot^\omega
\end{equation}
Then all identities in $P$ are satisfied under ${\ol{\interpret}}$: all $\First\shift$-identities are satisfied because they refer only to the first $K$ elements of ${\ol{\interpret}}(x)$ and ${\ol{\interpret}}(y)$ by \eqref{eq:defK}, and all non-trivial $\Nextsymbol$-identities are of the form $\shift^{*}(y)=\shift^{*}(y)$, and they are satisfied because ${\ol{\interpret}}(y)$ is a flat sequence.

We now consider the axiom \eqref{eq:proof-ax} and distinguish three cases:

\mn\textbf{\underline{Subsubcase 6.3.1:}} \eqref{eq:proof-ax} is of the form $P\implies \shift^i(x)=\shift^p(x)$ where $p\geq 0$.

\mn Then \eqref{eq:minimal-qe} is the quasiequation $\ol{P}\implies \shift^i(\sigma(x))=\shift^p(\sigma(x))$, i.e.~$\ol{P}\implies \shift^{i+j}(a)=\shift^{p+j}(a)$ since $\sigma(x)=\shift^{j}(a)$. Thus it is a quasiequation of type \eqref{eq:critical-eq1}, and so $N$ does not divide $(i+j)-(p+j)=i-p$. On the other hand, since \eqref{eq:proof-ax} is valid and all identities in $P$ are satisfied under the interpretation~${\ol{\interpret}}$, it follows that $\shift^i({\ol{\interpret}}(x)) = \shift^p({\ol{\interpret}}(x))$. But by definition of ${\ol{\interpret}}(x)$, this is only possible if $i=p$. In particular $N$ divides $i-p$, a contradiction.

\mn\textbf{\underline{Subsubcase 6.3.2:}} \eqref{eq:proof-ax} is of the form $P\implies \shift^i(x)=\shift^p(y)$ where $p\geq 0$.

\mn  Since \eqref{eq:proof-ax} is valid and all identities in $P$ are satisfied under the interpretation ${\ol{\interpret}}$, it follows that $\shift^i({\ol{\interpret}}(x))=\shift^p({\ol{\interpret}}(y))$. But this is impossible because $\shift^p({\ol{\interpret}}(y))$ is flat and $\shift^i({\ol{\interpret}}(x))$ is not.

\mn\textbf{\underline{Subsubcase 6.3.3:}} \eqref{eq:proof-ax} is of the form $P\implies \shift^i(x)=t$ where $t$ contains $\Firstsymbol$.

\mn Since \eqref{eq:proof-ax} is valid and all identities in $P$ are satisfied under the interpretation ${\ol{\interpret}}$, it follows that $\shift^i({\ol{\interpret}}(x))=\interpretsemol{t}$. But this is impossible: the sequence $\interpretsemol{t}$ is flat (since $t$ contains $\Firstsymbol$) but $\shift^i({\ol{\interpret}}(x))$ is not.

\mn\textbf{\underline{Subcase 6.4:}} $z=y$, and $P$ contains no identity $\shift^{k}(x)=\shift^{m}(y)$ where $k,m\geq 0$.

\mn Analogous to Subcase 6.3.

\mn This concludes the proof of \Cref{thm:incompleteness}.
\end{proof}


\section{Appendix: Infinitary Complete Axiomatization (\Cref{sec:infinitary-ax})}
\label{app:infinitary-ax}
Let us begin our approach towards the completeness proof for \(\AIC_\omega\).

\subsection{Completeness of Infinitary Quasiequational Logic}

As stated, it appears that the completeness of well-founded derivations for infinitary Horn logic is folklore. 
We include a proof here for convenience's sake, although we make no claim for originality.

\wellfoundedcompleteness*

\begin{proof}
    We are only going to cover the case where \(I\) is always a countable set, since this is the case relevant to \(\AIC_\omega\).
    However, our proof extends to the general case by replacing \(\omega_1\) with another regular cardinal that is an upper bound on the cardinalities of the index sets \(I\) that appear in quasiequations in \(\Ax \).

    (\(\Longrightarrow\)) Let us begin with \emph{soundness}.
    Suppose that \(\quasieq = (\bigwedge_{i \in I} \leftsuperscript{i}{s} = \leftsuperscript{i}{t}) \implies s = t\) 
is a quasiequation such that $\Ax\vdash_\wf \quasieq$, and that \(\Alg\) is an AIC-structure with \(\Alg\models \Ax \). We need to prove \(\Alg \models \quasieq\). Let $D$ be a well-founded derivation of $\quasieq$. We show more generally that for every node of $D$ with label $s'=t'$, we have $\Alg\models \quasieq'$ where \(\quasieq' = (\bigwedge_{i \in I} \leftsuperscript{i}{s} = \leftsuperscript{i}{t}) \implies s' = t'\). The proof is by well-founded induction.\footnote{Given a well-founded tree $T$, a property $P(n)$ holds for all nodes $n$ of $T$ whenever the following can be shown: for all nodes $n$, if $P(m)$ holds for all children $m$ of $n$, then $P(n)$ holds.} Thus, fix a node of $D$. If the node is a leaf labelled with one of the premises $\leftsuperscript{j}{s} = \leftsuperscript{j}{t}$ ($j\in I$), then $\quasieq'$ is the quasiequation $(\bigwedge_{i \in I} \leftsuperscript{i}{s} = \leftsuperscript{i}{t}) \implies  \leftsuperscript{j}{s} = \leftsuperscript{j}{t}$, which holds in every AIC-structure, so $\Alg\models \quasieq'$. Otherwise the node is a point of branching
    \begin{equation}\label{eq:branching}
        \inferrule{
                \inferrule{\phantom{D_l}~\vdots~D_l}{\leftsuperscript{l}{u} = \leftsuperscript{l}{v} }
                \quad\stackrel{\raisebox{.5em}{(for \(l,l' \in L\))}}{\cdots}\quad
                \inferrule{\phantom{D_{l'}}~\vdots~D_{l'}}{\leftsuperscript{l'}{u} = \leftsuperscript{l'}{v} }
            }{
                s' = t'
            }
    \end{equation}
and, by induction hypothesis, $\Alg$ satisfies all the quasiequations \((\bigwedge_{i \in I} \leftsuperscript{i}{s} = \leftsuperscript{i}{t}) \implies \leftsuperscript{l}{u} = \leftsuperscript{l}{v}\) where $l\in L$. Since \eqref{eq:branching} is a substitution instance of either some axiom of equational logic or of some axiom in $\Ax$, and $\Alg\models \Ax$, it follows that $\Alg\models \quasieq'$. 
\medskip

    \noindent%
    (\(\Longleftarrow\)) Let \(\quasieq = (\bigwedge_{i \in I} \leftsuperscript{i}{s} = \leftsuperscript{i}{t}) \implies s = t\) and suppose that \(\Alg \models \quasieq\) for every AIC-structure \(\Alg\) satisfying \(\Ax \).
    We construct an AIC-structure \(\Alg\) that satisfies the property that if \(\Alg \models \quasieq\), then \(\Ax  \vdash_\wf \quasieq\): let \({\approx} \subseteq \Terms^2\) be the relation defined such that \(u \approx v\) if and only if \(\Ax  \vdash_\wf (\bigwedge_{i \in I} \leftsuperscript{i}{s} = \leftsuperscript{i}{t}) \implies u = v\).
    Since \(\Ax \) includes equational logic (\cref{fig:equational logic}), \(\approx\) is a congruence. 
    This allows us to form the quotient homomorphism \([-]_\approx \colon \Terms \to \Alg\) where the AIC-structure \(\Alg\) is carried by \(A = \Terms/\approx\).

    Let us check that \(\Alg\models \Ax \).
    Let \(P = (\bigwedge_{j \in J} \leftsuperscript{j}{e} = \leftsuperscript{j}{f}) \implies e = f\) be an axiom in $\Ax$.
    Suppose that \(\sigma \colon \Vars \to A\)(\(= \Terms/{\approx}\)) is an interpretation that satisfies \(\leftsuperscript{\sigma}{\semantics{\leftsuperscript{j}{e}}} = \leftsuperscript{\sigma}{\semantics{\leftsuperscript{j}{f}}}\) for all \(j\).
    By definition, this means that \(\leftsuperscript{j}e[\sigma] \approx \leftsuperscript{j}f[\sigma]\), which is equivalent to there being a well-founded derivation \(D_j\) of \((\bigwedge_{i \in I} \leftsuperscript{i}{s} = \leftsuperscript{i}{t}) \implies \leftsuperscript{j}{e[\sigma]} = \leftsuperscript{j}{f[\sigma]}\), for every \(j \in J\).
    Since \( P \in \Ax \), we can then form the derivation 
    \begin{equation}
        \label{eq:P derivation}
        \inferrule{
            \inferrule{\phantom{~D_j}\vdots~D_j}
            {\leftsuperscript{j}{e[\sigma]} = \leftsuperscript{j}{f[\sigma]}}
            \quad\stackrel{\raisebox{.5em}{(for \(j,j' \in J\))}}{\cdots}\quad
            \inferrule{\phantom{~D_{j'}}\vdots~D_{j'}}
            {\leftsuperscript{{j'}}{e[\sigma]} = \leftsuperscript{{j'}}{f[\sigma]}}
        }{e[\sigma] = f[\sigma]}
    \end{equation}
    in \(\Ax \),
    where by our definition of a derivation of a quasiequation, each \(D_j\) is allowed to have \(\leftsuperscript{i}{s} = \leftsuperscript{i}{t}\) show up as a leaf for any \(i \in I\).
    This derivation is well-founded because \(D_j\) is well-founded for each \(j \in J\).%
    \footnote{
        We feel it is important to note that this is where well-foundedness plays a role. 
        If we instead altered our definition of derivation to require \emph{finite height} instead of well-foundedness, then this step no longer goes through: if the set of derivations \(D_j\) has no upper bound on their heights (say, \(J = \omega\) and the height of \(D_j\) is \(j\) like in \cref{eq:infinite but wf}), then we could no longer form the derivation in~\eqref{eq:P derivation}.
    }
    It follows that \(\Ax  \vdash_\wf (\bigwedge_{i \in I} \leftsuperscript{i}{s} = \leftsuperscript{i}{t}) \implies e[\sigma] = f[\sigma]\), and therefore by definition, \(e[\sigma] \approx f[\sigma]\) and thus \(\leftsuperscript{\sigma}{\semantics{e}} = \leftsuperscript{\sigma}{\semantics{f}}\).
    This proves that \(\Alg \models P\).
    Since \(P\) was an arbitrary axiom of \(\Ax \), we conclude that \(\Alg \models \Ax \).

    Now, by our assumption that \(\Alg \models \Ax \) implies \(\Alg \models \quasieq\), we infer that \(\Alg \models \quasieq\).
    This means that for any given interpretation \(\sigma \colon \Vars \to A\), if \(\leftsuperscript{\sigma}{{\semantics{\leftsuperscript{i}s}}} = \leftsuperscript{\sigma}{{\semantics{\leftsuperscript{i}t}}}\) for all \(i \in I\), then \(\leftsuperscript{\sigma}{\semantics{s}} = \leftsuperscript{\sigma}{\semantics{t}}\).
    We can therefore boldly choose \(
        \sigma \colon \Vars \to A
    \) to be the quotient map restricted to the variables, 
    \(\sigma(x) = [x]_{\approx}\), for which the hypothesis holds trivially: unravelling the definitions, for any \(j \in I\), 
    \(
    \leftsuperscript{\sigma}{\semantics{\leftsuperscript{j}{s}}} 
    = \leftsuperscript{\sigma}{\semantics{\leftsuperscript{j}{t}}}
    \) means that \(
    \leftsuperscript{j}{s} \approx \leftsuperscript{j}{t}
    \), which is the same as saying \(
    \Ax  
    \vdash_\wf (\bigwedge_{i \in I} \leftsuperscript{i}{s} = \leftsuperscript{i}{t}) \implies \leftsuperscript{j}{s} = 
    \leftsuperscript{j}{t}
    \).
    A derivation \(D\) of height $1$ exists for this quasiequation, because \(\leftsuperscript{j}{s} = \leftsuperscript{j}{t}\) is allowed to appear as a leaf---in particular, as both a leaf and the root.
    Hence, \(
    \leftsuperscript{\sigma}{\semantics{\leftsuperscript{i}{s}}} = \leftsuperscript{\sigma}{\semantics{\leftsuperscript{i}{t}}}
    \) for each \(i\), and from this we infer that \(
    \leftsuperscript{\sigma}{\semantics{s}} 
    = \leftsuperscript{\sigma}{\semantics{t}}
    \).
    In other words, \(s \approx t\), which is to say that \(\Ax  \vdash_\wf \quasieq\).
\end{proof}

\subsection{Completeness of \(\boldsymbol{\AIC_\omega}\)}

\newcommand{\latL}{L}
\newcommand{\latN}{N}

We are now ready to give the details behind the proof of \Cref{thm:embedding}.
We begin with the following result, which is stated without a proof in~\cite[Proposition 3.3]{TheunissenV07}.
Admittedly, we found it less ``completely straightforward'' than Theunissen and Venema, so we have included a proof for the sake of completeness.

\begin{lemma}[Monotone Lifting]
    \label{lem:complete extension}
    Let \(\latN\) be any lattice.
    There is a complete lattice \(\latL\) and an order-embedding\footnote{%
            A monotone map between posets \(f \colon X \to Y\) is an \emph{order-embedding} if for any \(x,x' \in X\), \(x \le x'\) if and only if \(f(x) \le f(x')\). Order-embedding maps are necessarily injective, and surjective order-embeddings are order isomorphisms.}  \(m \colon \latN \hookrightarrow \latL\) that preserves all meets and joins that exist in \(\latN\) such that for any monotone map \(F \colon \latN \to \latN\), there is a monotone \(\bar F \colon \latL \to \latL\) such that \(\bar F \circ m = m \circ F\).
\end{lemma}

\begin{proof}
    The lattice we have in mind is the \emph{Dedekind-MacNeille completion} of \(\latN\)~\cite{MacNeil36}, which can be described as follows.
    For each subset \(V \subseteq \latN\), define
    \[
        V^\Delta = \{x \in \latN \mid \forall v \in V, v \le x\}
        \qquad
        V^\nabla = \{x \in \latN \mid \forall v \in V, x \le v\}
    \]
    Then the operation \(C \colon 2^{|\latN|} \to 2^{|\latN|}\) defined \(CV = (V^\Delta)^\nabla\) is a closure operation. 
    This gives rise to the complete lattice \(\latL = \{V \in 2^{|\latN|} \mid V = CV\}\), ordered by set inclusion, and the embedding \(m \colon \latN \to \latL\) defined \(
        m(v) = C\{v\}
    \).
    It is known that the embedding \(m \colon \latN \to \latL\) preserves all meets and joins that exist in \(\latN\).
    It is worth noting that \((-)^\Delta\) and \((-)^\nabla\) are both antitone, and that this implies that \(C = (-)^\nabla \circ (-)^\Delta\) is monotone.

    Now let \(F \colon \latN \to \latN\) be a monotone map. 
    We extend \(F\) to \(\latL\) by defining \(\bar F \colon \latL \to \latL\) by \(\bar F(V) = C(F[V])\). 
    The operation \(\bar F\) is clearly monotone. 
    We are left with showing that \(\bar F \circ m = m \circ F\). 
    Observe that for any \(v \in \latN\), 
    \begin{gather*}
        m(v) 
        = C\{v\} 
        = (\{v\}^\Delta)^\nabla 
        = \{x \in \latN \mid \forall y \in \{v\}^\Delta, x \le y\}\\
        = \{x \in \latN \mid \forall y \in {\uparrow} v, x \le y\}
        = \{x \in \latN \mid x \le v\}
        = {\downarrow} v
    \end{gather*}
    The second to last equality can be seen by taking \(y = v\). 
    In particular, \(m \circ F(v) = C\{F(v)\} = {\downarrow} F(v)\) for any \(v \in \latN\).
    On the other side, we have \[
        \bar F \circ m(v) = CF[C\{v\}] = CF[{\downarrow} v] = C\{F(x) \mid x \le v\}
    \]
    Therefore, our goal is to show that \(C\{F(x) \mid x \le v\} = {\downarrow} F(v)\). 
    This will happen in two parts, one for each inclusion. 

    Given any \(x \le v\), by monotonicity we see that \(F(x) \le F(v)\). 
    It follows that \(\{F(x) \mid x \le v\} \subseteq {\downarrow} F(v)\). 
    Since \({\downarrow} F(v) \in \latL\), by monotonicity of \(C\) we have 
    \[
        C\{F(x) \mid x \le v\} \subseteq C({\downarrow}F(v)) = {\downarrow}F(v)
    \]
    This establishes one inclusion. 

    For the opposite inclusion, let \(y \le F(v)\).
    We want to show that \(y \in C\{F(x) \mid x \le v\} = (\{F(x) \mid x \le v\}^\Delta)^\nabla\), which is equivalent to the statement that \emph{for any \(u \in \latN\) such that \(F(x) \le u\) for all \(x \le v\), we have \(y \le u\)}.
    So, let \(u \in \latN\) such that \(F(x) \le u\) for all \(x \le v\).
    Then in particular, \(F(v) \le u\).
    But \(y \le F(v)\). 
    This tells us that \(y \le u\), so \(y \in C\{F(x) \mid x \le v\}\) as desired.
%
\end{proof}


\embeddingtheorem*

\begin{proof}
    Start by constructing the \emph{lattice of heads},
    \[
        \latL_0 = \left\{ \First^{\gray{\Alg_0}} a ~\middle|~ a \in A_0\right\}
    \]
    This is a lattice with lattice operations interpreted as 
    \begin{gather*}
        \bot^{\latL_0} = \First^{\gray{\Alg_0}} \bot^{\gray{\Alg_0}}
        \qquad
        \top^{\latL_0} = \First^{\gray{\Alg_0}} \top^{\gray{\Alg_0}}
        \\
        (\First^{\gray{\Alg_0}} a) \sqcap^{\latL_0} (\First^{\gray{\Alg_0}} b) = \First^{\gray{\Alg_0}}(a \lmin^{\gray{\Alg_0}} b)
        \qquad
        (\First^{\gray{\Alg_0}} a) \sqcup^{\latL_0} (\First^{\gray{\Alg_0}} b) = \First^{\gray{\Alg_0}}(a \lmax^{\gray{\Alg_0}} b)
    \end{gather*}
    because \(\First^{\gray{\Alg_0}}\) commutes with the lattice operations (see \Cref{tab:complete AIC}).
    This allows us to further define the lattice homomorphism \(e' \colon A_0 \to \latL_0^\omega\) given by 
    \[
        e'(a)(n) = \First^{\gray{\Alg_0}} (\Next^{\gray{\Alg_0}})^n a
    \]
    where \((\Next^{\gray{\Alg_0}})^n\) is the \(n\)-fold application of the shift operation. That $e'$ is a lattice homomorphisms folllows from the Head and Shift axioms, which ensure the the $\First$ and $\shift$ operations preserve commute with the lattice structure. It follows directly from the first sequence axiom in \Cref{tab:complete AIC} that \(e'\) is injective.  
    Note that the function operations \(F^{\gray{\Alg_0}}\) for \(F \in \Funcs\) restrict to \(\latL_0\), since \(F^{\gray{\Alg_0}}\First a = \First F^{\gray{\Alg_0}} a\).
    
    By \Cref{lem:complete extension}, there is a complete lattice \(\latL\) and an order-embedding \(m \colon \latL_0 \hookrightarrow \latL\) such that every monotone operation \(G \colon \latL_0 \to \latL_0\) extends to a monotone operation \(\bar G \colon \latL \to \latL\).
    Write \(m^\omega \colon \latL_0^\omega \to \latL^\omega\) for the pointwise embedding defined by \(m^\omega(\sigma)(n) = m(\sigma(n))\) for any \(\sigma\in \latL_0^\omega\) and \(n \in \omega\).
    We obtain a sequence algebra \(\Alg\) carried by \(\latL^\omega\) by lifting each operation \(F^{\gray{\Alg_0}} \colon \latL_0 \to \latL_0\) to a monotone operation \(\bar F\) on \(\latL\) and applying it pointwise to sequences to obtain the unary operation \(F^{\gray{\Alg}} \colon \latL^\omega \to \latL^\omega\).

    From here on, most superscripts indicating the type of the operations will be supressed, since they can be inferred from context.
    
    Finally, let \(\varepsilon = m^\omega \circ e'\).
    We obtain the following diagram:
    \[
        \begin{tikzcd}
            \Alg_0 \ar[rr, "\varepsilon"] \ar[rd, "e'"] 
            && \latL^\omega \\
            & \latL_0^\omega \ar[ru, "m^\omega"] &
        \end{tikzcd}
        \qquad 
        \varepsilon(a)(n) = m^\omega \circ e'(a)(n) = m(\First \Next^n a)
    \]
    It follows from the injectivity of \(m^\omega\) and \(e'\) that \(\varepsilon\) is injective. 
    It now suffices to see that \(\varepsilon\) is a homomorphism of AIC-structures. 
    \begin{description}
        \item[Lattice] Since \(m\) is a lattice homomorphism, so is the pointwise application $m^\omega$. Therefore $\varepsilon$ is a lattice homomorphism, being the composite of the lattice homomorphisms $e'$ and $m^\omega$.

\medskip


        \item[Head] Given \(a \in A_0\), we have
        \begin{align*}
            \varepsilon(\First a)(n) 
            &= m(\First \Next^n \First a) \\
            &= m(\First \First a)           \tag{Head} \\
            &= m(\First a)                  \tag{Head} \\
            &= \varepsilon(a)(0)                      \tag{definition of \(\varepsilon\)} \\
            &= (\First \varepsilon(a))(n)               \tag{sequence algebra}
        \end{align*}

        \item[Shift] Given \(a \in A_0\), we have
        \begin{align*}
            \varepsilon(\Next a)(n) 
            &= m(\First \Next^n \Next a) \\
            &= m(\First \Next^{n+1} a)     \\
            &= \varepsilon(a)(n + 1)                  \tag{definition of \(\varepsilon\)} \\
            &= (\Next \varepsilon(a))(n)               \tag{sequence algebra}
        \end{align*}

        \item[Functions] Given \(a \in A_0\) and \(F \in \Funcs\), we have
        \begin{align*}
            \varepsilon(F a)(n) 
            &= m(\First \Next^n F a) \\
            &= m(\First F \Next^{n} a)                 \tag{Shift} \\
            &= m(F \First \Next^{n} a)                 \tag{Head} \\
            &= \bar F ( m(\First \Next^{n} a))         \tag{\(\bar F\) extends \(F\)} \\
            &= \bar F(\varepsilon(a)(n))                          \tag{definition of \(\varepsilon\)} \\
            &= (F^{\gray{\Alg}}\varepsilon(a))(n)                          \tag{sequence algebra} 
        \end{align*}

        \item[Iteration] Given \(a \in A_0\) and \(F \in \Funcs\), we have
        \begin{align*}
            \varepsilon(F^* a)(n) 
            &= m(\First \Next^n F^* a) \\
            &= m(\First F^n F^* \Next^{n} a)           \tag{Shift} \\
            &= m(F^n \First F^* \Next^{n} a)           \tag{Head} \\
            &= m(F^n \First \Next^{n} a)               \tag{Head} \\
            &= \bar F^n \circ m(\First \Next^{n} a)    \tag{\(\bar F\) lifts \(F^{\gray{\Alg_0}}\)} \\
            &= \bar F^n(\varepsilon(a)(n))                        \tag{definition of \(\varepsilon\)}\\            
            &= ((F^*)^{\gray{\Alg}}\varepsilon(a))(n)                        \tag{sequence algebra}    
        \end{align*}

        \item[Majorum/Minorum] We will just do the majorum case (tail-suprema), because the minorum case is symmetric.
        We begin with the following observation: for any \(k \in \omega\) and \(c \in A_0\), by the Maj/Minorum and Shift axioms, 
        \(
            \Next^k \Sup c \sqsubseteq \Sup c
        \).
        We therefore have \(\First \Next^k \Sup c \sqsubseteq \First \Sup c\) for all \(k \in \omega\).
        The Sequence axioms now tell us that \(\First \Sup c\) is the \emph{least} upper bound of \(\{\First \Next^k c\mid k \in \omega\}\), and that this least upper bound exists in \(\latL_0\). 
        Now, if we let \(c = \Next^{n} a\), we see that 
        \begin{align*}
            \First \Sup \Next^{n} a 
            &= \bigsqcup_{k \ge 0} \First \Next^{k} \Next^{n} a \\
            &= \bigsqcup_{k \ge 0} \First \Next^{n + k} a \\
            &= \bigsqcup_{k \ge n} \First \Next^{k} a 
        \end{align*}
        Since \(m\) preserves all suprema that exist in \(\latL_0\), we can derive
        \begin{align*}
            \varepsilon(\Sup a)(n) 
            &= m(\First \Next^n \Sup a)                          \\
            &= m(\First \Sup \Next^n a)                          \tag{Shift}\\
            &= m\Big(\bigsqcup_{k \ge n} \First \Next^{k} a\Big)   \tag{\(\First \Sup a\) a least upper bound}\\
            &= \bigsqcup_{k \ge n} m(\First \Next^{k} a)           \tag{\(m\) preserves joins} \\
            &= \bigsqcup_{k \ge n} \varepsilon(a)(k)                          \tag{definition of \(\varepsilon\)}\\            
            &= (\Sup \varepsilon(a))(k)                                     \tag{sequence algebra}
        \end{align*}
    \end{description}
    Hence, \(\varepsilon\) is an algebra homomorphism, completing the proof of the embedding theorem.
\end{proof}

\end{document}